\documentclass[smallextended]{svjour3}  
\usepackage{color}
\usepackage{url}

\usepackage[usenames,dvipsnames]{xcolor}
\usepackage{multirow}
\usepackage{graphics}
\usepackage{a4wide}
\usepackage{pdfpages}
\usepackage{tikz}

\usepackage{todonotes}

\usepackage{caption}
\usepackage{graphicx}
\usepackage{booktabs}
\usepackage{enumitem}
\usepackage{soul}

\newcommand{\TBR}{{\rm TBR}}
\newcommand{\MP}{{\rm MP}}

\newcommand{\MAF}{{\rm MAF}}

\newcommand{\steven}[1]{\textcolor{black}{#1}}

\newcommand{\revisionSK}[1]{\textcolor{black}{#1}}

\newcommand{\purple}{\textcolor{black}}

\newcommand{\blue}{\textcolor{black}}

\usepackage{amsmath}
\usepackage{graphicx}

\title{Reflections on kernelizing and computing unrooted agreement forests}
\author{Rim van Wersch \and Steven Kelk \and Simone Linz \and Georgios Stamoulis}

\institute{R. van Wersch \at
              Department of Data Science and Knowledge Engineering (DKE), Maastricht University, The Netherlands,\\ \email{rim@hyperboreanventures.com}          
           \and
           S. Kelk \at 
           Department of Data Science and Knowledge Engineering (DKE), Maastricht University, The Netherlands,\\ \email{steven.kelk@maastrichtuniversity.nl}
           \and
           S. Linz \at
              School of Computer Science, University of Auckland, New Zealand,\\ \email{s.linz@auckland.ac.nz}
           \and 
           G. Stamoulis \at
            Department of Data Science and Knowledge Engineering (DKE), Maastricht University, The Netherlands,\\ \email{georgios.stamoulis@maastrichtuniversity.nl}  
}

\providecommand{\keywords}[1]{\textit{Keywords:} #1}

\begin{document}
\maketitle

\begin{abstract}
Phylogenetic trees are leaf-labelled trees used to model the evolution of species. Here we explore the practical impact of kernelization (i.e. data reduction) on the NP-hard problem of computing the TBR distance between two unrooted binary phylogenetic trees. This problem is better-known in the literature as the maximum agreement forest problem, where the goal is to partition the two trees into a \steven{minimum number of} common, non-overlapping subtrees. We have implemented two well-known reduction rules, the subtree and chain reduction, and five more recent, theoretically stronger reduction rules, and compare the reduction achieved with and without the stronger rules. We find that the new rules yield smaller reduced instances and thus have clear \steven{practical} added value. In many cases they also cause the TBR distance to decrease in a \steven{controlled} fashion, which can further facilitate solving the problem in practice. Next, we compare the achieved reduction to the known worst-case theoretical bounds of $15k-9$ and $11k-9$ respectively, on the number of leaves of the two reduced trees, where $k$ is the TBR distance, observing in both cases a far larger reduction in practice. As a \steven{by-product of} our experimental framework we obtain a number of new insights into the actual computation of TBR distance. We find, for example, that very strong lower bounds on TBR distance can be obtained efficiently by randomly sampling certain carefully constructed partitions of the leaf labels, and identify instances which seem particularly challenging to solve exactly. The reduction rules have been implemented within our new solver \emph{Tubro} which combines kernelization with an Integer Linear Programming (ILP) approach. Tubro also incorporates a number of additional features, such as a cluster reduction and a practical upper-bounding heuristic, and it can leverage combinatorial insights emerging from the proofs of correctness of the reduction rules to simplify the ILP.
%\marginpar{SL. Commented out last sentence of abstract.}
%We conclude with a number of proposals for future research in this area.

% Interestingly, a number of combinatorial insights underpinning the five recent reduction rules can be used to prune the feasible region of the ILP formulation. We observe that kernelization combines well with these further optimizations, yielding a further reduction in solving time.
\end{abstract}


\section{Introduction}
%\emph{Abstract and intro will need tuning depending on where we send the article!}
The central challenge of phylogenetics, which is the study of \purple{phylogenetic (evolutionary)} trees, is to infer a \purple{tree} whose leaves are bijectively \purple{labeled} by a set \blue{$X$} of species and which accurately represents the evolutionary events \purple{that} gave rise to $X$. There are many methods to construct phylogenetic trees; we refer to a standard text such as \cite{felsenstein2004inferring} for an overview. The complexity of this problem stems from the fact that we typically only have indirect data available, such as DNA sequences of \blue{a set $X$ of} species. It is quite common for different methods to construct different trees, or the same method to construct different trees depending on which part of a genome the DNA data is extracted from \cite{HusonRuppScornavacca10,richards2018variation,yoshida2019multilocus}. In such cases we might wish to formally quantify the \emph{dissimilarity} of these trees, and this has fuelled research into different distance measures on pairs of phylogenetic trees \cite{Kuhner2015}. Such distances also help us to understand (i) the space of all phylogenetic trees on a fixed number of leaves
%connectivity of tree-space
\cite{john2017shape}, (ii) the construction of phylogenetic supertrees \cite{Whidden2014}, which produces a single phylogenetic tree summarizing multiple trees; and (iii) phylogenetic networks \cite{HusonRuppScornavacca10}, which generalize phylogenetic trees to graphs.

In this article we focus on the following problem: given two unrooted binary phylogenetic trees $T$ and $T'$, both on leaf set $X$, what is the minimum number of Tree Bisection and Reconnection (TBR) moves required to turn $T$ into $T'$? We defer formal definitions to Section~\ref{sec:def}, but in essence a TBR move consists of deleting an edge of a tree \purple{and} then introducing a new edge to reconnect the two \purple{resulting} components back together. This distance measure, denoted $d_{\TBR}$, is NP-hard to compute \cite{AllenSteel2001,hein1996complexity}. Computation of $d_{\TBR}$ is essentially equivalent to the well-known problem of computing a \emph{maximum agreement forest} of $T$ and $T'$. \purple{Such a forest is} a partition of $X$ into \purple{the} minimum number of blocks \purple{so} that each block induces a subtree in $T$ that is isomorphic (up to subdivision) to that induced in  $T'$, and \blue{no two} blocks induce \blue{overlapping subtrees} in $T$ and $T'$.
%the induced subtrees in $T$ and $T'$ are non-overlapping and (up to subdivision) isomorphic.
This minimum number, $d_{\MAF}$, is equal to $d_{\TBR} + 1$ \cite{AllenSteel2001}. Due to this equivalence, everything in this article that applies to $d_{\TBR}$ also applies to the computation of $d_{\MAF}$.

In the last few years there has been quite some research into the \emph{fixed parameter tractability} of computing $d_{\TBR}$. Fixed parameter tractable  (FPT) algorithms are those with running time $O( h(k) \cdot \text{poly}(n) )$ where $n$ is the size of the input and $h$ is a computable function that depends only on some well-chosen parameter $k$ \cite{downey2013fundamentals,Cygan:2015:PA:2815661}. Such algorithms have the potential to run quickly when $k$ is small, even if $n$ is large. In any case, they run more quickly than algorithms with running times of the form $O( n^{h(k)} )$. A common technique for developing \purple{FPT} algorithms is to navigate a computational search tree such that the number of \blue{vertices} in the tree is bounded by a function of $k$. Using this branching technique, running times of $O( 4^k \cdot \text{poly}(|X|) )$ \cite{whidden2013fixed} and
later $O( 3^k \cdot \text{poly}(|X|) )$~\cite{chen2015parameterized} have been obtained for the computation of $d_{\TBR}$, where $k = d_{\TBR}$. \steven{Another main} technique is \emph{kernelization} \cite{kernelization2019}, which is the focus of this article. The goal here is to apply polynomial-time pre-processing reduction rules such that the reduced instance - the \emph{kernel} - has size bounded by a function of $k$. Such a kernel can then be solved by an exact exponential-time algorithm, which ultimately also yields a running time of the form  $O( h(k) \cdot \text{poly}(n) )$. For the computation of $d_\TBR$, a kernel of size $28k$ was established in 2001; recently this \purple{was} reduced to $15k-9$, and later $11k-9$, by two of the present authors \cite{tightkernel,kelk2020new}. The $28k$ and $15k-9$ results are based on two well-known reduction rules, the \emph{subtree} and \emph{chain} reduction rules, while the $11k-9$ result was obtained by adding five new reduction rules to this repertoire. \blue{Roughly speaking, the latter five reduction rules target short chains that cannot be (further) reduced by the chain reduction rule and, so, the $11k-9$ result is based on repeated applications of all seven reduction rules.} The $15k-9$ and $11k-9$ bounds are both tight under their respective sets of reduction rules.

In this article we adopt an experimental approach to answering the following question: \purple{
%if at all,
do}
%\marginpar{We want the question to be a yes/no question because of the answer we give.}
the new reduction rules from \cite{kelk2020new} produce smaller kernels \emph{in practice} than, say, when only the subtree and chain reductions are applied? This mirrors several recent articles in the algorithmic graph theory literature where the practical effectiveness of kernelization has also been analyzed~ \cite{fellows2018known,ferizovic2020engineering,henzinger2020shared,mertzios2020power,alber2006experiments}.
The question is relevant, since earlier studies of kernelization in phylogenetics have noted that, despite its theoretical importance, in an empirical setting the chain reduction seems to have very limited effect compared to the subtree reduction \blue{\cite{hickey2008spr,vanIersel20161075}}. In this sense, it is natural to ask whether the five new reduction rules have any real added value in practice. Happily, our experiments indicate that the answer is \emph{yes}; \steven{across a range of experimental settings, the average percentage of \blue{species contained in $X$} that survive after application of the entire ensemble of reduction rules, is \blue{substantially} smaller than when only the subtree and chain reductions are applied}.  We also explore whether the reduction achieved in practice is, expressed as a function of $k$, significantly better than the theoretical worst-case bound of $11k-9$. Again, we answer this affirmatively. \purple{Our} experiments show that kernels of size $6k$ or lower are obtained for 93\% of the tree pairs that we used in our experimental study. We also derive insights into how frequently the phenomenon of \emph{parameter reduction} occurs: this is when a reduction rule \purple{is triggered} that does not preserve $d_{\TBR}$ but instead reduces it in some predictable and well-understood fashion. Parameter reduction is particularly helpful if the kernel is subsequently solved by an FPT algorithm whose running time is exponentially dependent on $d_{\TBR}$.

\begin{figure}[t]
\center
\scalebox{1}{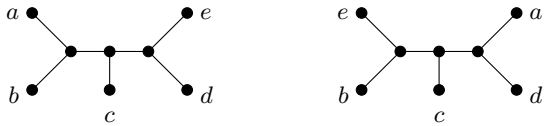}
\caption{Two unrooted binary phylogenetic trees on $X=\{a,b,c,d,e\}$.}
\label{fig:trees}
\end{figure}

Our experimental framework also yields a number of new insights concerning the actual \emph{computation} of $d_{\TBR}$  as opposed to only reducing the problem in size. To compute the empirical kernel size
%$f(k)$ bound,  %%% SL. $f(k)$ comes a bit out of the blue.
mentioned in the previous paragraph, it is necessary to know $d_{\TBR}$, or a good lower bound on $d_{\TBR}$. To that end our experiments deploy the only existing, publicly-available TBR solver uSPR~\cite{whidden2018calculating} to compute $d_{\TBR}$, where possible, which is not always the case. In this way we build up a picture of which instances are particularly challenging to compute in practice. For instances where computation of $d_{\TBR}$ is too challenging \purple{for uSPR}, we use a novel lower bound on $d_\TBR$ based on randomly sampling certain types of so-called \emph{convex characters}, which are specially constructed partitions of $X$ \cite{SempleSteel2003}. This turns out to be \emph{extremely} effective, often yielding lower bounds within a few seconds that are very close to $d_{\TBR}$. Arguably, alongside our study of kernel size, the discovery of this bound is the most important contribution of this paper.

A final contribution of this article is our new solver \emph{Tubro}. This was initially intended simply as a vehicle for the new reduction rules from \cite{kelk2020new}. However, during development we reasoned that it would be useful for \emph{Tubro} to incorporate its own solver. To this end, we give a new Integer Linear Programming (ILP) formulation for computation of $d_{\TBR}$, and prove its correctness. Tubro feeds this ILP to the powerful solver Gurobi \cite{gurobi}, strengthening the formulation %\marginpar{Unclear to me what `this' refers to.}
with heuristically computed upper bounds and the aforementioned lower bound. %Although it is not the purpose of the article to compare uSPR to Tubro
We note that Tubro is capable of solving a number of comparatively small instances that are out-of-range for uSPR in our experiments.
% SK: I removed this -- we can put it back if we do actually do this :-)
%, and speculate on why this is.
Tubro \purple{additionally} incorporates a number of extra features that are possibly of independent interest, such as a \emph{cluster} reduction \cite{bordewich2017fixed,li2017computing}. It is also able to significantly reduce the number of variables in the ILP by specifying that even in fully kernelized instances certain chains should never be cut, leveraging an insight used in the proofs of correctness in \cite{kelk2020new}. An executable version of Tubro is available upon request.

The paper is organized as follows. The next section contains some preliminaries that are used throughout the paper and details the seven reduction rules that are underlying the kernelization experiments presented \purple{here}. In Section~\ref{subsec:def_dmpbound}, we introduce the maximum parsimony distance between two unrooted binary phylogenetic trees and the associated sampling strategy to obtain a lower bound on the TBR distance. The following three sections describe our experiments. We first present the experimental setup in Section~\ref{sec:exp-framework} and, subsequently, summarize \steven{and analyse} the results in Section~\ref{sec:results}. A \steven{high-level} discussion follows in Section~\ref{sec:toughcookies}, \steven{where we
reflect on what our experiments tell us about which TBR instances are easy, and difficult, to solve, and why this might be.} In our conclusion (Section~\ref{sec:discussion}) we make a number of suggestions for further research.
% in computation of $d_{\TBR}$.
Technical details of Tubro are deferred to the appendix, which establishes correctness of the ILP that Tubro is based on and gives details of the various options that can be switched on or off when running Tubro, such as chain preservation and cluster reduction.

\section{Preliminaries}
\subsection{Definitions} \label{sec:def}

Our notation closely follows \cite{kelk2020new}. Throughout this paper, \blue{$X$ denotes a finite set of \emph{taxa}.}
%with $|X| \geq 4$. 
An {\it unrooted binary phylogenetic tree} $T$ on $X$  is a  simple, connected, and undirected tree whose leaves are bijectively labeled with $X$ and whose other vertices all have degree 3. See Figure \ref{fig:trees} \purple{for an example of two unrooted binary phylogenetic trees on $X=\{a,b,c,d,e\}$}. For simplicity and since most phylogenetic trees in this paper are unrooted and binary, we refer to an unrooted binary phylogenetic trees as a {\it phylogenetic tree}. If a definition or statement applies to all unrooted phylogenetic trees, regardless of whether they are binary or not, we make this explicit. Two leaves, say $a$ and $b$, of $T$ are called a {\it cherry} $\{a,b\}$ of $T$ if they are adjacent to a common vertex.
%We say that a vertex $v$ is the (unique) {\it parent} of a leaf $a$ in $N$ if $v$ is adjacent to  $a$.
For $X' \purple{\subseteq} X$, we write $T[X']$ to denote the unique, minimal subtree of $T$ that connects all elements in $X'$. For brevity we call $T[X']$ the \emph{embedding} of  $X'$ in $T$.
%If $X''$ is also a subset of $X$, we denote by $T[X']\cap T[X'']$ the set of \textcolor{black}{vertices} in $T$ that are contained in $T[X']$ and $T[X'']$.
%\marginpar{Commented out a sentence here.}
Furthermore, we refer to the phylogenetic tree on $X'$ obtained from $T[X']$ by suppressing non-root degree-2 vertices as  the {\it restriction of $T$ to $X'$}  and we denote this by $T|X'$.\\
\\
\noindent{\bf Tree bisection and reconnection.} Let $T$ be a phylogenetic tree on $X$. Apply the following three-step operation to $T$:
\begin{enumerate}
\item Delete an edge in $T$ and suppress any resulting degree-2 vertex. Let $T_1$ and $T_2$ be the two resulting phylogenetic trees.
\item If $T_1$ (resp. $T_2$) has at least one edge, subdivide an edge in $T_1$ (resp. $T_2$) with a new vertex $v_1$ (resp. $v_2$) and otherwise set $v_1$ (resp. $v_2$) to be the single isolated vertex of $T_1$ (resp. $T_2$).
\item Add a new edge $\{v_1,v_2\}$ to obtain a new phylogenetic tree $T'$ on $X$.
% it was (v1,v2), I changed it to {v1, v2}
\end{enumerate}
We say that $T'$ has been obtained from $T$ by a single  \purple{{\it tree bisection and reconnection (TBR) operation} (or, {\it TBR move})}. %\blueblue{An example of a TBR operation is illustrated in Figure~\ref{fig:tbr}.}
Furthermore, we define the TBR {\it distance}  between two phylogenetic trees $T$ and $T'$ on $X$,  denoted by $d_\TBR(T,T')$, to be the minimum number of TBR operations that \purple{are} required to transform $T$ into $T'$.
\steven{\purple{To illustrate,} the trees $T$ and $T'$ in Figure \ref{fig:tbr} have a TBR distance of 1.}
It is well known that $d_\TBR$ is a metric~\cite{AllenSteel2001}. By building on an earlier result by Hein et al.~\cite[Theorem 8]{hein1996complexity}, Allen and Steel~\cite{AllenSteel2001} showed that computing the TBR distance is an NP-hard problem. \\

\begin{figure}[t]
\center
\scalebox{1}{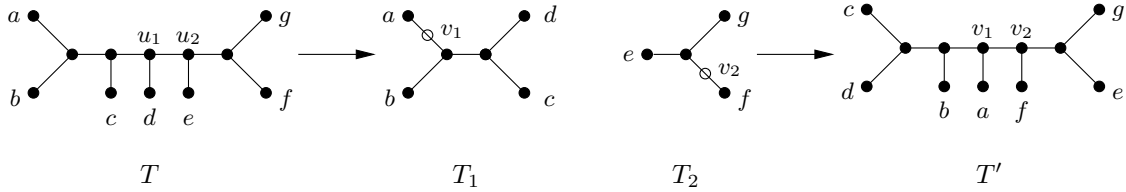}
\caption{A single TBR operation that transforms $T$ into $T'$. First, $T_1$ and $T_2$ are obtained from $T$ by deleting the edge $\{u_1,u_2\}$ in $T$. Second, $T'$ is obtained from $T_1$ and $T_2$ by subdividing an edge in both trees as indicated by the open circles $v_1$ and $v_2$ and adding a new edge $\{v_1,v_2\}$.}
\label{fig:tbr}
\end{figure}

\noindent {\bf Agreement forests.}
Let $T$ and $T'$ be two phylogenetic trees  on $X$. Furthermore, let $F = \{B_0, B_1,B_2,\ldots,B_k\}$ be a partition of $X$, where each block $B_i$ with $i\in\{0,1,2,\ldots,k\}$ is  referred to as a \emph{component} of $F$. We say that $F$ is an \emph{agreement forest} for $T$ and $T'$ if the following conditions hold.
\begin{enumerate}
\item [(1)] For each $i\in\{0,1,2,\ldots,k\}$, we have $T|B_i = T'|B_i$.
\item [(2)] For each pair $i,j\in\{0,1,2,\ldots,k\}$ with $i \neq j$, we have that
$T[B_i]$ and $T[B_j]$ are vertex-disjoint in $T$, and $T'[B_i]$ and $T'[B_j]$ are vertex-disjoint in $T'$.
\end{enumerate}
\noindent
Let $F=\{B_0,B_1,B_2,\ldots,B_k\}$ be an agreement forest for $T$ and $T'$. The \emph{size} of $F$ is simply its number of components; i.e. $k+1$. Moreover, an agreement forest with the minimum number of components (over all agreement forests for $T$ and $T'$) is called a \emph{maximum agreement forest \purple{(MAF)}} for $T$ and $T'$. The number of components of a maximum agreement forest for $T$ and $T'$ is denoted by $d_\MAF(T,T')$. The following theorem is well known.

\begin{theorem}\cite[Theorem 2.13]{AllenSteel2001}
Let $T$ and $T'$ be two phylogenetic trees  on $X$. Then $$d_\TBR(T,T') = d_\MAF(T,T') - 1.$$
\end{theorem}

\steven{A maximum agreement forest for the trees $T$ and $T'$ shown in Figure \ref{fig:tbr}, which have TBR distance 1, therefore contains two components. $F=\{\{a,b,c,d\},\{e,f,g\}\}$ is an example of such a forest (in fact, here it is the only maximum agreement forest).}

We conclude this section with a number of algorithmic definitions. A \emph{parameterized problem} is a problem for which the inputs are of the form $(x,k)$, where $k$ is a non-negative integer, called the \emph{parameter}. A parameterized problem is \emph{fixed-parameter tractable} (FPT) if there exists an algorithm that solves\footnote{Note that the formalism described here actually concerns \emph{decision} (i.e. yes/no) problems, which in the context of the current article is most naturally ``Is $d_{\TBR} \leq k$?''. An FPT algorithm for answering this question can easily be transformed into an algorithm for computing $d_{\TBR}$ with similar asymptotic time complexity by increasing $k$ incrementally from 0 until a yes-answer is obtained.} any instance $(x,k)$ in $h(k)\cdot |x|^{O(1)}$ time, where \revisionSK{$h(\cdot)$} is a computable function depending only on $k$. A parameterized problem has a \emph{kernel} of size \revisionSK{$g(k)$}, where $g(\cdot)$ is a computable function depending only on $k$, if there exists a polynomial time algorithm transforming any instance $(x,k)$ into an equivalent \blue{instance} $(x',k')$, with $|x'|,k' \leq g(k)$. \blue{Here, {\it equivalent} means that $(x,k)$ is a yes-instance if and only if $(x',k')$ is a yes-instance.} If $g(k)$ is a polynomial in $k$ then we call this a \emph{polynomial kernel}; if $g(k) = O(k)$ then it is a \emph{linear kernel}. \blue{A well-known theorem from parameterized complexity states} that a parameterized problem is fixed-parameter tractable if and only if it has a (not necessarily polynomial) kernel. For more background information on fixed parameter tractability and kernelization, we refer the reader to standard texts such as~\cite{Cygan:2015:PA:2815661,downey2013fundamentals,kernelization2019}.

In this article, we take $d_{\TBR}$ as the parameter $k$ and take $|X|$, the number of leaves, as the size of the instance $|x|$. The reduction rules described in the following section produce a linear kernel and run in $\text{poly}(|X|)$ time.

%An algorithm is said to be \emph{fixed parameter tractable} (FPT) with respect to a parameter $k$ if it has running time $O( f(k) \cdot \text{poly}(n) )$ where $n$ is the size of
%the input and $f$ is a computable function that depends ony on $k$. We refer to a standard text such as \cite{Cygan:2015:PA:2815661} for more background on FPT algorithms. In this article we take $d_{\TBR}$ as the parameter
%$k$ and $n=|X|$.  The concept of \emph{kernelization} is closely related. We refer

\subsection{Description of the subtree, chain and other reduction rules}
\label{subsc:def_reducrules}

In this section we describe the existing reduction rules - seven in total - that will be analyzed in this article.

Let $T$ and $T'$ be two phylogenetic trees on $X$. We say that a subtree of $T$ is {\it pendant} if it can be detached from $T$ by deleting a single edge. \blue{For reasons of simplicity, we assume for the remainder of this paragraph that $|X|\geq 4$. Note that no generality is lost because $d_\TBR(T,T')=0$ if $|X|<4$.} For $n\geq 2$, let $C = (\ell_1,\ell_2\ldots,\ell_n)$ be a sequence of distinct taxa in $X$. For each $i\in\{1,2\ldots,n\}$, let $p_i$ denote the unique parent of $\ell_i$ in $T$. We call $C$ an $n$-chain of $T$ if there exists a walk $p_1,p_2,\ldots,p_n$ in $T$ and the elements \blue{$p_2,p_3,\ldots,p_{n-1}$} are all pairwise distinct. 
%walk $p_1,p_2,\ldots,p_n$ in $T$ and the \blue{vertices $p_2,p_3,\ldots,p_{n-1}$ are} pairwise distinct. 
Note that $\ell_1$ and $\ell_2$ may have a common parent or $\ell_{n-1}$ and $\ell_n$ may have a common parent. Furthermore, if  $p_1 = p_2$ or $p_{n-1} = p_n$ holds, then $C$ is pendant in $T$. To ease reading, we sometimes write $C$ to denote the set $\{\ell_1,\ell_2,\ldots,\ell_n\}$. It will always be clear from the context whether $C$ refers to the associated sequence or set of taxa. If a pendant subtree $S$ (resp. an $n$-chain $C$) exists in $T$ and $T'$, we say that $S$ (resp. $C$) is a {\it common} subtree (resp. chain) of $T$ and $T'$.
%Moreover, if $C$ is a common $n$-chain of $T$ and $T'$, reducing $C$ to a chain of length $k$ with $1\leq k <n$ yields the two new trees $T_r = T|X\setminus\{\ell_{k+1},\ell_{k+2},\ldots,\ell_n\}$ and $T_r' = T'|X\setminus\{\ell_{k+1},\ell_{k+2},\ldots,\ell_n\}$.\\

We are now in a position to state all seven reduction rules. Let $T$ and $T'$ be two phylogenetic trees on \purple{$X$}. \blue{We already note here that  an application of Reduction 1, 2, 6, or 7 results in two new trees whose TBR distance is the same as that of $T$ and $T'$, and an application of Reduction 3, 4, or 5, results in two new trees whose TBR distance is one less than that of $T$ and $T'$.}\\
%with $|X|\geq 4$.\\

\noindent {\bf Reduction 1.}~\cite{AllenSteel2001} If $T$ and $T'$ have a maximal common pendant subtree $S$ with at least two leaves, then reduce $T$ and $T'$ to $T_r$ and $T'_r$, respectively, by replacing $S$ with a single leaf with a new label.

\noindent {\bf Reduction 2.}~\cite{AllenSteel2001} If $T$ and $T'$ have a maximal common $n$-chain $C=(\ell_1,\ell_2,\ldots,\ell_n)$ with $n\geq 4$, then reduce $T$ and $T'$ to $T_r=T|X\setminus \{\ell_4,\ell_5,\ldots,\ell_n\}$ and $T_r'=T'|X\setminus \{\ell_4,\ell_5,\ldots,\ell_n\}$, respectively.

\noindent {\bf Reduction 3.}~\cite{kelk2020new} If $T$ and $T'$ have a common 3-chain $C=(\ell_1,\ell_2,\ell_3)$ such that $\{\ell_1,\ell_2\}$ is a cherry in $T$ and $\{\ell_2,\ell_3\}$ is a cherry in $T'$, then reduce $T$ and $T'$ to $T_r=T|X\setminus C$ and $T_r'=T'|X\setminus C$, respectively.

\noindent {\bf Reduction 4.}~\cite{kelk2020new} If $T$ and $T'$ have a common 3-chain $C=(\ell_1,\ell_2,\ell_3)$ such that $\{\ell_2,\ell_3\}$ is a cherry in $T$ and $\{\ell_3,x\}$ is a cherry in $T'$ with $x\in X\setminus C$, then reduce $T$ and $T'$ to $T_r=T|X\setminus \{x\}$ and $T_r'=T'|X\setminus \{x\}$, respectively.

\noindent {\bf Reduction 5.}~\cite{kelk2020new} If $T$ and $T'$ have two common 2-chains $C_1=(\ell_1,\ell_2)$ and $C_2=(\ell_3,\ell_4)$ such that $T$ has cherries $\{\ell_2,x\}$ and $\{\ell_3,\ell_4\}$, and $T'$ has cherries $\{\ell_1,\ell_2\}$ and $\{\ell_4,x\}$ with $x\in X\setminus (C_1\cup C_2)$, then reduce $T$ and $T'$ to $T_r=T|X\setminus \{x\}$ and $T_r'=T'|X\setminus \{x\}$, respectively.

\noindent {\bf Reduction 6.}~\cite{kelk2020new} If $T$ and $T'$ have two common 3-chains $C_1=(\ell_1,\ell_2,\ell_3)$ and $C_2=(\ell_4,\ell_5,\ell_6)$ such that $T$ has cherries $\{\ell_2,\ell_3\}$ and $\{\ell_4,\ell_5\}$, and $(\ell_1,\ell_2,\ldots,\ell_6)$ is a 6-chain of $T'$, then reduce $T$ and $T'$ to $T_r=T|X\setminus \{\ell_4,\ell_5\}$ and $T_r'=T'|X\setminus \{\ell_4,\ell_5\}$, respectively.

\noindent {\bf Reduction 7.}~\cite{kelk2020new} If $T$ and $T'$ have common chains $C_1=(\ell_1,\ell_2,\ell_3)$ and $C_2=(\ell_4,\ell_5)$ such that $T$ has cherries $\{\ell_2,\ell_3\}$ and $\{\ell_4,\ell_5\}$, and $(\ell_1,\ell_2,\ldots,\ell_5)$ is a 5-chain of $T'$, then reduce $T$ and $T'$ to $T_r=T|X\setminus \{\ell_4\}$ and $T_r'=T'|X\setminus \{\ell_4\}$, respectively.\\

\noindent An example of Reduction 7 is illustrated in Figure~\ref{fig:reduction7}. Reduction 1 is known as {\it subtree reduction} while Reduction 2 is known as {\it chain reduction} in the literature \blue{(for example, see~\cite{AllenSteel2001})}. \revisionSK{\mbox{Reductions~3--7} assume that Reductions 1 and 2 have already been applied to exhaustion.}

The following two lemmas summarize results established in~\cite{AllenSteel2001,tightkernel,kelk2020new}.

\begin{figure}[t]
\center
\scalebox{1}{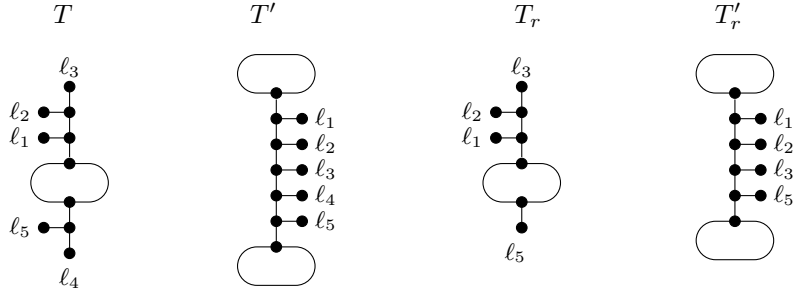}
\caption{An example of Reduction 7.  \blue{Ovals} indicate subtrees.}
%Note that we have omitted some parts of the trees and that it is not required that sets $P$, $Q$, $P'$, and $Q'$ are all non-empty.}
\label{fig:reduction7}
\end{figure}

\begin{lemma}
Let $T$ and $T'$ be two phylogenetic trees on $X$.
%with $|X|\geq 4$.
If $T_r$ and $T_r'$ are two phylogenetic trees obtained from $T$ and $T'$, respectively, by  a single application of Reduction 1,2, 6, or 7, then $d_\TBR(T,T') = d_\TBR(T_r,T_r')$. Moreover, if $T_r$ and $T_r'$ are two trees obtained from $T$ and $T'$, respectively, by a single application of Reduction 3, 4, or 5, then $d_\TBR(T,T')-1 = d_\TBR(T_r,T_r')$.
\end{lemma}

\begin{lemma}
Let $S$ and $S'$ be two phylogenetic trees on $X$  that cannot be reduced by Reduction 1 or 2, and let $T$ and $T'$ be \blue{two} phylogenetic trees on $Y$ that cannot be reduced by any of the seven reductions. If $d_\TBR(S,S')\geq 2$, then $|X|\leq 15d_\TBR(S,S')-9$. Furthermore, if $d_\TBR(T,T')\geq 2$, then $|Y|\leq 11d_\TBR(T,T')-9$.
\end{lemma}

Note that \blue{an application of Reduction} 3, 4, \blue{or} 5 triggers a \emph{parameter reduction}, whereby the TBR distance is reduced. In these cases, an element of $X$ is located which definitely comprises a singleton component in some maximum agreement forest, and whose deletion thus lowers the TBR distance by 1. Reductions 1, 2, 6 and 7, on the other hand, preserve TBR distance. Reductions~6 and 7 work by truncating short chains, i.e. chains which escape Reduction 2, to be even shorter. \revisionSK{It was noted in \cite{kelk2020new} that an application of Reduction 4 or 5 immediately triggers an application of Reduction 1.}

\section{Maximum Parsimony distance and a new approach to computing lower bounds on TBR distance}
\label{subsec:def_dmpbound}

Throughout this section, an unrooted phylogenetic tree $T$ is not necessarily binary, i.e. each internal vertex of $T$ has degree at least 3. A {\it character} $f$ on $X$ is a function $f:X \rightarrow C$, where $C=\{c_1,c_2,\ldots,c_r\}$ is a set of {\it character states} for some positive integer $r$. Let $T$ be an unrooted phylogenetic tree on $X$ with vertex set $V$, and let $f$ be a character on $X$ whose set of character states is $C$. An {\it extension} $g$ of $f$ to $V$ is  a function $g: V \rightarrow C$ such that $g(\ell)=f(\ell)$ for each $\ell\in X$. Given an extension $g$ of $f$, let $l_g(T)$ denote the number of edges $\{u,v\}$ in $T$ such that $g(u)\ne g(v)$. Then the {\it parsimony score} of $f$ on $T$, denoted by $l_f(T)$, is obtained by minimizing  $l_g(T)$ over all possible extensions $g$ of $f$; this score can easily be computed in polynomial time~\cite{fitch1971}.  \blue{Following standard terminology in phylogenetics~\cite{SempleSteel2003}, we say that a character $f$ is} \emph{convex} on $T$ if $l_f(T) = r-1$. \revisionSK{Equivalently, a convex character is a character where the $r$ minimal spanning trees induced by the $r$ states of the character are vertex disjoint}.

For two unrooted phylogenetic trees $T$ and $T'$ on $X$, the {\it maximum parsimony distance} $d_{\MP}$~\cite{fischer2014,moulton2015} is defined as  $$d_{\MP}(T,T')=\max_f |l_f(T)-l_f(T')|.$$  It is known that at least one such maximizing $f$ has the following two properties: (i) it is \blue{convex}
%\emph{convex} 
on at least one of $T$ and $T'$, (ii) each state in $f$ occurs on at least 2 taxa \cite{kelk2017complexity}.

It has been noted several times in the literature that $d_{\MP}$, which is itself NP-hard to compute, is a lower bound on $d_{\TBR}$ \cite{fischer2014,moulton2015}. Experiments in \cite{kelk2016reduction} on very small trees (up to 25 taxa), in which $d_{\MP}$ was computed exactly using the exponential-time algorithm from \cite{kelk2017note}, suggest that $d_{\MP}$ is often very close to $d_{\TBR}$. Inspired by this, we here propose a new scaleable strategy for computing good lower bounds on $d_{\TBR}$. Rather than computing $d_{\MP}$ exactly, which is too time-intensive, we leverage the fact that \emph{every} (convex) character $f$ gives a lower bound on $d_{\MP}$ and thus on $d_{\TBR}$. By sampling many such characters $f$, and selecting the one that maximizes  $|l_f(T)-l_f(T')|$ we expect to achieve a strong lower bound on $d_{\MP}$. This sampling strategy is made possible by \cite[Corollary 5]{kelk2017note} which states that for a given unrooted phylogenetic tree $T$, a character that is convex on $T$ whereby each state occurs on at least $p \geq 2$ taxa ($p$ constant) can be sampled uniformly at random in linear time and space; let us call such characters \emph{eligible}. The overall strategy, therefore, is to randomly choose between $T$ and $T'$, and then to uniformly at random select an eligible character $f$ (taking $p=2$), repeating this as often as desired, and at the end returning the largest value of $|l_f(T)-l_f(T')|$ observed. We henceforth call this the $d_{\MP}$ \emph{lower bound sampling strategy} or simply $d_{\MP}$ \emph{lower bound} if no confusion arises from the context and denote this bound by $d^\ell_\MP$.\\
%\\
%In our experiments we run this strategy for 10 seconds. We apply it \emph{after} the subtree reduction, since $d_{\MP}$ is also preserved by this reduction rule \cite{kelk2016reduction}.

\section{Experimental framework}\label{sec:exp-framework}

\subsection{uSPR}
\label{subsec:def_uspr}
To compute TBR distances in our experiments we mainly used the solver uSPR that is described in \cite{whidden2018calculating} and
available from \url{https://github.com/cwhidden/uspr}, version 1.0.1. This solver is designed to compute the so-called unrooted subtree prune and regraft distance between two phylogenetic trees $T$ and $T'$ which is always at least as large as $d_\TBR(T,T')$.  As one of its subroutines, uSPR  incorporates an exact TBR solver, which can be invoked independently  (using the \texttt{--tbr} switch). The TBR solver uses iterative deepening. \purple{More precisely, for two phylogenetic trees $T$ and $T'$ on $X$,} it starts from a polynomial-time computed lower bound $k'$ and repeatedly asks ``Is $d_{\TBR}(T,T') \leq k'$?'' for increasing values of $k'$ until the question is answered positively. According to \cite{whidden2018calculating} the TBR solver of uSPR incorporates a number of the enhanced branching cases described in Chen's algorithm~\cite{chen2015parameterized}, which computes TBR distance in time $O( 3^{d_{\TBR}} \cdot \text{poly}(|X|) )$.

\subsection{Tubro}
\label{subsec:def_tubro}
%{\it Do we want to describe Tubro before presenting the experimental results?}
The polynomial-time reduction rules described in Section \ref{subsc:def_reducrules} have been implemented in our package \emph{Tubro}. Tubro's main role in our experiments is to apply these reduction rules. However, it also has a secondary role. Specifically, Tubro can compute TBR distance using Integer Linear Programming (ILP). As explained in the next section, we use Tubro in an attempt to compute TBR distances for those pairs of phylogenetic trees that are out of range of uSPR. Given that the ILP formulation has $O(|X|^4)$ constraints, however, Tubro is limited to pairs of phylogenetic trees on $X$ which, after reduction, have (roughly) $|X| \leq 70$. Tubro has many other features; we defer a full
description of the package to the appendix.

\subsection{High-level description of the experiments}

All experiments were conducted on
the \revisionSK{Windows Subsystem for Linux (WSL) [Ubuntu 16.04.6 LTS]}, running under Windows 10, on a 64-bit HP Envy Laptop 13-ad0xx (quad-core i7-7500 @ 2.7 GHz), with 8 \steven{Gb} of memory. \revisionSK{In experiments such as this, which mainly run in memory and whereby disk accesses are limited, WSL has comparable performance to a native Linux system\footnote{\revisionSK{See \url{https://en.wikipedia.org/wiki/Windows_Subsystem_for_Linux}, section `Benchmarks'. Accessed 20th August 2021.}}.}

Our main dataset comprises 735 {\it tree pairs}, where each such pair $(T,T')$ consists of two phylogenetic trees on the same set of taxa and that were constructed as follows. For each $t \in  \{50, 100, 150, 200, 250, 300, 350\}$
we generated a random phylogenetic tree $T$ on $t$ taxa with \emph{skew} $s \in \{50, 70, 90\}$, where the concept of skew is explained below. Then, for each $k \in \{5, 10, 15, 20, 25, 30, 35\}$ we applied $k$ random TBR moves to $T$ to obtain
a phylogenetic tree $T'$. This ensures that $d_\TBR(T,T') \leq k$. Note that equality might not hold, due to different random TBR moves potentially cancelling each other out.  We produced 5 \emph{replicates} of each tree pair: that is, for each parameter combination
$(t,s,k)$ we independently produced 5 different pairs of trees. This yields $7 \times 3 \times 7 \times 5 = 735$ pairs of trees \purple{in total}. The skew $s$ refers to the random generation of \blue{a} binary phylogenetic tree via the following recursive process. \revisionSK{(The process actually generates a rooted binary phylogenetic tree, but upon completion it is turned into an unrooted tree by suppressing the degree-2 root.) \blue{Starting with a rooted binary phylogenetic tree on two leaves, we} place a taxon on one side of the root
%\marginpar{Which partition is chosen here?}
with probability $\frac{s}{100}$, and on the other side with probability $(1 - \frac{s}{100})$, \steven{and then recurse on the two sides of the \revisionSK{root} until the leaves are reached}.} Phylogenetic trees with a skew of 50 will be fairly balanced, corresponding to the standard Yule-Harding distribution~\cite{harding1971probabilities}
while a skew of 90 will tend to produce a heavily skewed, more ``linear'' phylogenetic tree with smaller subtrees hanging off a main backbone. We include the skew parameter to incorporate more variety into the topology of the starting tree $T$.

For each tree pair $(T,T')$, we computed and collected the following core data:
\begin{enumerate}[label=\bf{D\arabic*},start=1]
\item The number of taxa in the instance after application of the subtree reduction, henceforth denoted $s(T,T')$.
\item The number of taxa in the instance after application of the subtree and chain reduction, henceforth denoted $sc(T,T')$.
\item The number of taxa in the instance after application of all seven reductions, henceforth denoted $scn(T,T'$). For brevity we often call such tree pairs \emph{fully reduced}.
\item The \purple{number} of parameter reductions that took place.
\item The lower bound $d_\MP^\ell(T,T')$ on $d_\TBR(T,T')$ obtained by sampling convex characters for 10 seconds  \emph{after} the subtree reduction, since $d_{\MP}$ is also preserved by this reduction rule \cite{kelk2016reduction}.
%(see Section \ref{subsec:def_dmpbound} below for technical details).
We computed this because at the outset of the experiments it was not clear whether we would be able to compute $d_\TBR(T,T')$ exactly. However, the exact TBR distance or a \emph{lower} bound on this distance is needed to compute empirical kernel sizes (details follow below).
%a safe bound on the reduction achieved by the \green{seven reductions}.
\item For each tree pair $(T,T')$, the exact value of $d_\TBR(T,T')$ if known. The distance was declared \emph{known} if at least one of the following solution approaches terminated.
\begin{itemize}
\item[(i)] Run uSPR  for 5 minutes on the original tree pair $(T, T')$.
\item[(ii)] Run uSPR  for 5 minutes on the fully reduced tree pair obtained from $(T,T')$  (and take into account the effect of any parameter reduction achieved during the kernelization).
\item[(iii)] Run Tubro for 5 minutes on the fully reduced tree pair obtained from $(T,T')$ (and take into account the effect of any parameter reduction achieved during the kernelization). We only ran Tubro on those instances for which (i) and (ii) both failed to terminate, since it is not the goal of this article to directly compare
the solving power of uSPR and Tubro.
% This decision is motivated in \green{Section~\ref{sec:toughcookies}}\marginpar{Check section reference} while Tubro itself is described in detail \green{in the appendix}.
%Section \ref{subsec:def_tubro}.
\end{itemize}
\end{enumerate}
Note that, if uSPR (or Tubro) produces an intermediate lower bound of $k$ for a tree pair $(T,T')$, where $k$ is the number of TBR moves that were applied to create the pair, it follows that $d_{\TBR}(T,T')=k$. However, if the solver does not also terminate in the allotted time, we do not consider such instances known. This is because the solvers do not have this upper bound information available outside the experimental framework that we describe here.

From D1-D6, we have computed various secondary statistics which are presented in more detail in Section~\ref{sec:results}. Most notably:
\begin{enumerate}
\item {\bf \steven{The availability of exact TBR distances / difficult trees.}} The distribution of tree pairs whose TBR distance could not be computed exactly with uSPR or Tubro.
\item {\bf Average \steven{percentage} of remaining taxa.} For each pair $t$ and $k$, the average percentage of remaining taxa after the subtree reduction, after the subtree and chain reduction, and after all seven reductions over all 15 tree pairs with this parameter combination.
\item {\bf Empirical kernel size.} For each tree pair $(T,T')$, and for each of the different levels of reduction  $s(T,T')$, $sc(T,T')$, and $scn(T,T')$, the number of taxa in the reduced instance divided by $d_\TBR(T,T')$ if the TBR distance is known and, otherwise, divided by $d_\MP^\ell(T,T')$.
\item {\bf Parameter reductions.} The distribution of tree pairs that have undergone at least one of the three reductions that trigger parameter reductions.
\item {\bf The quality of the $d_{\MP}$ lower bound.} The distribution of tree pairs $(T,T')$ for which $d_\MP^\ell(T,T')< d_\TBR(T,T')$.
\end{enumerate}

For a second dataset that comprises 90 tree pairs of larger size, we generated 5 replicates for each parameter combination $(t,s,k)$ with  $t\in\{500,1000,1500,2000,2500,3000\}$, $s\in\{50,70, 90\}$, and $k=35$. To this end, we have followed the same approach as described for the main dataset  and collected the same data and statistics for each tree pair. We will refer to this dataset as the {\it larger trees dataset}.

While we will focus on analyzing the main dataset in the following section, we use the larger trees dataset to confirm that our results do not only apply to pairs of phylogenetic trees with at most 350 taxa but instead describe general trends and observations that are not restricted to trees of a certain size.

We have made both our datasets, plus the spreadsheets describing our results, available on the page
\steven{\url{https://github.com/skelk2001/kernelizing-agreement-forests/}}. The GitHub page also includes a stand-alone implementation of the $d_{\MP}$ lower
bound code, since we feel this is of independent interest. Source code for Tubro is available upon request.

\section{Results and analysis}\label{sec:results}

\steven{In each subsection below we combine the presentation of our results with some analysis and reflection.}
%In Section \ref{sec:toughcookies} we zoom out to provide a more global and less detail-heavy perspective on our overall findings.}

% SK: DONE (well enough...)
%\marginpar{todo: be careful about mixing up $k$ and $d_{TBR}$ in this section! $k$ is the experimental parameter}

\subsection{The availability of exact TBR distances / \steven{difficult trees}} \label{sec:exactTBR}
% SK: I deleted this text, it felt out-dated
%
%The goal of this subsection is not to show how much kernelization does, or does not, help us to solve more challenging instances in practice. Rather, it is simply to show for which part of the dataset we could obtain the exact TBR distance, and thus for which we could compute the size of the kernel exactly \steven{(expressed as a function of $d_{\TBR}$)}, as opposed to just a bound on it (for details, see Section~\ref{sec:f(k)}).
%We defer comments on tractability issues to a later section.

We start by providing the number of tree pairs of both  datasets for which we have computed the exact TBR distance, i.e. the number of tree pairs for which the TBR distance was declared \steven{\emph{known}}.
\begin{itemize}
\item 625 (85.0$\%$) of the 735 tree pairs of the main dataset could be solved by uSPR in 5 minutes, when using the original (unreduced)  trees.
\item 646 (87.9$\%$) of the 735 tree pairs could be solved by uSPR in 5 minutes, when using the fully reduced trees. In all cases where the original trees could be solved in 5 minutes, so too could the fully reduced trees.
\item Of the remaining 89 tree pairs, Tubro could solve 51 after allowing the ILP solver Gurobi to run for 5 minutes. Hence, in total, the exact TBR distance was calculated for 646+51=697 (94.8$\%$) of the tree pairs. For the remaining 38
tree pairs,  $d_\MP^\ell$ was used as a lower bound on the exact TBR distance. This lower bound was in particular used in the determination of empirical kernel sizes \steven{in  Section~\ref{sec:f(k)}}.
%the bounds on $f(k)$.
\item For the larger trees dataset, 86 tree pairs could be solved by running uSPR on the unreduced trees for 5 minutes and and additional 3 tree pairs could be solved by running uSPR on the reduced tree pairs. The remaining tree pair could not be solved by Tubro.
\end{itemize}
\steven{It is not the goal of this subsection to discuss how much kernelization does, or does not, help us to solve more challenging instances in practice. However, the experiments did yield some auxiliary insights in this direction, which we now describe.}
Applying all seven reductions allowed uSPR to solve 21 more instances of the main dataset than prior to reduction. While welcome, this is not a particularly large increase. This is not so surprising, because uSPR is a branching algorithm with exponential dependency on $d_\TBR(T,T')$ and only polynomial dependency on the number of taxa. Plus, uSPR contains an internal subtree reduction which already helps to reduce the number of taxa quite significantly. On the other hand, uSPR could potentially exploit parameter reduction, which does occur reasonably often (see Section~\ref{sec:param-red}).  There is some one-sided evidence that parameter reduction did help. Specifically, \blue{20 of the 21 instances that
uSPR solved \emph{after} reduction, exhibited parameter reduction.} On the other hand, amongst the
89 instances that uSPR could still not solve after reduction, only $27.2\%$ exhibited parameter reduction.

Regarding Tubro, we note that kernelization certainly \steven{helped} in the following rather vacuous sense: the generation and solving time for the underlying ILP becomes prohibitively large for instances with more than (roughly) 70 taxa. Prior to kernelization, 85.7\%  of the tree pairs had more than 70 taxa, and after full kernelization only 50.5\% had this property. As an \revisionSK{unparameterized} exponential-time algorithm, Tubro is naturally assisted more than an \revisionSK{FPT} branching algorithm by \revisionSK{the reduction in instance size (i.e., $|X|$) achieved by} kernelization. \steven{The 51 instances that Tubro could solve, but which uSPR could not, \blue{are summarized in Figure~\ref{fig:51-solved}. The average TBR distance of these 51 tree pairs is} 27 with a standard deviation of 4.1. Of the 51 instances, 14 had 50 taxa, 22 had 100 taxa, 11 had 150 taxa, 3 had 250 taxa, and 1 had 250 taxa.}

\begin{figure}[h!]
\begin{center}
\scalebox{0.7}{\input{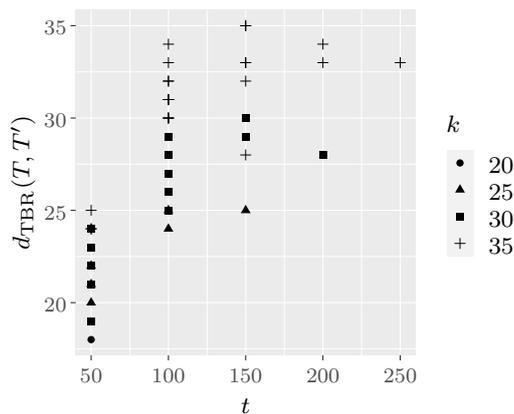}}
\caption{\blue{Distribution of all 51 tree pairs of the main dataset that could be solved by Tubro but could not be solved by uSPR when applied to the fully reduced instances. Note that some of these 51 tree pairs have the same value for $t$, $k$, and $d_\TBR$ in which case their visualized data points coincide.}}
\label{fig:51-solved}
\end{center}
\end{figure}

\begin{figure}[h!]
\noindent\begin{minipage}[t]{0.33\textwidth}
\scalebox{0.5}{\input{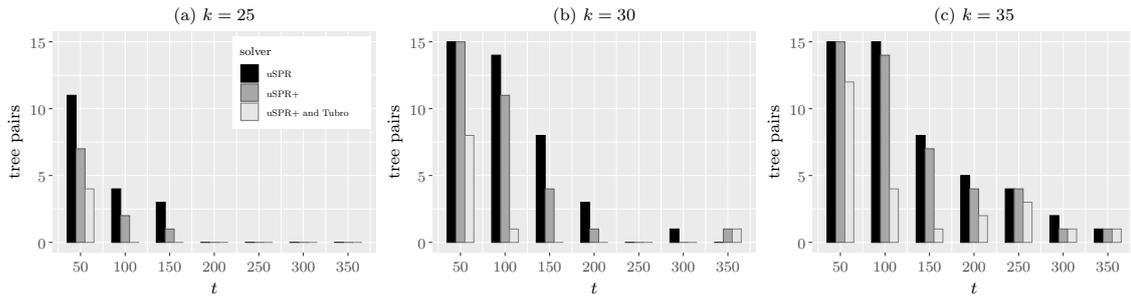}}
\end{minipage}
\begin{minipage}[t]{0.33\textwidth}
\scalebox{0.5}{% Created by tikzDevice version 0.12.3.1 on 2021-08-28 18:22:54
% !TEX encoding = UTF-8 Unicode
\begin{tikzpicture}[x=1pt,y=1pt]
\definecolor{fillColor}{RGB}{255,255,255}
\path[use as bounding box,fill=fillColor,fill opacity=0.00] (0,0) rectangle (289.08,231.26);
\begin{scope}
\path[clip] (  0.00,  0.00) rectangle (289.08,231.26);
\definecolor{drawColor}{RGB}{255,255,255}
\definecolor{fillColor}{RGB}{255,255,255}

\path[draw=drawColor,line width= 0.6pt,line join=round,line cap=round,fill=fillColor] (  0.00,  0.00) rectangle (289.08,231.26);
\end{scope}
\begin{scope}
\path[clip] ( 41.33, 39.69) rectangle (283.58,206.13);
\definecolor{fillColor}{gray}{0.92}

\path[fill=fillColor] ( 41.33, 39.69) rectangle (283.58,206.13);
\definecolor{drawColor}{RGB}{255,255,255}

\path[draw=drawColor,line width= 0.3pt,line join=round] ( 41.33, 72.48) --
	(283.58, 72.48);

\path[draw=drawColor,line width= 0.3pt,line join=round] ( 41.33,122.91) --
	(283.58,122.91);

\path[draw=drawColor,line width= 0.3pt,line join=round] ( 41.33,173.35) --
	(283.58,173.35);

\path[draw=drawColor,line width= 0.3pt,line join=round] ( 45.66, 39.69) --
	( 45.66,206.13);

\path[draw=drawColor,line width= 0.3pt,line join=round] ( 79.03, 39.69) --
	( 79.03,206.13);

\path[draw=drawColor,line width= 0.3pt,line join=round] (112.40, 39.69) --
	(112.40,206.13);

\path[draw=drawColor,line width= 0.3pt,line join=round] (145.77, 39.69) --
	(145.77,206.13);

\path[draw=drawColor,line width= 0.3pt,line join=round] (179.14, 39.69) --
	(179.14,206.13);

\path[draw=drawColor,line width= 0.3pt,line join=round] (212.51, 39.69) --
	(212.51,206.13);

\path[draw=drawColor,line width= 0.3pt,line join=round] (245.87, 39.69) --
	(245.87,206.13);

\path[draw=drawColor,line width= 0.3pt,line join=round] (279.24, 39.69) --
	(279.24,206.13);

\path[draw=drawColor,line width= 0.6pt,line join=round] ( 41.33, 47.26) --
	(283.58, 47.26);

\path[draw=drawColor,line width= 0.6pt,line join=round] ( 41.33, 97.70) --
	(283.58, 97.70);

\path[draw=drawColor,line width= 0.6pt,line join=round] ( 41.33,148.13) --
	(283.58,148.13);

\path[draw=drawColor,line width= 0.6pt,line join=round] ( 41.33,198.57) --
	(283.58,198.57);

\path[draw=drawColor,line width= 0.6pt,line join=round] ( 62.35, 39.69) --
	( 62.35,206.13);

\path[draw=drawColor,line width= 0.6pt,line join=round] ( 95.72, 39.69) --
	( 95.72,206.13);

\path[draw=drawColor,line width= 0.6pt,line join=round] (129.08, 39.69) --
	(129.08,206.13);

\path[draw=drawColor,line width= 0.6pt,line join=round] (162.45, 39.69) --
	(162.45,206.13);

\path[draw=drawColor,line width= 0.6pt,line join=round] (195.82, 39.69) --
	(195.82,206.13);

\path[draw=drawColor,line width= 0.6pt,line join=round] (229.19, 39.69) --
	(229.19,206.13);

\path[draw=drawColor,line width= 0.6pt,line join=round] (262.56, 39.69) --
	(262.56,206.13);
\definecolor{drawColor}{RGB}{0,0,0}
\definecolor{fillColor}{RGB}{0,0,0}

\path[draw=drawColor,line width= 0.6pt,line cap=rect,fill=fillColor] ( 52.34, 47.26) rectangle ( 59.01,198.57);

\path[draw=drawColor,line width= 0.6pt,line cap=rect,fill=fillColor] ( 85.71, 47.26) rectangle ( 92.38,188.48);

\path[draw=drawColor,line width= 0.6pt,line cap=rect,fill=fillColor] (119.07, 47.26) rectangle (125.75,127.96);

\path[draw=drawColor,line width= 0.6pt,line cap=rect,fill=fillColor] (152.44, 47.26) rectangle (159.12, 77.52);

\path[draw=drawColor,line width= 0.6pt,line cap=rect,fill=fillColor] (185.81, 47.26) rectangle (192.48, 47.26);

\path[draw=drawColor,line width= 0.6pt,line cap=rect,fill=fillColor] (219.18, 47.26) rectangle (225.85, 57.35);

\path[draw=drawColor,line width= 0.6pt,line cap=rect,fill=fillColor] (252.55, 47.26) rectangle (259.22, 47.26);
\definecolor{drawColor}{RGB}{93,93,93}
\definecolor{fillColor}{RGB}{167,167,167}

\path[draw=drawColor,line width= 0.6pt,line cap=rect,fill=fillColor] ( 59.01, 47.26) rectangle ( 65.69,198.57);

\path[draw=drawColor,line width= 0.6pt,line cap=rect,fill=fillColor] ( 92.38, 47.26) rectangle ( 99.05,158.22);

\path[draw=drawColor,line width= 0.6pt,line cap=rect,fill=fillColor] (125.75, 47.26) rectangle (132.42, 87.61);

\path[draw=drawColor,line width= 0.6pt,line cap=rect,fill=fillColor] (159.12, 47.26) rectangle (165.79, 57.35);

\path[draw=drawColor,line width= 0.6pt,line cap=rect,fill=fillColor] (192.48, 47.26) rectangle (199.16, 47.26);

\path[draw=drawColor,line width= 0.6pt,line cap=rect,fill=fillColor] (225.85, 47.26) rectangle (232.53, 47.26);

\path[draw=drawColor,line width= 0.6pt,line cap=rect,fill=fillColor] (259.22, 47.26) rectangle (265.89, 57.35);
\definecolor{drawColor}{RGB}{128,128,128}
\definecolor{fillColor}{RGB}{230,230,230}

\path[draw=drawColor,line width= 0.6pt,line cap=rect,fill=fillColor] ( 65.69, 47.26) rectangle ( 72.36,127.96);

\path[draw=drawColor,line width= 0.6pt,line cap=rect,fill=fillColor] ( 99.05, 47.26) rectangle (105.73, 57.35);

\path[draw=drawColor,line width= 0.6pt,line cap=rect,fill=fillColor] (132.42, 47.26) rectangle (139.10, 47.26);

\path[draw=drawColor,line width= 0.6pt,line cap=rect,fill=fillColor] (165.79, 47.26) rectangle (172.46, 47.26);

\path[draw=drawColor,line width= 0.6pt,line cap=rect,fill=fillColor] (199.16, 47.26) rectangle (205.83, 47.26);

\path[draw=drawColor,line width= 0.6pt,line cap=rect,fill=fillColor] (232.53, 47.26) rectangle (239.20, 47.26);

\path[draw=drawColor,line width= 0.6pt,line cap=rect,fill=fillColor] (265.89, 47.26) rectangle (272.57, 57.35);
\end{scope}
\begin{scope}
\path[clip] (  0.00,  0.00) rectangle (289.08,231.26);
\definecolor{drawColor}{gray}{0.30}

\node[text=drawColor,anchor=base east,inner sep=0pt, outer sep=0pt, scale=  1.40] at ( 36.38, 42.44) {0};

\node[text=drawColor,anchor=base east,inner sep=0pt, outer sep=0pt, scale=  1.40] at ( 36.38, 92.87) {5};

\node[text=drawColor,anchor=base east,inner sep=0pt, outer sep=0pt, scale=  1.40] at ( 36.38,143.31) {10};

\node[text=drawColor,anchor=base east,inner sep=0pt, outer sep=0pt, scale=  1.40] at ( 36.38,193.75) {15};
\end{scope}
\begin{scope}
\path[clip] (  0.00,  0.00) rectangle (289.08,231.26);
\definecolor{drawColor}{gray}{0.20}

\path[draw=drawColor,line width= 0.6pt,line join=round] ( 38.58, 47.26) --
	( 41.33, 47.26);

\path[draw=drawColor,line width= 0.6pt,line join=round] ( 38.58, 97.70) --
	( 41.33, 97.70);

\path[draw=drawColor,line width= 0.6pt,line join=round] ( 38.58,148.13) --
	( 41.33,148.13);

\path[draw=drawColor,line width= 0.6pt,line join=round] ( 38.58,198.57) --
	( 41.33,198.57);
\end{scope}
\begin{scope}
\path[clip] (  0.00,  0.00) rectangle (289.08,231.26);
\definecolor{drawColor}{gray}{0.20}

\path[draw=drawColor,line width= 0.6pt,line join=round] ( 62.35, 36.94) --
	( 62.35, 39.69);

\path[draw=drawColor,line width= 0.6pt,line join=round] ( 95.72, 36.94) --
	( 95.72, 39.69);

\path[draw=drawColor,line width= 0.6pt,line join=round] (129.08, 36.94) --
	(129.08, 39.69);

\path[draw=drawColor,line width= 0.6pt,line join=round] (162.45, 36.94) --
	(162.45, 39.69);

\path[draw=drawColor,line width= 0.6pt,line join=round] (195.82, 36.94) --
	(195.82, 39.69);

\path[draw=drawColor,line width= 0.6pt,line join=round] (229.19, 36.94) --
	(229.19, 39.69);

\path[draw=drawColor,line width= 0.6pt,line join=round] (262.56, 36.94) --
	(262.56, 39.69);
\end{scope}
\begin{scope}
\path[clip] (  0.00,  0.00) rectangle (289.08,231.26);
\definecolor{drawColor}{gray}{0.30}

\node[text=drawColor,anchor=base,inner sep=0pt, outer sep=0pt, scale=  1.40] at ( 62.35, 25.10) {50};

\node[text=drawColor,anchor=base,inner sep=0pt, outer sep=0pt, scale=  1.40] at ( 95.72, 25.10) {100};

\node[text=drawColor,anchor=base,inner sep=0pt, outer sep=0pt, scale=  1.40] at (129.08, 25.10) {150};

\node[text=drawColor,anchor=base,inner sep=0pt, outer sep=0pt, scale=  1.40] at (162.45, 25.10) {200};

\node[text=drawColor,anchor=base,inner sep=0pt, outer sep=0pt, scale=  1.40] at (195.82, 25.10) {250};

\node[text=drawColor,anchor=base,inner sep=0pt, outer sep=0pt, scale=  1.40] at (229.19, 25.10) {300};

\node[text=drawColor,anchor=base,inner sep=0pt, outer sep=0pt, scale=  1.40] at (262.56, 25.10) {350};
\end{scope}
\begin{scope}
\path[clip] (  0.00,  0.00) rectangle (289.08,231.26);
\definecolor{drawColor}{RGB}{0,0,0}

\node[text=drawColor,anchor=base,inner sep=0pt, outer sep=0pt, scale=  1.60] at (162.45,  8.61) {$t$};
\end{scope}
\begin{scope}
\path[clip] (  0.00,  0.00) rectangle (289.08,231.26);
\definecolor{drawColor}{RGB}{0,0,0}

\node[text=drawColor,rotate= 90.00,anchor=base,inner sep=0pt, outer sep=0pt, scale=  1.60] at ( 16.52,122.91) {tree pairs};
\end{scope}
\begin{scope}
\path[clip] (  0.00,  0.00) rectangle (289.08,231.26);
\definecolor{drawColor}{RGB}{0,0,0}

\node[text=drawColor,anchor=base,inner sep=0pt, outer sep=0pt, scale=  1.60] at (162.45,214.74) {(b) $k=30$};
\end{scope}
\end{tikzpicture}}
\end{minipage}
\begin{minipage}[t]{0.33\textwidth}
\scalebox{0.5}{% Created by tikzDevice version 0.12.3.1 on 2021-08-28 18:22:55
% !TEX encoding = UTF-8 Unicode
\begin{tikzpicture}[x=1pt,y=1pt]
\definecolor{fillColor}{RGB}{255,255,255}
\path[use as bounding box,fill=fillColor,fill opacity=0.00] (0,0) rectangle (289.08,231.26);
\begin{scope}
\path[clip] (  0.00,  0.00) rectangle (289.08,231.26);
\definecolor{drawColor}{RGB}{255,255,255}
\definecolor{fillColor}{RGB}{255,255,255}

\path[draw=drawColor,line width= 0.6pt,line join=round,line cap=round,fill=fillColor] (  0.00,  0.00) rectangle (289.08,231.26);
\end{scope}
\begin{scope}
\path[clip] ( 41.33, 39.69) rectangle (283.58,206.13);
\definecolor{fillColor}{gray}{0.92}

\path[fill=fillColor] ( 41.33, 39.69) rectangle (283.58,206.13);
\definecolor{drawColor}{RGB}{255,255,255}

\path[draw=drawColor,line width= 0.3pt,line join=round] ( 41.33, 72.48) --
	(283.58, 72.48);

\path[draw=drawColor,line width= 0.3pt,line join=round] ( 41.33,122.91) --
	(283.58,122.91);

\path[draw=drawColor,line width= 0.3pt,line join=round] ( 41.33,173.35) --
	(283.58,173.35);

\path[draw=drawColor,line width= 0.3pt,line join=round] ( 45.66, 39.69) --
	( 45.66,206.13);

\path[draw=drawColor,line width= 0.3pt,line join=round] ( 79.03, 39.69) --
	( 79.03,206.13);

\path[draw=drawColor,line width= 0.3pt,line join=round] (112.40, 39.69) --
	(112.40,206.13);

\path[draw=drawColor,line width= 0.3pt,line join=round] (145.77, 39.69) --
	(145.77,206.13);

\path[draw=drawColor,line width= 0.3pt,line join=round] (179.14, 39.69) --
	(179.14,206.13);

\path[draw=drawColor,line width= 0.3pt,line join=round] (212.51, 39.69) --
	(212.51,206.13);

\path[draw=drawColor,line width= 0.3pt,line join=round] (245.87, 39.69) --
	(245.87,206.13);

\path[draw=drawColor,line width= 0.3pt,line join=round] (279.24, 39.69) --
	(279.24,206.13);

\path[draw=drawColor,line width= 0.6pt,line join=round] ( 41.33, 47.26) --
	(283.58, 47.26);

\path[draw=drawColor,line width= 0.6pt,line join=round] ( 41.33, 97.70) --
	(283.58, 97.70);

\path[draw=drawColor,line width= 0.6pt,line join=round] ( 41.33,148.13) --
	(283.58,148.13);

\path[draw=drawColor,line width= 0.6pt,line join=round] ( 41.33,198.57) --
	(283.58,198.57);

\path[draw=drawColor,line width= 0.6pt,line join=round] ( 62.35, 39.69) --
	( 62.35,206.13);

\path[draw=drawColor,line width= 0.6pt,line join=round] ( 95.72, 39.69) --
	( 95.72,206.13);

\path[draw=drawColor,line width= 0.6pt,line join=round] (129.08, 39.69) --
	(129.08,206.13);

\path[draw=drawColor,line width= 0.6pt,line join=round] (162.45, 39.69) --
	(162.45,206.13);

\path[draw=drawColor,line width= 0.6pt,line join=round] (195.82, 39.69) --
	(195.82,206.13);

\path[draw=drawColor,line width= 0.6pt,line join=round] (229.19, 39.69) --
	(229.19,206.13);

\path[draw=drawColor,line width= 0.6pt,line join=round] (262.56, 39.69) --
	(262.56,206.13);
\definecolor{drawColor}{RGB}{0,0,0}
\definecolor{fillColor}{RGB}{0,0,0}

\path[draw=drawColor,line width= 0.6pt,line cap=rect,fill=fillColor] ( 52.34, 47.26) rectangle ( 59.01,198.57);

\path[draw=drawColor,line width= 0.6pt,line cap=rect,fill=fillColor] ( 85.71, 47.26) rectangle ( 92.38,198.57);

\path[draw=drawColor,line width= 0.6pt,line cap=rect,fill=fillColor] (119.07, 47.26) rectangle (125.75,127.96);

\path[draw=drawColor,line width= 0.6pt,line cap=rect,fill=fillColor] (152.44, 47.26) rectangle (159.12, 97.70);

\path[draw=drawColor,line width= 0.6pt,line cap=rect,fill=fillColor] (185.81, 47.26) rectangle (192.48, 87.61);

\path[draw=drawColor,line width= 0.6pt,line cap=rect,fill=fillColor] (219.18, 47.26) rectangle (225.85, 67.43);

\path[draw=drawColor,line width= 0.6pt,line cap=rect,fill=fillColor] (252.55, 47.26) rectangle (259.22, 57.35);
\definecolor{drawColor}{RGB}{93,93,93}
\definecolor{fillColor}{RGB}{167,167,167}

\path[draw=drawColor,line width= 0.6pt,line cap=rect,fill=fillColor] ( 59.01, 47.26) rectangle ( 65.69,198.57);

\path[draw=drawColor,line width= 0.6pt,line cap=rect,fill=fillColor] ( 92.38, 47.26) rectangle ( 99.05,188.48);

\path[draw=drawColor,line width= 0.6pt,line cap=rect,fill=fillColor] (125.75, 47.26) rectangle (132.42,117.87);

\path[draw=drawColor,line width= 0.6pt,line cap=rect,fill=fillColor] (159.12, 47.26) rectangle (165.79, 87.61);

\path[draw=drawColor,line width= 0.6pt,line cap=rect,fill=fillColor] (192.48, 47.26) rectangle (199.16, 87.61);

\path[draw=drawColor,line width= 0.6pt,line cap=rect,fill=fillColor] (225.85, 47.26) rectangle (232.53, 57.35);

\path[draw=drawColor,line width= 0.6pt,line cap=rect,fill=fillColor] (259.22, 47.26) rectangle (265.89, 57.35);
\definecolor{drawColor}{RGB}{128,128,128}
\definecolor{fillColor}{RGB}{230,230,230}

\path[draw=drawColor,line width= 0.6pt,line cap=rect,fill=fillColor] ( 65.69, 47.26) rectangle ( 72.36,168.31);

\path[draw=drawColor,line width= 0.6pt,line cap=rect,fill=fillColor] ( 99.05, 47.26) rectangle (105.73, 87.61);

\path[draw=drawColor,line width= 0.6pt,line cap=rect,fill=fillColor] (132.42, 47.26) rectangle (139.10, 57.35);

\path[draw=drawColor,line width= 0.6pt,line cap=rect,fill=fillColor] (165.79, 47.26) rectangle (172.46, 67.43);

\path[draw=drawColor,line width= 0.6pt,line cap=rect,fill=fillColor] (199.16, 47.26) rectangle (205.83, 77.52);

\path[draw=drawColor,line width= 0.6pt,line cap=rect,fill=fillColor] (232.53, 47.26) rectangle (239.20, 57.35);

\path[draw=drawColor,line width= 0.6pt,line cap=rect,fill=fillColor] (265.89, 47.26) rectangle (272.57, 57.35);
\end{scope}
\begin{scope}
\path[clip] (  0.00,  0.00) rectangle (289.08,231.26);
\definecolor{drawColor}{gray}{0.30}

\node[text=drawColor,anchor=base east,inner sep=0pt, outer sep=0pt, scale=  1.40] at ( 36.38, 42.44) {0};

\node[text=drawColor,anchor=base east,inner sep=0pt, outer sep=0pt, scale=  1.40] at ( 36.38, 92.87) {5};

\node[text=drawColor,anchor=base east,inner sep=0pt, outer sep=0pt, scale=  1.40] at ( 36.38,143.31) {10};

\node[text=drawColor,anchor=base east,inner sep=0pt, outer sep=0pt, scale=  1.40] at ( 36.38,193.75) {15};
\end{scope}
\begin{scope}
\path[clip] (  0.00,  0.00) rectangle (289.08,231.26);
\definecolor{drawColor}{gray}{0.20}

\path[draw=drawColor,line width= 0.6pt,line join=round] ( 38.58, 47.26) --
	( 41.33, 47.26);

\path[draw=drawColor,line width= 0.6pt,line join=round] ( 38.58, 97.70) --
	( 41.33, 97.70);

\path[draw=drawColor,line width= 0.6pt,line join=round] ( 38.58,148.13) --
	( 41.33,148.13);

\path[draw=drawColor,line width= 0.6pt,line join=round] ( 38.58,198.57) --
	( 41.33,198.57);
\end{scope}
\begin{scope}
\path[clip] (  0.00,  0.00) rectangle (289.08,231.26);
\definecolor{drawColor}{gray}{0.20}

\path[draw=drawColor,line width= 0.6pt,line join=round] ( 62.35, 36.94) --
	( 62.35, 39.69);

\path[draw=drawColor,line width= 0.6pt,line join=round] ( 95.72, 36.94) --
	( 95.72, 39.69);

\path[draw=drawColor,line width= 0.6pt,line join=round] (129.08, 36.94) --
	(129.08, 39.69);

\path[draw=drawColor,line width= 0.6pt,line join=round] (162.45, 36.94) --
	(162.45, 39.69);

\path[draw=drawColor,line width= 0.6pt,line join=round] (195.82, 36.94) --
	(195.82, 39.69);

\path[draw=drawColor,line width= 0.6pt,line join=round] (229.19, 36.94) --
	(229.19, 39.69);

\path[draw=drawColor,line width= 0.6pt,line join=round] (262.56, 36.94) --
	(262.56, 39.69);
\end{scope}
\begin{scope}
\path[clip] (  0.00,  0.00) rectangle (289.08,231.26);
\definecolor{drawColor}{gray}{0.30}

\node[text=drawColor,anchor=base,inner sep=0pt, outer sep=0pt, scale=  1.40] at ( 62.35, 25.10) {50};

\node[text=drawColor,anchor=base,inner sep=0pt, outer sep=0pt, scale=  1.40] at ( 95.72, 25.10) {100};

\node[text=drawColor,anchor=base,inner sep=0pt, outer sep=0pt, scale=  1.40] at (129.08, 25.10) {150};

\node[text=drawColor,anchor=base,inner sep=0pt, outer sep=0pt, scale=  1.40] at (162.45, 25.10) {200};

\node[text=drawColor,anchor=base,inner sep=0pt, outer sep=0pt, scale=  1.40] at (195.82, 25.10) {250};

\node[text=drawColor,anchor=base,inner sep=0pt, outer sep=0pt, scale=  1.40] at (229.19, 25.10) {300};

\node[text=drawColor,anchor=base,inner sep=0pt, outer sep=0pt, scale=  1.40] at (262.56, 25.10) {350};
\end{scope}
\begin{scope}
\path[clip] (  0.00,  0.00) rectangle (289.08,231.26);
\definecolor{drawColor}{RGB}{0,0,0}

\node[text=drawColor,anchor=base,inner sep=0pt, outer sep=0pt, scale=  1.60] at (162.45,  8.61) {$t$};
\end{scope}
\begin{scope}
\path[clip] (  0.00,  0.00) rectangle (289.08,231.26);
\definecolor{drawColor}{RGB}{0,0,0}

\node[text=drawColor,rotate= 90.00,anchor=base,inner sep=0pt, outer sep=0pt, scale=  1.60] at ( 16.52,122.91) {tree pairs};
\end{scope}
\begin{scope}
\path[clip] (  0.00,  0.00) rectangle (289.08,231.26);
\definecolor{drawColor}{RGB}{0,0,0}

\node[text=drawColor,anchor=base,inner sep=0pt, outer sep=0pt, scale=  1.60] at (162.45,214.74) {(c) $k=35$};
\end{scope}
\end{tikzpicture}}
\end{minipage}
\caption{\blue{Distribution of all 109 (black, uSPR), 88 (dark gray, uSPR+), and 38 (light gray, uSPR+ and Tubro) tree pairs of the main dataset \revisionSK{with $k \geq 25$} that could not be solved exactly with uSPR applied to the subtree reduced tree pairs, with uSPR  applied to the fully reduced tree pairs, and with Tubro applied to the fully reduced tree pairs, respectively. Note that one tree pair with $k=20$ and $t=50$ could not be solved with uSPR or uSPR+ but could be solved with Tubro.}}
%Right: Summary of all 38 tree pairs that could not be solved exactly depending on their skew.}
\label{fig:histo-unsolved-tree-pairs}
\end{figure}

\blue{Three histograms that present the distribution of all tree pairs of the main dataset that could not be solved with the different solvers} over all  combinations of $k\in\{5,10, 15,\ldots,35\}$ and $t=\{50,100,150,\ldots,350\}$ \revisionSK{are} shown in Figure~\ref{fig:histo-unsolved-tree-pairs}. \blue{ The 110 tree pairs that could not be solved with uSPR applied to the subtree reduced tree pairs are presented by black bars. Similarly, the 89 (resp. 38) tree pairs that could not be solved with uSPR applied to the fully reduced tree pairs (resp. uSPR or Tubro applied to the fully reduced tree pairs) are presented by dark gray (resp. light gray) bars. Regardless of which solver was used, all unsolvable tree pairs \revisionSK{were} generated by applying $k\geq 25$ random TBR moves, except for one pair with $k=20$ and $t=50$ that could not be solved with uSPR but could be solved with Tubro \revisionSK{(which is why the numbers 109 and 88 are reported in the figure caption).}
Turning to the 38 tree pairs that could not be solved by any of the solvers uSPR and Tubro,}  \revisionSK{12 had skew 50, 17 had skew 70 and 9 had skew 90, so tree pairs with a skew of 70 seem mildly over-represented. Furthermore, amongst the 38 unsolved instances, a majority of 24  tree pairs have the property $t=50$.}
%This suggests that instances consisting of two relatively small trees that have a large TBR distance are more difficult to solve than tree pairs with a large number of taxa and a large TBR distance.
% SK: is this better?
%\marginpar{SK: we could consider moving this new text partially or fully to Section \ref{sec:toughcookies}? \green{SL. Yes, I think that would be good. It's sometimes hard to differentiate between results and discussion}}
Given the parameterized running time of uSPR, which increases exponentially as a function of $k$, it is not so surprising that uSPR had difficulties with the large TBR distance $(k \geq 25)$ of the 38 unsolved tree pairs. It is less obvious why Tubro, which is based on  ILP, finds these instances difficult. Clearly, \blue{the ILP generated by Tubro quickly becomes prohibitively large for tree pairs that have more than 70 taxa after kernelization}.
%that still have many taxa after reduction, will vacuously fall out of range of Tubro since the ILP generated by Tubro quickly becomes prohibitively large for more than 70 taxa. 
However, for $t=50$ the ILP generated by Tubro will be comparatively small, and quick to generate, so this does not explain why Tubro struggles for some such tree pairs \purple{while} it succeeds on others (14 of the 51 tree pairs that Tubro could solve but uSPR could not, had $t=50$.) It is difficult to attach far-reaching
conclusions to this, but in any case it seems that instances with a small number of leaves, where the TBR distance is high (as a function of the number of leaves) are potentially a challenge for both uSPR and Tubro. For uSPR, these instances might
even be \emph{harder} than instances with the same TBR distance, but more taxa. \revisionSK{An \emph{informal}, intuitive argument for this could be that, when the TBR distance is high and the number of taxa is low, the TBR moves required to turn one tree into the other heavily `overlap' and `interfere' with each other. Viewed through the lens of agreement
forests, the underlying maximum agreement forest subsequently has very little structure, causing uSPR and Tubro to approach their worst-case behaviour.} We return to these tractability issues in Section \ref{sec:toughcookies}.
% it might be linked to the phenomenon that tree pairs with high TBR distance relative to the number of leaves, will tend to have MAFs consisting of many very small components.}

\subsection{\steven{Average percentage of remaining taxa}}
\label{subsec:avgreduc}

\begin{figure}[h!]
\noindent\begin{minipage}[t]{0.33\textwidth}
\scalebox{0.48}{% Created by tikzDevice version 0.12.3.1 on 2021-08-29 19:41:12
% !TEX encoding = UTF-8 Unicode
\begin{tikzpicture}[x=1pt,y=1pt]
\definecolor{fillColor}{RGB}{255,255,255}
\path[use as bounding box,fill=fillColor,fill opacity=0.00] (0,0) rectangle (289.08,289.08);
\begin{scope}
\path[clip] (  0.00,  0.00) rectangle (289.08,289.08);
\definecolor{drawColor}{RGB}{255,255,255}
\definecolor{fillColor}{RGB}{255,255,255}

\path[draw=drawColor,line width= 0.6pt,line join=round,line cap=round,fill=fillColor] (  0.00,  0.00) rectangle (289.08,289.08);
\end{scope}
\begin{scope}
\path[clip] ( 48.32, 39.69) rectangle (283.58,263.95);
\definecolor{fillColor}{gray}{0.92}

\path[fill=fillColor] ( 48.32, 39.69) rectangle (283.58,263.95);
\definecolor{drawColor}{RGB}{255,255,255}

\path[draw=drawColor,line width= 0.3pt,line join=round] ( 48.32, 75.37) --
	(283.58, 75.37);

\path[draw=drawColor,line width= 0.3pt,line join=round] ( 48.32,126.34) --
	(283.58,126.34);

\path[draw=drawColor,line width= 0.3pt,line join=round] ( 48.32,177.31) --
	(283.58,177.31);

\path[draw=drawColor,line width= 0.3pt,line join=round] ( 48.32,228.27) --
	(283.58,228.27);

\path[draw=drawColor,line width= 0.3pt,line join=round] ( 59.02, 39.69) --
	( 59.02,263.95);

\path[draw=drawColor,line width= 0.3pt,line join=round] (130.31, 39.69) --
	(130.31,263.95);

\path[draw=drawColor,line width= 0.3pt,line join=round] (201.60, 39.69) --
	(201.60,263.95);

\path[draw=drawColor,line width= 0.3pt,line join=round] (272.89, 39.69) --
	(272.89,263.95);

\path[draw=drawColor,line width= 0.6pt,line join=round] ( 48.32, 49.89) --
	(283.58, 49.89);

\path[draw=drawColor,line width= 0.6pt,line join=round] ( 48.32,100.85) --
	(283.58,100.85);

\path[draw=drawColor,line width= 0.6pt,line join=round] ( 48.32,151.82) --
	(283.58,151.82);

\path[draw=drawColor,line width= 0.6pt,line join=round] ( 48.32,202.79) --
	(283.58,202.79);

\path[draw=drawColor,line width= 0.6pt,line join=round] ( 48.32,253.76) --
	(283.58,253.76);

\path[draw=drawColor,line width= 0.6pt,line join=round] ( 94.66, 39.69) --
	( 94.66,263.95);

\path[draw=drawColor,line width= 0.6pt,line join=round] (165.95, 39.69) --
	(165.95,263.95);

\path[draw=drawColor,line width= 0.6pt,line join=round] (237.24, 39.69) --
	(237.24,263.95);
\definecolor{drawColor}{RGB}{0,0,0}

\path[draw=drawColor,line width= 0.6pt,line join=round] ( 59.02,159.16) --
	( 94.66,127.49) --
	(130.31,112.32) --
	(165.95, 99.77) --
	(201.60, 92.89) --
	(237.24, 88.67) --
	(272.89, 84.25);

\path[draw=drawColor,line width= 0.6pt,dash pattern=on 1pt off 3pt ,line join=round] ( 59.02,148.29) --
	( 94.66,110.78) --
	(130.31, 96.37) --
	(165.95, 86.45) --
	(201.60, 78.70) --
	(237.24, 75.80) --
	(272.89, 72.60);

\path[draw=drawColor,line width= 0.6pt,dash pattern=on 4pt off 4pt ,line join=round] ( 59.02,136.60) --
	( 94.66,102.21) --
	(130.31, 84.05) --
	(165.95, 78.16) --
	(201.60, 74.03) --
	(237.24, 70.23) --
	(272.89, 67.71);
\definecolor{fillColor}{RGB}{0,0,0}

\path[draw=drawColor,line width= 0.4pt,line join=round,line cap=round,fill=fillColor] ( 59.02,159.16) circle (  1.96);

\path[draw=drawColor,line width= 0.4pt,line join=round,line cap=round,fill=fillColor] ( 94.66,127.49) circle (  1.96);

\path[draw=drawColor,line width= 0.4pt,line join=round,line cap=round,fill=fillColor] (130.31,112.32) circle (  1.96);

\path[draw=drawColor,line width= 0.4pt,line join=round,line cap=round,fill=fillColor] (165.95, 99.77) circle (  1.96);

\path[draw=drawColor,line width= 0.4pt,line join=round,line cap=round,fill=fillColor] (201.60, 92.89) circle (  1.96);

\path[draw=drawColor,line width= 0.4pt,line join=round,line cap=round,fill=fillColor] (237.24, 88.67) circle (  1.96);

\path[draw=drawColor,line width= 0.4pt,line join=round,line cap=round,fill=fillColor] (272.89, 84.25) circle (  1.96);

\path[draw=drawColor,line width= 0.4pt,line join=round,line cap=round,fill=fillColor] ( 59.02,148.29) circle (  1.96);

\path[draw=drawColor,line width= 0.4pt,line join=round,line cap=round,fill=fillColor] ( 94.66,110.78) circle (  1.96);

\path[draw=drawColor,line width= 0.4pt,line join=round,line cap=round,fill=fillColor] (130.31, 96.37) circle (  1.96);

\path[draw=drawColor,line width= 0.4pt,line join=round,line cap=round,fill=fillColor] (165.95, 86.45) circle (  1.96);

\path[draw=drawColor,line width= 0.4pt,line join=round,line cap=round,fill=fillColor] (201.60, 78.70) circle (  1.96);

\path[draw=drawColor,line width= 0.4pt,line join=round,line cap=round,fill=fillColor] (237.24, 75.80) circle (  1.96);

\path[draw=drawColor,line width= 0.4pt,line join=round,line cap=round,fill=fillColor] (272.89, 72.60) circle (  1.96);

\path[draw=drawColor,line width= 0.4pt,line join=round,line cap=round,fill=fillColor] ( 59.02,136.60) circle (  1.96);

\path[draw=drawColor,line width= 0.4pt,line join=round,line cap=round,fill=fillColor] ( 94.66,102.21) circle (  1.96);

\path[draw=drawColor,line width= 0.4pt,line join=round,line cap=round,fill=fillColor] (130.31, 84.05) circle (  1.96);

\path[draw=drawColor,line width= 0.4pt,line join=round,line cap=round,fill=fillColor] (165.95, 78.16) circle (  1.96);

\path[draw=drawColor,line width= 0.4pt,line join=round,line cap=round,fill=fillColor] (201.60, 74.03) circle (  1.96);

\path[draw=drawColor,line width= 0.4pt,line join=round,line cap=round,fill=fillColor] (237.24, 70.23) circle (  1.96);

\path[draw=drawColor,line width= 0.4pt,line join=round,line cap=round,fill=fillColor] (272.89, 67.71) circle (  1.96);
\end{scope}
\begin{scope}
\path[clip] (  0.00,  0.00) rectangle (289.08,289.08);
\definecolor{drawColor}{gray}{0.30}

\node[text=drawColor,anchor=base east,inner sep=0pt, outer sep=0pt, scale=  1.40] at ( 43.37, 45.07) {0};

\node[text=drawColor,anchor=base east,inner sep=0pt, outer sep=0pt, scale=  1.40] at ( 43.37, 96.03) {25};

\node[text=drawColor,anchor=base east,inner sep=0pt, outer sep=0pt, scale=  1.40] at ( 43.37,147.00) {50};

\node[text=drawColor,anchor=base east,inner sep=0pt, outer sep=0pt, scale=  1.40] at ( 43.37,197.97) {75};

\node[text=drawColor,anchor=base east,inner sep=0pt, outer sep=0pt, scale=  1.40] at ( 43.37,248.94) {100};
\end{scope}
\begin{scope}
\path[clip] (  0.00,  0.00) rectangle (289.08,289.08);
\definecolor{drawColor}{gray}{0.20}

\path[draw=drawColor,line width= 0.6pt,line join=round] ( 45.57, 49.89) --
	( 48.32, 49.89);

\path[draw=drawColor,line width= 0.6pt,line join=round] ( 45.57,100.85) --
	( 48.32,100.85);

\path[draw=drawColor,line width= 0.6pt,line join=round] ( 45.57,151.82) --
	( 48.32,151.82);

\path[draw=drawColor,line width= 0.6pt,line join=round] ( 45.57,202.79) --
	( 48.32,202.79);

\path[draw=drawColor,line width= 0.6pt,line join=round] ( 45.57,253.76) --
	( 48.32,253.76);
\end{scope}
\begin{scope}
\path[clip] (  0.00,  0.00) rectangle (289.08,289.08);
\definecolor{drawColor}{gray}{0.20}

\path[draw=drawColor,line width= 0.6pt,line join=round] ( 94.66, 36.94) --
	( 94.66, 39.69);

\path[draw=drawColor,line width= 0.6pt,line join=round] (165.95, 36.94) --
	(165.95, 39.69);

\path[draw=drawColor,line width= 0.6pt,line join=round] (237.24, 36.94) --
	(237.24, 39.69);
\end{scope}
\begin{scope}
\path[clip] (  0.00,  0.00) rectangle (289.08,289.08);
\definecolor{drawColor}{gray}{0.30}

\node[text=drawColor,anchor=base,inner sep=0pt, outer sep=0pt, scale=  1.40] at ( 94.66, 25.10) {100};

\node[text=drawColor,anchor=base,inner sep=0pt, outer sep=0pt, scale=  1.40] at (165.95, 25.10) {200};

\node[text=drawColor,anchor=base,inner sep=0pt, outer sep=0pt, scale=  1.40] at (237.24, 25.10) {300};
\end{scope}
\begin{scope}
\path[clip] (  0.00,  0.00) rectangle (289.08,289.08);
\definecolor{drawColor}{RGB}{0,0,0}

\node[text=drawColor,anchor=base,inner sep=0pt, outer sep=0pt, scale=  1.60] at (165.95,  8.61) {\itshape $t$};
\end{scope}
\begin{scope}
\path[clip] (  0.00,  0.00) rectangle (289.08,289.08);
\definecolor{drawColor}{RGB}{0,0,0}

\node[text=drawColor,rotate= 90.00,anchor=base,inner sep=0pt, outer sep=0pt, scale=  1.60] at ( 16.52,151.82) {remaining taxa $\left[\%\right]$};
\end{scope}
\begin{scope}
\path[clip] (  0.00,  0.00) rectangle (289.08,289.08);
\definecolor{drawColor}{RGB}{0,0,0}

\node[text=drawColor,anchor=base,inner sep=0pt, outer sep=0pt, scale=  1.60] at (165.95,272.56) {(a) $k=5$};
\end{scope}
\end{tikzpicture}}
\end{minipage}
\begin{minipage}[t]{0.33\textwidth}
\scalebox{0.48}{% Created by tikzDevice version 0.12.3.1 on 2021-08-04 18:45:06
% !TEX encoding = UTF-8 Unicode
\begin{tikzpicture}[x=1pt,y=1pt]
\definecolor{fillColor}{RGB}{255,255,255}
\path[use as bounding box,fill=fillColor,fill opacity=0.00] (0,0) rectangle (289.08,289.08);
\begin{scope}
\path[clip] (  0.00,  0.00) rectangle (289.08,289.08);
\definecolor{drawColor}{RGB}{255,255,255}
\definecolor{fillColor}{RGB}{255,255,255}

\path[draw=drawColor,line width= 0.6pt,line join=round,line cap=round,fill=fillColor] (  0.00,  0.00) rectangle (289.08,289.08);
\end{scope}
\begin{scope}
\path[clip] ( 48.32, 39.69) rectangle (283.58,263.95);
\definecolor{fillColor}{gray}{0.92}

\path[fill=fillColor] ( 48.32, 39.69) rectangle (283.58,263.95);
\definecolor{drawColor}{RGB}{255,255,255}

\path[draw=drawColor,line width= 0.3pt,line join=round] ( 48.32, 75.37) --
	(283.58, 75.37);

\path[draw=drawColor,line width= 0.3pt,line join=round] ( 48.32,126.34) --
	(283.58,126.34);

\path[draw=drawColor,line width= 0.3pt,line join=round] ( 48.32,177.31) --
	(283.58,177.31);

\path[draw=drawColor,line width= 0.3pt,line join=round] ( 48.32,228.27) --
	(283.58,228.27);

\path[draw=drawColor,line width= 0.3pt,line join=round] ( 59.02, 39.69) --
	( 59.02,263.95);

\path[draw=drawColor,line width= 0.3pt,line join=round] (130.31, 39.69) --
	(130.31,263.95);

\path[draw=drawColor,line width= 0.3pt,line join=round] (201.60, 39.69) --
	(201.60,263.95);

\path[draw=drawColor,line width= 0.3pt,line join=round] (272.89, 39.69) --
	(272.89,263.95);

\path[draw=drawColor,line width= 0.6pt,line join=round] ( 48.32, 49.89) --
	(283.58, 49.89);

\path[draw=drawColor,line width= 0.6pt,line join=round] ( 48.32,100.85) --
	(283.58,100.85);

\path[draw=drawColor,line width= 0.6pt,line join=round] ( 48.32,151.82) --
	(283.58,151.82);

\path[draw=drawColor,line width= 0.6pt,line join=round] ( 48.32,202.79) --
	(283.58,202.79);

\path[draw=drawColor,line width= 0.6pt,line join=round] ( 48.32,253.76) --
	(283.58,253.76);

\path[draw=drawColor,line width= 0.6pt,line join=round] ( 94.66, 39.69) --
	( 94.66,263.95);

\path[draw=drawColor,line width= 0.6pt,line join=round] (165.95, 39.69) --
	(165.95,263.95);

\path[draw=drawColor,line width= 0.6pt,line join=round] (237.24, 39.69) --
	(237.24,263.95);
\definecolor{drawColor}{RGB}{0,0,0}

\path[draw=drawColor,line width= 0.6pt,line join=round] ( 59.02,189.61) --
	( 94.66,161.88) --
	(130.31,139.68) --
	(165.95,128.92) --
	(201.60,117.73) --
	(237.24,111.36) --
	(272.89,103.86);

\path[draw=drawColor,line width= 0.6pt,dash pattern=on 1pt off 3pt ,line join=round] ( 59.02,184.71) --
	( 94.66,149.24) --
	(130.31,124.73) --
	(165.95,111.46) --
	(201.60,101.43) --
	(237.24, 94.69) --
	(272.89, 88.10);

\path[draw=drawColor,line width= 0.6pt,dash pattern=on 4pt off 4pt ,line join=round] ( 59.02,177.92) --
	( 94.66,139.73) --
	(130.31,116.03) --
	(165.95,101.40) --
	(201.60, 93.60) --
	(237.24, 85.04) --
	(272.89, 82.47);
\definecolor{fillColor}{RGB}{0,0,0}

\path[draw=drawColor,line width= 0.4pt,line join=round,line cap=round,fill=fillColor] ( 59.02,189.61) circle (  1.96);

\path[draw=drawColor,line width= 0.4pt,line join=round,line cap=round,fill=fillColor] ( 94.66,161.88) circle (  1.96);

\path[draw=drawColor,line width= 0.4pt,line join=round,line cap=round,fill=fillColor] (130.31,139.68) circle (  1.96);

\path[draw=drawColor,line width= 0.4pt,line join=round,line cap=round,fill=fillColor] (165.95,128.92) circle (  1.96);

\path[draw=drawColor,line width= 0.4pt,line join=round,line cap=round,fill=fillColor] (201.60,117.73) circle (  1.96);

\path[draw=drawColor,line width= 0.4pt,line join=round,line cap=round,fill=fillColor] (237.24,111.36) circle (  1.96);

\path[draw=drawColor,line width= 0.4pt,line join=round,line cap=round,fill=fillColor] (272.89,103.86) circle (  1.96);

\path[draw=drawColor,line width= 0.4pt,line join=round,line cap=round,fill=fillColor] ( 59.02,184.71) circle (  1.96);

\path[draw=drawColor,line width= 0.4pt,line join=round,line cap=round,fill=fillColor] ( 94.66,149.24) circle (  1.96);

\path[draw=drawColor,line width= 0.4pt,line join=round,line cap=round,fill=fillColor] (130.31,124.73) circle (  1.96);

\path[draw=drawColor,line width= 0.4pt,line join=round,line cap=round,fill=fillColor] (165.95,111.46) circle (  1.96);

\path[draw=drawColor,line width= 0.4pt,line join=round,line cap=round,fill=fillColor] (201.60,101.43) circle (  1.96);

\path[draw=drawColor,line width= 0.4pt,line join=round,line cap=round,fill=fillColor] (237.24, 94.69) circle (  1.96);

\path[draw=drawColor,line width= 0.4pt,line join=round,line cap=round,fill=fillColor] (272.89, 88.10) circle (  1.96);

\path[draw=drawColor,line width= 0.4pt,line join=round,line cap=round,fill=fillColor] ( 59.02,177.92) circle (  1.96);

\path[draw=drawColor,line width= 0.4pt,line join=round,line cap=round,fill=fillColor] ( 94.66,139.73) circle (  1.96);

\path[draw=drawColor,line width= 0.4pt,line join=round,line cap=round,fill=fillColor] (130.31,116.03) circle (  1.96);

\path[draw=drawColor,line width= 0.4pt,line join=round,line cap=round,fill=fillColor] (165.95,101.40) circle (  1.96);

\path[draw=drawColor,line width= 0.4pt,line join=round,line cap=round,fill=fillColor] (201.60, 93.60) circle (  1.96);

\path[draw=drawColor,line width= 0.4pt,line join=round,line cap=round,fill=fillColor] (237.24, 85.04) circle (  1.96);

\path[draw=drawColor,line width= 0.4pt,line join=round,line cap=round,fill=fillColor] (272.89, 82.47) circle (  1.96);
\end{scope}
\begin{scope}
\path[clip] (  0.00,  0.00) rectangle (289.08,289.08);
\definecolor{drawColor}{gray}{0.30}

\node[text=drawColor,anchor=base east,inner sep=0pt, outer sep=0pt, scale=  1.40] at ( 43.37, 45.07) {0};

\node[text=drawColor,anchor=base east,inner sep=0pt, outer sep=0pt, scale=  1.40] at ( 43.37, 96.03) {25};

\node[text=drawColor,anchor=base east,inner sep=0pt, outer sep=0pt, scale=  1.40] at ( 43.37,147.00) {50};

\node[text=drawColor,anchor=base east,inner sep=0pt, outer sep=0pt, scale=  1.40] at ( 43.37,197.97) {75};

\node[text=drawColor,anchor=base east,inner sep=0pt, outer sep=0pt, scale=  1.40] at ( 43.37,248.94) {100};
\end{scope}
\begin{scope}
\path[clip] (  0.00,  0.00) rectangle (289.08,289.08);
\definecolor{drawColor}{gray}{0.20}

\path[draw=drawColor,line width= 0.6pt,line join=round] ( 45.57, 49.89) --
	( 48.32, 49.89);

\path[draw=drawColor,line width= 0.6pt,line join=round] ( 45.57,100.85) --
	( 48.32,100.85);

\path[draw=drawColor,line width= 0.6pt,line join=round] ( 45.57,151.82) --
	( 48.32,151.82);

\path[draw=drawColor,line width= 0.6pt,line join=round] ( 45.57,202.79) --
	( 48.32,202.79);

\path[draw=drawColor,line width= 0.6pt,line join=round] ( 45.57,253.76) --
	( 48.32,253.76);
\end{scope}
\begin{scope}
\path[clip] (  0.00,  0.00) rectangle (289.08,289.08);
\definecolor{drawColor}{gray}{0.20}

\path[draw=drawColor,line width= 0.6pt,line join=round] ( 94.66, 36.94) --
	( 94.66, 39.69);

\path[draw=drawColor,line width= 0.6pt,line join=round] (165.95, 36.94) --
	(165.95, 39.69);

\path[draw=drawColor,line width= 0.6pt,line join=round] (237.24, 36.94) --
	(237.24, 39.69);
\end{scope}
\begin{scope}
\path[clip] (  0.00,  0.00) rectangle (289.08,289.08);
\definecolor{drawColor}{gray}{0.30}

\node[text=drawColor,anchor=base,inner sep=0pt, outer sep=0pt, scale=  1.40] at ( 94.66, 25.10) {100};

\node[text=drawColor,anchor=base,inner sep=0pt, outer sep=0pt, scale=  1.40] at (165.95, 25.10) {200};

\node[text=drawColor,anchor=base,inner sep=0pt, outer sep=0pt, scale=  1.40] at (237.24, 25.10) {300};
\end{scope}
\begin{scope}
\path[clip] (  0.00,  0.00) rectangle (289.08,289.08);
\definecolor{drawColor}{RGB}{0,0,0}

\node[text=drawColor,anchor=base,inner sep=0pt, outer sep=0pt, scale=  1.60] at (165.95,  8.61) {\itshape $t$};
\end{scope}
\begin{scope}
\path[clip] (  0.00,  0.00) rectangle (289.08,289.08);
\definecolor{drawColor}{RGB}{0,0,0}

\node[text=drawColor,rotate= 90.00,anchor=base,inner sep=0pt, outer sep=0pt, scale=  1.60] at ( 16.52,151.82) {remaining taxa $\left[\%\right]$};
\end{scope}
\begin{scope}
\path[clip] (  0.00,  0.00) rectangle (289.08,289.08);
\definecolor{drawColor}{RGB}{0,0,0}

\node[text=drawColor,anchor=base,inner sep=0pt, outer sep=0pt, scale=  1.60] at (165.95,272.56) {(b) $k=10$};
\end{scope}
\end{tikzpicture}}
\end{minipage}
\begin{minipage}[t]{0.33\textwidth}
\scalebox{0.48}{% Created by tikzDevice version 0.12.3.1 on 2021-08-04 18:45:06
% !TEX encoding = UTF-8 Unicode
\begin{tikzpicture}[x=1pt,y=1pt]
\definecolor{fillColor}{RGB}{255,255,255}
\path[use as bounding box,fill=fillColor,fill opacity=0.00] (0,0) rectangle (289.08,289.08);
\begin{scope}
\path[clip] (  0.00,  0.00) rectangle (289.08,289.08);
\definecolor{drawColor}{RGB}{255,255,255}
\definecolor{fillColor}{RGB}{255,255,255}

\path[draw=drawColor,line width= 0.6pt,line join=round,line cap=round,fill=fillColor] (  0.00,  0.00) rectangle (289.08,289.08);
\end{scope}
\begin{scope}
\path[clip] ( 48.32, 39.69) rectangle (283.58,263.95);
\definecolor{fillColor}{gray}{0.92}

\path[fill=fillColor] ( 48.32, 39.69) rectangle (283.58,263.95);
\definecolor{drawColor}{RGB}{255,255,255}

\path[draw=drawColor,line width= 0.3pt,line join=round] ( 48.32, 75.37) --
	(283.58, 75.37);

\path[draw=drawColor,line width= 0.3pt,line join=round] ( 48.32,126.34) --
	(283.58,126.34);

\path[draw=drawColor,line width= 0.3pt,line join=round] ( 48.32,177.31) --
	(283.58,177.31);

\path[draw=drawColor,line width= 0.3pt,line join=round] ( 48.32,228.27) --
	(283.58,228.27);

\path[draw=drawColor,line width= 0.3pt,line join=round] ( 59.02, 39.69) --
	( 59.02,263.95);

\path[draw=drawColor,line width= 0.3pt,line join=round] (130.31, 39.69) --
	(130.31,263.95);

\path[draw=drawColor,line width= 0.3pt,line join=round] (201.60, 39.69) --
	(201.60,263.95);

\path[draw=drawColor,line width= 0.3pt,line join=round] (272.89, 39.69) --
	(272.89,263.95);

\path[draw=drawColor,line width= 0.6pt,line join=round] ( 48.32, 49.89) --
	(283.58, 49.89);

\path[draw=drawColor,line width= 0.6pt,line join=round] ( 48.32,100.85) --
	(283.58,100.85);

\path[draw=drawColor,line width= 0.6pt,line join=round] ( 48.32,151.82) --
	(283.58,151.82);

\path[draw=drawColor,line width= 0.6pt,line join=round] ( 48.32,202.79) --
	(283.58,202.79);

\path[draw=drawColor,line width= 0.6pt,line join=round] ( 48.32,253.76) --
	(283.58,253.76);

\path[draw=drawColor,line width= 0.6pt,line join=round] ( 94.66, 39.69) --
	( 94.66,263.95);

\path[draw=drawColor,line width= 0.6pt,line join=round] (165.95, 39.69) --
	(165.95,263.95);

\path[draw=drawColor,line width= 0.6pt,line join=round] (237.24, 39.69) --
	(237.24,263.95);
\definecolor{drawColor}{RGB}{0,0,0}

\path[draw=drawColor,line width= 0.6pt,line join=round] ( 59.02,210.81) --
	( 94.66,176.01) --
	(130.31,158.07) --
	(165.95,144.62) --
	(201.60,134.32) --
	(237.24,127.86) --
	(272.89,119.32);

\path[draw=drawColor,line width= 0.6pt,dash pattern=on 1pt off 3pt ,line join=round] ( 59.02,207.00) --
	( 94.66,167.18) --
	(130.31,146.29) --
	(165.95,129.46) --
	(201.60,118.39) --
	(237.24,110.59) --
	(272.89,102.74);

\path[draw=drawColor,line width= 0.6pt,dash pattern=on 4pt off 4pt ,line join=round] ( 59.02,203.20) --
	( 94.66,156.31) --
	(130.31,137.32) --
	(165.95,120.49) --
	(201.60,105.83) --
	(237.24,103.75) --
	(272.89, 94.19);
\definecolor{fillColor}{RGB}{0,0,0}

\path[draw=drawColor,line width= 0.4pt,line join=round,line cap=round,fill=fillColor] ( 59.02,210.81) circle (  1.96);

\path[draw=drawColor,line width= 0.4pt,line join=round,line cap=round,fill=fillColor] ( 94.66,176.01) circle (  1.96);

\path[draw=drawColor,line width= 0.4pt,line join=round,line cap=round,fill=fillColor] (130.31,158.07) circle (  1.96);

\path[draw=drawColor,line width= 0.4pt,line join=round,line cap=round,fill=fillColor] (165.95,144.62) circle (  1.96);

\path[draw=drawColor,line width= 0.4pt,line join=round,line cap=round,fill=fillColor] (201.60,134.32) circle (  1.96);

\path[draw=drawColor,line width= 0.4pt,line join=round,line cap=round,fill=fillColor] (237.24,127.86) circle (  1.96);

\path[draw=drawColor,line width= 0.4pt,line join=round,line cap=round,fill=fillColor] (272.89,119.32) circle (  1.96);

\path[draw=drawColor,line width= 0.4pt,line join=round,line cap=round,fill=fillColor] ( 59.02,207.00) circle (  1.96);

\path[draw=drawColor,line width= 0.4pt,line join=round,line cap=round,fill=fillColor] ( 94.66,167.18) circle (  1.96);

\path[draw=drawColor,line width= 0.4pt,line join=round,line cap=round,fill=fillColor] (130.31,146.29) circle (  1.96);

\path[draw=drawColor,line width= 0.4pt,line join=round,line cap=round,fill=fillColor] (165.95,129.46) circle (  1.96);

\path[draw=drawColor,line width= 0.4pt,line join=round,line cap=round,fill=fillColor] (201.60,118.39) circle (  1.96);

\path[draw=drawColor,line width= 0.4pt,line join=round,line cap=round,fill=fillColor] (237.24,110.59) circle (  1.96);

\path[draw=drawColor,line width= 0.4pt,line join=round,line cap=round,fill=fillColor] (272.89,102.74) circle (  1.96);

\path[draw=drawColor,line width= 0.4pt,line join=round,line cap=round,fill=fillColor] ( 59.02,203.20) circle (  1.96);

\path[draw=drawColor,line width= 0.4pt,line join=round,line cap=round,fill=fillColor] ( 94.66,156.31) circle (  1.96);

\path[draw=drawColor,line width= 0.4pt,line join=round,line cap=round,fill=fillColor] (130.31,137.32) circle (  1.96);

\path[draw=drawColor,line width= 0.4pt,line join=round,line cap=round,fill=fillColor] (165.95,120.49) circle (  1.96);

\path[draw=drawColor,line width= 0.4pt,line join=round,line cap=round,fill=fillColor] (201.60,105.83) circle (  1.96);

\path[draw=drawColor,line width= 0.4pt,line join=round,line cap=round,fill=fillColor] (237.24,103.75) circle (  1.96);

\path[draw=drawColor,line width= 0.4pt,line join=round,line cap=round,fill=fillColor] (272.89, 94.19) circle (  1.96);
\end{scope}
\begin{scope}
\path[clip] (  0.00,  0.00) rectangle (289.08,289.08);
\definecolor{drawColor}{gray}{0.30}

\node[text=drawColor,anchor=base east,inner sep=0pt, outer sep=0pt, scale=  1.40] at ( 43.37, 45.07) {0};

\node[text=drawColor,anchor=base east,inner sep=0pt, outer sep=0pt, scale=  1.40] at ( 43.37, 96.03) {25};

\node[text=drawColor,anchor=base east,inner sep=0pt, outer sep=0pt, scale=  1.40] at ( 43.37,147.00) {50};

\node[text=drawColor,anchor=base east,inner sep=0pt, outer sep=0pt, scale=  1.40] at ( 43.37,197.97) {75};

\node[text=drawColor,anchor=base east,inner sep=0pt, outer sep=0pt, scale=  1.40] at ( 43.37,248.94) {100};
\end{scope}
\begin{scope}
\path[clip] (  0.00,  0.00) rectangle (289.08,289.08);
\definecolor{drawColor}{gray}{0.20}

\path[draw=drawColor,line width= 0.6pt,line join=round] ( 45.57, 49.89) --
	( 48.32, 49.89);

\path[draw=drawColor,line width= 0.6pt,line join=round] ( 45.57,100.85) --
	( 48.32,100.85);

\path[draw=drawColor,line width= 0.6pt,line join=round] ( 45.57,151.82) --
	( 48.32,151.82);

\path[draw=drawColor,line width= 0.6pt,line join=round] ( 45.57,202.79) --
	( 48.32,202.79);

\path[draw=drawColor,line width= 0.6pt,line join=round] ( 45.57,253.76) --
	( 48.32,253.76);
\end{scope}
\begin{scope}
\path[clip] (  0.00,  0.00) rectangle (289.08,289.08);
\definecolor{drawColor}{gray}{0.20}

\path[draw=drawColor,line width= 0.6pt,line join=round] ( 94.66, 36.94) --
	( 94.66, 39.69);

\path[draw=drawColor,line width= 0.6pt,line join=round] (165.95, 36.94) --
	(165.95, 39.69);

\path[draw=drawColor,line width= 0.6pt,line join=round] (237.24, 36.94) --
	(237.24, 39.69);
\end{scope}
\begin{scope}
\path[clip] (  0.00,  0.00) rectangle (289.08,289.08);
\definecolor{drawColor}{gray}{0.30}

\node[text=drawColor,anchor=base,inner sep=0pt, outer sep=0pt, scale=  1.40] at ( 94.66, 25.10) {100};

\node[text=drawColor,anchor=base,inner sep=0pt, outer sep=0pt, scale=  1.40] at (165.95, 25.10) {200};

\node[text=drawColor,anchor=base,inner sep=0pt, outer sep=0pt, scale=  1.40] at (237.24, 25.10) {300};
\end{scope}
\begin{scope}
\path[clip] (  0.00,  0.00) rectangle (289.08,289.08);
\definecolor{drawColor}{RGB}{0,0,0}

\node[text=drawColor,anchor=base,inner sep=0pt, outer sep=0pt, scale=  1.60] at (165.95,  8.61) {\itshape $t$};
\end{scope}
\begin{scope}
\path[clip] (  0.00,  0.00) rectangle (289.08,289.08);
\definecolor{drawColor}{RGB}{0,0,0}

\node[text=drawColor,rotate= 90.00,anchor=base,inner sep=0pt, outer sep=0pt, scale=  1.60] at ( 16.52,151.82) {remaining taxa $\left[\%\right]$};
\end{scope}
\begin{scope}
\path[clip] (  0.00,  0.00) rectangle (289.08,289.08);
\definecolor{drawColor}{RGB}{0,0,0}

\node[text=drawColor,anchor=base,inner sep=0pt, outer sep=0pt, scale=  1.60] at (165.95,272.56) {(c) $k=15$};
\end{scope}
\end{tikzpicture}}
\end{minipage}
\begin{minipage}[t]{0.33\textwidth}
\scalebox{0.48}{% Created by tikzDevice version 0.12.3.1 on 2021-08-04 18:45:06
% !TEX encoding = UTF-8 Unicode
\begin{tikzpicture}[x=1pt,y=1pt]
\definecolor{fillColor}{RGB}{255,255,255}
\path[use as bounding box,fill=fillColor,fill opacity=0.00] (0,0) rectangle (289.08,289.08);
\begin{scope}
\path[clip] (  0.00,  0.00) rectangle (289.08,289.08);
\definecolor{drawColor}{RGB}{255,255,255}
\definecolor{fillColor}{RGB}{255,255,255}

\path[draw=drawColor,line width= 0.6pt,line join=round,line cap=round,fill=fillColor] (  0.00,  0.00) rectangle (289.08,289.08);
\end{scope}
\begin{scope}
\path[clip] ( 48.32, 39.69) rectangle (283.58,263.95);
\definecolor{fillColor}{gray}{0.92}

\path[fill=fillColor] ( 48.32, 39.69) rectangle (283.58,263.95);
\definecolor{drawColor}{RGB}{255,255,255}

\path[draw=drawColor,line width= 0.3pt,line join=round] ( 48.32, 75.37) --
	(283.58, 75.37);

\path[draw=drawColor,line width= 0.3pt,line join=round] ( 48.32,126.34) --
	(283.58,126.34);

\path[draw=drawColor,line width= 0.3pt,line join=round] ( 48.32,177.31) --
	(283.58,177.31);

\path[draw=drawColor,line width= 0.3pt,line join=round] ( 48.32,228.27) --
	(283.58,228.27);

\path[draw=drawColor,line width= 0.3pt,line join=round] ( 59.02, 39.69) --
	( 59.02,263.95);

\path[draw=drawColor,line width= 0.3pt,line join=round] (130.31, 39.69) --
	(130.31,263.95);

\path[draw=drawColor,line width= 0.3pt,line join=round] (201.60, 39.69) --
	(201.60,263.95);

\path[draw=drawColor,line width= 0.3pt,line join=round] (272.89, 39.69) --
	(272.89,263.95);

\path[draw=drawColor,line width= 0.6pt,line join=round] ( 48.32, 49.89) --
	(283.58, 49.89);

\path[draw=drawColor,line width= 0.6pt,line join=round] ( 48.32,100.85) --
	(283.58,100.85);

\path[draw=drawColor,line width= 0.6pt,line join=round] ( 48.32,151.82) --
	(283.58,151.82);

\path[draw=drawColor,line width= 0.6pt,line join=round] ( 48.32,202.79) --
	(283.58,202.79);

\path[draw=drawColor,line width= 0.6pt,line join=round] ( 48.32,253.76) --
	(283.58,253.76);

\path[draw=drawColor,line width= 0.6pt,line join=round] ( 94.66, 39.69) --
	( 94.66,263.95);

\path[draw=drawColor,line width= 0.6pt,line join=round] (165.95, 39.69) --
	(165.95,263.95);

\path[draw=drawColor,line width= 0.6pt,line join=round] (237.24, 39.69) --
	(237.24,263.95);
\definecolor{drawColor}{RGB}{0,0,0}

\path[draw=drawColor,line width= 0.6pt,line join=round] ( 59.02,220.05) --
	( 94.66,190.56) --
	(130.31,175.38) --
	(165.95,157.87) --
	(201.60,146.82) --
	(237.24,137.23) --
	(272.89,132.17);

\path[draw=drawColor,line width= 0.6pt,dash pattern=on 1pt off 3pt ,line join=round] ( 59.02,217.60) --
	( 94.66,182.81) --
	(130.31,164.33) --
	(165.95,144.82) --
	(201.60,131.76) --
	(237.24,122.24) --
	(272.89,114.47);

\path[draw=drawColor,line width= 0.6pt,dash pattern=on 4pt off 4pt ,line join=round] ( 59.02,213.80) --
	( 94.66,171.26) --
	(130.31,157.53) --
	(165.95,137.89) --
	(201.60,126.27) --
	(237.24,112.68) --
	(272.89,104.33);
\definecolor{fillColor}{RGB}{0,0,0}

\path[draw=drawColor,line width= 0.4pt,line join=round,line cap=round,fill=fillColor] ( 59.02,220.05) circle (  1.96);

\path[draw=drawColor,line width= 0.4pt,line join=round,line cap=round,fill=fillColor] ( 94.66,190.56) circle (  1.96);

\path[draw=drawColor,line width= 0.4pt,line join=round,line cap=round,fill=fillColor] (130.31,175.38) circle (  1.96);

\path[draw=drawColor,line width= 0.4pt,line join=round,line cap=round,fill=fillColor] (165.95,157.87) circle (  1.96);

\path[draw=drawColor,line width= 0.4pt,line join=round,line cap=round,fill=fillColor] (201.60,146.82) circle (  1.96);

\path[draw=drawColor,line width= 0.4pt,line join=round,line cap=round,fill=fillColor] (237.24,137.23) circle (  1.96);

\path[draw=drawColor,line width= 0.4pt,line join=round,line cap=round,fill=fillColor] (272.89,132.17) circle (  1.96);

\path[draw=drawColor,line width= 0.4pt,line join=round,line cap=round,fill=fillColor] ( 59.02,217.60) circle (  1.96);

\path[draw=drawColor,line width= 0.4pt,line join=round,line cap=round,fill=fillColor] ( 94.66,182.81) circle (  1.96);

\path[draw=drawColor,line width= 0.4pt,line join=round,line cap=round,fill=fillColor] (130.31,164.33) circle (  1.96);

\path[draw=drawColor,line width= 0.4pt,line join=round,line cap=round,fill=fillColor] (165.95,144.82) circle (  1.96);

\path[draw=drawColor,line width= 0.4pt,line join=round,line cap=round,fill=fillColor] (201.60,131.76) circle (  1.96);

\path[draw=drawColor,line width= 0.4pt,line join=round,line cap=round,fill=fillColor] (237.24,122.24) circle (  1.96);

\path[draw=drawColor,line width= 0.4pt,line join=round,line cap=round,fill=fillColor] (272.89,114.47) circle (  1.96);

\path[draw=drawColor,line width= 0.4pt,line join=round,line cap=round,fill=fillColor] ( 59.02,213.80) circle (  1.96);

\path[draw=drawColor,line width= 0.4pt,line join=round,line cap=round,fill=fillColor] ( 94.66,171.26) circle (  1.96);

\path[draw=drawColor,line width= 0.4pt,line join=round,line cap=round,fill=fillColor] (130.31,157.53) circle (  1.96);

\path[draw=drawColor,line width= 0.4pt,line join=round,line cap=round,fill=fillColor] (165.95,137.89) circle (  1.96);

\path[draw=drawColor,line width= 0.4pt,line join=round,line cap=round,fill=fillColor] (201.60,126.27) circle (  1.96);

\path[draw=drawColor,line width= 0.4pt,line join=round,line cap=round,fill=fillColor] (237.24,112.68) circle (  1.96);

\path[draw=drawColor,line width= 0.4pt,line join=round,line cap=round,fill=fillColor] (272.89,104.33) circle (  1.96);
\end{scope}
\begin{scope}
\path[clip] (  0.00,  0.00) rectangle (289.08,289.08);
\definecolor{drawColor}{gray}{0.30}

\node[text=drawColor,anchor=base east,inner sep=0pt, outer sep=0pt, scale=  1.40] at ( 43.37, 45.07) {0};

\node[text=drawColor,anchor=base east,inner sep=0pt, outer sep=0pt, scale=  1.40] at ( 43.37, 96.03) {25};

\node[text=drawColor,anchor=base east,inner sep=0pt, outer sep=0pt, scale=  1.40] at ( 43.37,147.00) {50};

\node[text=drawColor,anchor=base east,inner sep=0pt, outer sep=0pt, scale=  1.40] at ( 43.37,197.97) {75};

\node[text=drawColor,anchor=base east,inner sep=0pt, outer sep=0pt, scale=  1.40] at ( 43.37,248.94) {100};
\end{scope}
\begin{scope}
\path[clip] (  0.00,  0.00) rectangle (289.08,289.08);
\definecolor{drawColor}{gray}{0.20}

\path[draw=drawColor,line width= 0.6pt,line join=round] ( 45.57, 49.89) --
	( 48.32, 49.89);

\path[draw=drawColor,line width= 0.6pt,line join=round] ( 45.57,100.85) --
	( 48.32,100.85);

\path[draw=drawColor,line width= 0.6pt,line join=round] ( 45.57,151.82) --
	( 48.32,151.82);

\path[draw=drawColor,line width= 0.6pt,line join=round] ( 45.57,202.79) --
	( 48.32,202.79);

\path[draw=drawColor,line width= 0.6pt,line join=round] ( 45.57,253.76) --
	( 48.32,253.76);
\end{scope}
\begin{scope}
\path[clip] (  0.00,  0.00) rectangle (289.08,289.08);
\definecolor{drawColor}{gray}{0.20}

\path[draw=drawColor,line width= 0.6pt,line join=round] ( 94.66, 36.94) --
	( 94.66, 39.69);

\path[draw=drawColor,line width= 0.6pt,line join=round] (165.95, 36.94) --
	(165.95, 39.69);

\path[draw=drawColor,line width= 0.6pt,line join=round] (237.24, 36.94) --
	(237.24, 39.69);
\end{scope}
\begin{scope}
\path[clip] (  0.00,  0.00) rectangle (289.08,289.08);
\definecolor{drawColor}{gray}{0.30}

\node[text=drawColor,anchor=base,inner sep=0pt, outer sep=0pt, scale=  1.40] at ( 94.66, 25.10) {100};

\node[text=drawColor,anchor=base,inner sep=0pt, outer sep=0pt, scale=  1.40] at (165.95, 25.10) {200};

\node[text=drawColor,anchor=base,inner sep=0pt, outer sep=0pt, scale=  1.40] at (237.24, 25.10) {300};
\end{scope}
\begin{scope}
\path[clip] (  0.00,  0.00) rectangle (289.08,289.08);
\definecolor{drawColor}{RGB}{0,0,0}

\node[text=drawColor,anchor=base,inner sep=0pt, outer sep=0pt, scale=  1.60] at (165.95,  8.61) {\itshape $t$};
\end{scope}
\begin{scope}
\path[clip] (  0.00,  0.00) rectangle (289.08,289.08);
\definecolor{drawColor}{RGB}{0,0,0}

\node[text=drawColor,rotate= 90.00,anchor=base,inner sep=0pt, outer sep=0pt, scale=  1.60] at ( 16.52,151.82) {remaining taxa $\left[\%\right]$};
\end{scope}
\begin{scope}
\path[clip] (  0.00,  0.00) rectangle (289.08,289.08);
\definecolor{drawColor}{RGB}{0,0,0}

\node[text=drawColor,anchor=base,inner sep=0pt, outer sep=0pt, scale=  1.60] at (165.95,272.56) {(d) $k=20$};
\end{scope}
\end{tikzpicture}}
\end{minipage}
\begin{minipage}[t]{0.33\textwidth}
\scalebox{0.48}{% Created by tikzDevice version 0.12.3.1 on 2021-08-04 18:45:06
% !TEX encoding = UTF-8 Unicode
\begin{tikzpicture}[x=1pt,y=1pt]
\definecolor{fillColor}{RGB}{255,255,255}
\path[use as bounding box,fill=fillColor,fill opacity=0.00] (0,0) rectangle (289.08,289.08);
\begin{scope}
\path[clip] (  0.00,  0.00) rectangle (289.08,289.08);
\definecolor{drawColor}{RGB}{255,255,255}
\definecolor{fillColor}{RGB}{255,255,255}

\path[draw=drawColor,line width= 0.6pt,line join=round,line cap=round,fill=fillColor] (  0.00,  0.00) rectangle (289.08,289.08);
\end{scope}
\begin{scope}
\path[clip] ( 48.32, 39.69) rectangle (283.58,263.95);
\definecolor{fillColor}{gray}{0.92}

\path[fill=fillColor] ( 48.32, 39.69) rectangle (283.58,263.95);
\definecolor{drawColor}{RGB}{255,255,255}

\path[draw=drawColor,line width= 0.3pt,line join=round] ( 48.32, 75.37) --
	(283.58, 75.37);

\path[draw=drawColor,line width= 0.3pt,line join=round] ( 48.32,126.34) --
	(283.58,126.34);

\path[draw=drawColor,line width= 0.3pt,line join=round] ( 48.32,177.31) --
	(283.58,177.31);

\path[draw=drawColor,line width= 0.3pt,line join=round] ( 48.32,228.27) --
	(283.58,228.27);

\path[draw=drawColor,line width= 0.3pt,line join=round] ( 59.02, 39.69) --
	( 59.02,263.95);

\path[draw=drawColor,line width= 0.3pt,line join=round] (130.31, 39.69) --
	(130.31,263.95);

\path[draw=drawColor,line width= 0.3pt,line join=round] (201.60, 39.69) --
	(201.60,263.95);

\path[draw=drawColor,line width= 0.3pt,line join=round] (272.89, 39.69) --
	(272.89,263.95);

\path[draw=drawColor,line width= 0.6pt,line join=round] ( 48.32, 49.89) --
	(283.58, 49.89);

\path[draw=drawColor,line width= 0.6pt,line join=round] ( 48.32,100.85) --
	(283.58,100.85);

\path[draw=drawColor,line width= 0.6pt,line join=round] ( 48.32,151.82) --
	(283.58,151.82);

\path[draw=drawColor,line width= 0.6pt,line join=round] ( 48.32,202.79) --
	(283.58,202.79);

\path[draw=drawColor,line width= 0.6pt,line join=round] ( 48.32,253.76) --
	(283.58,253.76);

\path[draw=drawColor,line width= 0.6pt,line join=round] ( 94.66, 39.69) --
	( 94.66,263.95);

\path[draw=drawColor,line width= 0.6pt,line join=round] (165.95, 39.69) --
	(165.95,263.95);

\path[draw=drawColor,line width= 0.6pt,line join=round] (237.24, 39.69) --
	(237.24,263.95);
\definecolor{drawColor}{RGB}{0,0,0}

\path[draw=drawColor,line width= 0.6pt,line join=round] ( 59.02,227.39) --
	( 94.66,200.75) --
	(130.31,180.73) --
	(165.95,166.30) --
	(201.60,158.40) --
	(237.24,150.64) --
	(272.89,139.08);

\path[draw=drawColor,line width= 0.6pt,dash pattern=on 1pt off 3pt ,line join=round] ( 59.02,226.03) --
	( 94.66,196.67) --
	(130.31,172.66) --
	(165.95,156.44) --
	(201.60,144.97) --
	(237.24,137.05) --
	(272.89,124.41);

\path[draw=drawColor,line width= 0.6pt,dash pattern=on 4pt off 4pt ,line join=round] ( 59.02,221.14) --
	( 94.66,188.25) --
	(130.31,163.06) --
	(165.95,151.21) --
	(201.60,134.53) --
	(237.24,127.31) --
	(272.89,116.60);
\definecolor{fillColor}{RGB}{0,0,0}

\path[draw=drawColor,line width= 0.4pt,line join=round,line cap=round,fill=fillColor] ( 59.02,227.39) circle (  1.96);

\path[draw=drawColor,line width= 0.4pt,line join=round,line cap=round,fill=fillColor] ( 94.66,200.75) circle (  1.96);

\path[draw=drawColor,line width= 0.4pt,line join=round,line cap=round,fill=fillColor] (130.31,180.73) circle (  1.96);

\path[draw=drawColor,line width= 0.4pt,line join=round,line cap=round,fill=fillColor] (165.95,166.30) circle (  1.96);

\path[draw=drawColor,line width= 0.4pt,line join=round,line cap=round,fill=fillColor] (201.60,158.40) circle (  1.96);

\path[draw=drawColor,line width= 0.4pt,line join=round,line cap=round,fill=fillColor] (237.24,150.64) circle (  1.96);

\path[draw=drawColor,line width= 0.4pt,line join=round,line cap=round,fill=fillColor] (272.89,139.08) circle (  1.96);

\path[draw=drawColor,line width= 0.4pt,line join=round,line cap=round,fill=fillColor] ( 59.02,226.03) circle (  1.96);

\path[draw=drawColor,line width= 0.4pt,line join=round,line cap=round,fill=fillColor] ( 94.66,196.67) circle (  1.96);

\path[draw=drawColor,line width= 0.4pt,line join=round,line cap=round,fill=fillColor] (130.31,172.66) circle (  1.96);

\path[draw=drawColor,line width= 0.4pt,line join=round,line cap=round,fill=fillColor] (165.95,156.44) circle (  1.96);

\path[draw=drawColor,line width= 0.4pt,line join=round,line cap=round,fill=fillColor] (201.60,144.97) circle (  1.96);

\path[draw=drawColor,line width= 0.4pt,line join=round,line cap=round,fill=fillColor] (237.24,137.05) circle (  1.96);

\path[draw=drawColor,line width= 0.4pt,line join=round,line cap=round,fill=fillColor] (272.89,124.41) circle (  1.96);

\path[draw=drawColor,line width= 0.4pt,line join=round,line cap=round,fill=fillColor] ( 59.02,221.14) circle (  1.96);

\path[draw=drawColor,line width= 0.4pt,line join=round,line cap=round,fill=fillColor] ( 94.66,188.25) circle (  1.96);

\path[draw=drawColor,line width= 0.4pt,line join=round,line cap=round,fill=fillColor] (130.31,163.06) circle (  1.96);

\path[draw=drawColor,line width= 0.4pt,line join=round,line cap=round,fill=fillColor] (165.95,151.21) circle (  1.96);

\path[draw=drawColor,line width= 0.4pt,line join=round,line cap=round,fill=fillColor] (201.60,134.53) circle (  1.96);

\path[draw=drawColor,line width= 0.4pt,line join=round,line cap=round,fill=fillColor] (237.24,127.31) circle (  1.96);

\path[draw=drawColor,line width= 0.4pt,line join=round,line cap=round,fill=fillColor] (272.89,116.60) circle (  1.96);
\end{scope}
\begin{scope}
\path[clip] (  0.00,  0.00) rectangle (289.08,289.08);
\definecolor{drawColor}{gray}{0.30}

\node[text=drawColor,anchor=base east,inner sep=0pt, outer sep=0pt, scale=  1.40] at ( 43.37, 45.07) {0};

\node[text=drawColor,anchor=base east,inner sep=0pt, outer sep=0pt, scale=  1.40] at ( 43.37, 96.03) {25};

\node[text=drawColor,anchor=base east,inner sep=0pt, outer sep=0pt, scale=  1.40] at ( 43.37,147.00) {50};

\node[text=drawColor,anchor=base east,inner sep=0pt, outer sep=0pt, scale=  1.40] at ( 43.37,197.97) {75};

\node[text=drawColor,anchor=base east,inner sep=0pt, outer sep=0pt, scale=  1.40] at ( 43.37,248.94) {100};
\end{scope}
\begin{scope}
\path[clip] (  0.00,  0.00) rectangle (289.08,289.08);
\definecolor{drawColor}{gray}{0.20}

\path[draw=drawColor,line width= 0.6pt,line join=round] ( 45.57, 49.89) --
	( 48.32, 49.89);

\path[draw=drawColor,line width= 0.6pt,line join=round] ( 45.57,100.85) --
	( 48.32,100.85);

\path[draw=drawColor,line width= 0.6pt,line join=round] ( 45.57,151.82) --
	( 48.32,151.82);

\path[draw=drawColor,line width= 0.6pt,line join=round] ( 45.57,202.79) --
	( 48.32,202.79);

\path[draw=drawColor,line width= 0.6pt,line join=round] ( 45.57,253.76) --
	( 48.32,253.76);
\end{scope}
\begin{scope}
\path[clip] (  0.00,  0.00) rectangle (289.08,289.08);
\definecolor{drawColor}{gray}{0.20}

\path[draw=drawColor,line width= 0.6pt,line join=round] ( 94.66, 36.94) --
	( 94.66, 39.69);

\path[draw=drawColor,line width= 0.6pt,line join=round] (165.95, 36.94) --
	(165.95, 39.69);

\path[draw=drawColor,line width= 0.6pt,line join=round] (237.24, 36.94) --
	(237.24, 39.69);
\end{scope}
\begin{scope}
\path[clip] (  0.00,  0.00) rectangle (289.08,289.08);
\definecolor{drawColor}{gray}{0.30}

\node[text=drawColor,anchor=base,inner sep=0pt, outer sep=0pt, scale=  1.40] at ( 94.66, 25.10) {100};

\node[text=drawColor,anchor=base,inner sep=0pt, outer sep=0pt, scale=  1.40] at (165.95, 25.10) {200};

\node[text=drawColor,anchor=base,inner sep=0pt, outer sep=0pt, scale=  1.40] at (237.24, 25.10) {300};
\end{scope}
\begin{scope}
\path[clip] (  0.00,  0.00) rectangle (289.08,289.08);
\definecolor{drawColor}{RGB}{0,0,0}

\node[text=drawColor,anchor=base,inner sep=0pt, outer sep=0pt, scale=  1.60] at (165.95,  8.61) {\itshape $t$};
\end{scope}
\begin{scope}
\path[clip] (  0.00,  0.00) rectangle (289.08,289.08);
\definecolor{drawColor}{RGB}{0,0,0}

\node[text=drawColor,rotate= 90.00,anchor=base,inner sep=0pt, outer sep=0pt, scale=  1.60] at ( 16.52,151.82) {remaining taxa $\left[\%\right]$};
\end{scope}
\begin{scope}
\path[clip] (  0.00,  0.00) rectangle (289.08,289.08);
\definecolor{drawColor}{RGB}{0,0,0}

\node[text=drawColor,anchor=base,inner sep=0pt, outer sep=0pt, scale=  1.60] at (165.95,272.56) {(e) $k=25$};
\end{scope}
\end{tikzpicture}}
\end{minipage}
\begin{minipage}[t]{0.33\textwidth}
\scalebox{0.48}{% Created by tikzDevice version 0.12.3.1 on 2021-08-04 18:45:06
% !TEX encoding = UTF-8 Unicode
\begin{tikzpicture}[x=1pt,y=1pt]
\definecolor{fillColor}{RGB}{255,255,255}
\path[use as bounding box,fill=fillColor,fill opacity=0.00] (0,0) rectangle (289.08,289.08);
\begin{scope}
\path[clip] (  0.00,  0.00) rectangle (289.08,289.08);
\definecolor{drawColor}{RGB}{255,255,255}
\definecolor{fillColor}{RGB}{255,255,255}

\path[draw=drawColor,line width= 0.6pt,line join=round,line cap=round,fill=fillColor] (  0.00,  0.00) rectangle (289.08,289.08);
\end{scope}
\begin{scope}
\path[clip] ( 48.32, 39.69) rectangle (283.58,263.95);
\definecolor{fillColor}{gray}{0.92}

\path[fill=fillColor] ( 48.32, 39.69) rectangle (283.58,263.95);
\definecolor{drawColor}{RGB}{255,255,255}

\path[draw=drawColor,line width= 0.3pt,line join=round] ( 48.32, 75.37) --
	(283.58, 75.37);

\path[draw=drawColor,line width= 0.3pt,line join=round] ( 48.32,126.34) --
	(283.58,126.34);

\path[draw=drawColor,line width= 0.3pt,line join=round] ( 48.32,177.31) --
	(283.58,177.31);

\path[draw=drawColor,line width= 0.3pt,line join=round] ( 48.32,228.27) --
	(283.58,228.27);

\path[draw=drawColor,line width= 0.3pt,line join=round] ( 59.02, 39.69) --
	( 59.02,263.95);

\path[draw=drawColor,line width= 0.3pt,line join=round] (130.31, 39.69) --
	(130.31,263.95);

\path[draw=drawColor,line width= 0.3pt,line join=round] (201.60, 39.69) --
	(201.60,263.95);

\path[draw=drawColor,line width= 0.3pt,line join=round] (272.89, 39.69) --
	(272.89,263.95);

\path[draw=drawColor,line width= 0.6pt,line join=round] ( 48.32, 49.89) --
	(283.58, 49.89);

\path[draw=drawColor,line width= 0.6pt,line join=round] ( 48.32,100.85) --
	(283.58,100.85);

\path[draw=drawColor,line width= 0.6pt,line join=round] ( 48.32,151.82) --
	(283.58,151.82);

\path[draw=drawColor,line width= 0.6pt,line join=round] ( 48.32,202.79) --
	(283.58,202.79);

\path[draw=drawColor,line width= 0.6pt,line join=round] ( 48.32,253.76) --
	(283.58,253.76);

\path[draw=drawColor,line width= 0.6pt,line join=round] ( 94.66, 39.69) --
	( 94.66,263.95);

\path[draw=drawColor,line width= 0.6pt,line join=round] (165.95, 39.69) --
	(165.95,263.95);

\path[draw=drawColor,line width= 0.6pt,line join=round] (237.24, 39.69) --
	(237.24,263.95);
\definecolor{drawColor}{RGB}{0,0,0}

\path[draw=drawColor,line width= 0.6pt,line join=round] ( 59.02,233.37) --
	( 94.66,206.73) --
	(130.31,186.62) --
	(165.95,176.49) --
	(201.60,171.56) --
	(237.24,155.54) --
	(272.89,152.56);

\path[draw=drawColor,line width= 0.6pt,dash pattern=on 1pt off 3pt ,line join=round] ( 59.02,233.37) --
	( 94.66,203.20) --
	(130.31,179.73) --
	(165.95,168.34) --
	(201.60,157.69) --
	(237.24,143.12) --
	(272.89,138.70);

\path[draw=drawColor,line width= 0.6pt,dash pattern=on 4pt off 4pt ,line join=round] ( 59.02,233.10) --
	( 94.66,196.13) --
	(130.31,172.30) --
	(165.95,160.93) --
	(201.60,147.64) --
	(237.24,135.74) --
	(272.89,126.74);
\definecolor{fillColor}{RGB}{0,0,0}

\path[draw=drawColor,line width= 0.4pt,line join=round,line cap=round,fill=fillColor] ( 59.02,233.37) circle (  1.96);

\path[draw=drawColor,line width= 0.4pt,line join=round,line cap=round,fill=fillColor] ( 94.66,206.73) circle (  1.96);

\path[draw=drawColor,line width= 0.4pt,line join=round,line cap=round,fill=fillColor] (130.31,186.62) circle (  1.96);

\path[draw=drawColor,line width= 0.4pt,line join=round,line cap=round,fill=fillColor] (165.95,176.49) circle (  1.96);

\path[draw=drawColor,line width= 0.4pt,line join=round,line cap=round,fill=fillColor] (201.60,171.56) circle (  1.96);

\path[draw=drawColor,line width= 0.4pt,line join=round,line cap=round,fill=fillColor] (237.24,155.54) circle (  1.96);

\path[draw=drawColor,line width= 0.4pt,line join=round,line cap=round,fill=fillColor] (272.89,152.56) circle (  1.96);

\path[draw=drawColor,line width= 0.4pt,line join=round,line cap=round,fill=fillColor] ( 59.02,233.37) circle (  1.96);

\path[draw=drawColor,line width= 0.4pt,line join=round,line cap=round,fill=fillColor] ( 94.66,203.20) circle (  1.96);

\path[draw=drawColor,line width= 0.4pt,line join=round,line cap=round,fill=fillColor] (130.31,179.73) circle (  1.96);

\path[draw=drawColor,line width= 0.4pt,line join=round,line cap=round,fill=fillColor] (165.95,168.34) circle (  1.96);

\path[draw=drawColor,line width= 0.4pt,line join=round,line cap=round,fill=fillColor] (201.60,157.69) circle (  1.96);

\path[draw=drawColor,line width= 0.4pt,line join=round,line cap=round,fill=fillColor] (237.24,143.12) circle (  1.96);

\path[draw=drawColor,line width= 0.4pt,line join=round,line cap=round,fill=fillColor] (272.89,138.70) circle (  1.96);

\path[draw=drawColor,line width= 0.4pt,line join=round,line cap=round,fill=fillColor] ( 59.02,233.10) circle (  1.96);

\path[draw=drawColor,line width= 0.4pt,line join=round,line cap=round,fill=fillColor] ( 94.66,196.13) circle (  1.96);

\path[draw=drawColor,line width= 0.4pt,line join=round,line cap=round,fill=fillColor] (130.31,172.30) circle (  1.96);

\path[draw=drawColor,line width= 0.4pt,line join=round,line cap=round,fill=fillColor] (165.95,160.93) circle (  1.96);

\path[draw=drawColor,line width= 0.4pt,line join=round,line cap=round,fill=fillColor] (201.60,147.64) circle (  1.96);

\path[draw=drawColor,line width= 0.4pt,line join=round,line cap=round,fill=fillColor] (237.24,135.74) circle (  1.96);

\path[draw=drawColor,line width= 0.4pt,line join=round,line cap=round,fill=fillColor] (272.89,126.74) circle (  1.96);
\end{scope}
\begin{scope}
\path[clip] (  0.00,  0.00) rectangle (289.08,289.08);
\definecolor{drawColor}{gray}{0.30}

\node[text=drawColor,anchor=base east,inner sep=0pt, outer sep=0pt, scale=  1.40] at ( 43.37, 45.07) {0};

\node[text=drawColor,anchor=base east,inner sep=0pt, outer sep=0pt, scale=  1.40] at ( 43.37, 96.03) {25};

\node[text=drawColor,anchor=base east,inner sep=0pt, outer sep=0pt, scale=  1.40] at ( 43.37,147.00) {50};

\node[text=drawColor,anchor=base east,inner sep=0pt, outer sep=0pt, scale=  1.40] at ( 43.37,197.97) {75};

\node[text=drawColor,anchor=base east,inner sep=0pt, outer sep=0pt, scale=  1.40] at ( 43.37,248.94) {100};
\end{scope}
\begin{scope}
\path[clip] (  0.00,  0.00) rectangle (289.08,289.08);
\definecolor{drawColor}{gray}{0.20}

\path[draw=drawColor,line width= 0.6pt,line join=round] ( 45.57, 49.89) --
	( 48.32, 49.89);

\path[draw=drawColor,line width= 0.6pt,line join=round] ( 45.57,100.85) --
	( 48.32,100.85);

\path[draw=drawColor,line width= 0.6pt,line join=round] ( 45.57,151.82) --
	( 48.32,151.82);

\path[draw=drawColor,line width= 0.6pt,line join=round] ( 45.57,202.79) --
	( 48.32,202.79);

\path[draw=drawColor,line width= 0.6pt,line join=round] ( 45.57,253.76) --
	( 48.32,253.76);
\end{scope}
\begin{scope}
\path[clip] (  0.00,  0.00) rectangle (289.08,289.08);
\definecolor{drawColor}{gray}{0.20}

\path[draw=drawColor,line width= 0.6pt,line join=round] ( 94.66, 36.94) --
	( 94.66, 39.69);

\path[draw=drawColor,line width= 0.6pt,line join=round] (165.95, 36.94) --
	(165.95, 39.69);

\path[draw=drawColor,line width= 0.6pt,line join=round] (237.24, 36.94) --
	(237.24, 39.69);
\end{scope}
\begin{scope}
\path[clip] (  0.00,  0.00) rectangle (289.08,289.08);
\definecolor{drawColor}{gray}{0.30}

\node[text=drawColor,anchor=base,inner sep=0pt, outer sep=0pt, scale=  1.40] at ( 94.66, 25.10) {100};

\node[text=drawColor,anchor=base,inner sep=0pt, outer sep=0pt, scale=  1.40] at (165.95, 25.10) {200};

\node[text=drawColor,anchor=base,inner sep=0pt, outer sep=0pt, scale=  1.40] at (237.24, 25.10) {300};
\end{scope}
\begin{scope}
\path[clip] (  0.00,  0.00) rectangle (289.08,289.08);
\definecolor{drawColor}{RGB}{0,0,0}

\node[text=drawColor,anchor=base,inner sep=0pt, outer sep=0pt, scale=  1.60] at (165.95,  8.61) {\itshape $t$};
\end{scope}
\begin{scope}
\path[clip] (  0.00,  0.00) rectangle (289.08,289.08);
\definecolor{drawColor}{RGB}{0,0,0}

\node[text=drawColor,rotate= 90.00,anchor=base,inner sep=0pt, outer sep=0pt, scale=  1.60] at ( 16.52,151.82) {remaining taxa $\left[\%\right]$};
\end{scope}
\begin{scope}
\path[clip] (  0.00,  0.00) rectangle (289.08,289.08);
\definecolor{drawColor}{RGB}{0,0,0}

\node[text=drawColor,anchor=base,inner sep=0pt, outer sep=0pt, scale=  1.60] at (165.95,272.56) {(f) $k=30$};
\end{scope}
\end{tikzpicture}}
\end{minipage}
\begin{minipage}[t]{0.33\textwidth}
\scalebox{0.48}{% Created by tikzDevice version 0.12.3.1 on 2021-08-04 18:45:07
% !TEX encoding = UTF-8 Unicode
\begin{tikzpicture}[x=1pt,y=1pt]
\definecolor{fillColor}{RGB}{255,255,255}
\path[use as bounding box,fill=fillColor,fill opacity=0.00] (0,0) rectangle (289.08,289.08);
\begin{scope}
\path[clip] (  0.00,  0.00) rectangle (289.08,289.08);
\definecolor{drawColor}{RGB}{255,255,255}
\definecolor{fillColor}{RGB}{255,255,255}

\path[draw=drawColor,line width= 0.6pt,line join=round,line cap=round,fill=fillColor] (  0.00,  0.00) rectangle (289.08,289.08);
\end{scope}
\begin{scope}
\path[clip] ( 48.32, 39.69) rectangle (283.58,263.95);
\definecolor{fillColor}{gray}{0.92}

\path[fill=fillColor] ( 48.32, 39.69) rectangle (283.58,263.95);
\definecolor{drawColor}{RGB}{255,255,255}

\path[draw=drawColor,line width= 0.3pt,line join=round] ( 48.32, 75.37) --
	(283.58, 75.37);

\path[draw=drawColor,line width= 0.3pt,line join=round] ( 48.32,126.34) --
	(283.58,126.34);

\path[draw=drawColor,line width= 0.3pt,line join=round] ( 48.32,177.31) --
	(283.58,177.31);

\path[draw=drawColor,line width= 0.3pt,line join=round] ( 48.32,228.27) --
	(283.58,228.27);

\path[draw=drawColor,line width= 0.3pt,line join=round] ( 59.02, 39.69) --
	( 59.02,263.95);

\path[draw=drawColor,line width= 0.3pt,line join=round] (130.31, 39.69) --
	(130.31,263.95);

\path[draw=drawColor,line width= 0.3pt,line join=round] (201.60, 39.69) --
	(201.60,263.95);

\path[draw=drawColor,line width= 0.3pt,line join=round] (272.89, 39.69) --
	(272.89,263.95);

\path[draw=drawColor,line width= 0.6pt,line join=round] ( 48.32, 49.89) --
	(283.58, 49.89);

\path[draw=drawColor,line width= 0.6pt,line join=round] ( 48.32,100.85) --
	(283.58,100.85);

\path[draw=drawColor,line width= 0.6pt,line join=round] ( 48.32,151.82) --
	(283.58,151.82);

\path[draw=drawColor,line width= 0.6pt,line join=round] ( 48.32,202.79) --
	(283.58,202.79);

\path[draw=drawColor,line width= 0.6pt,line join=round] ( 48.32,253.76) --
	(283.58,253.76);

\path[draw=drawColor,line width= 0.6pt,line join=round] ( 94.66, 39.69) --
	( 94.66,263.95);

\path[draw=drawColor,line width= 0.6pt,line join=round] (165.95, 39.69) --
	(165.95,263.95);

\path[draw=drawColor,line width= 0.6pt,line join=round] (237.24, 39.69) --
	(237.24,263.95);
\definecolor{drawColor}{RGB}{0,0,0}

\path[draw=drawColor,line width= 0.6pt,line join=round] ( 59.02,241.25) --
	( 94.66,217.33) --
	(130.31,196.13) --
	(165.95,189.54) --
	(201.60,172.70) --
	(237.24,165.55) --
	(272.89,155.55);

\path[draw=drawColor,line width= 0.6pt,dash pattern=on 1pt off 3pt ,line join=round] ( 59.02,241.25) --
	( 94.66,215.29) --
	(130.31,191.33) --
	(165.95,180.84) --
	(201.60,164.00) --
	(237.24,155.13) --
	(272.89,142.74);

\path[draw=drawColor,line width= 0.6pt,dash pattern=on 4pt off 4pt ,line join=round] ( 59.02,240.44) --
	( 94.66,211.76) --
	(130.31,182.63) --
	(165.95,171.39) --
	(201.60,155.14) --
	(237.24,145.89) --
	(272.89,133.42);
\definecolor{fillColor}{RGB}{0,0,0}

\path[draw=drawColor,line width= 0.4pt,line join=round,line cap=round,fill=fillColor] ( 59.02,241.25) circle (  1.96);

\path[draw=drawColor,line width= 0.4pt,line join=round,line cap=round,fill=fillColor] ( 94.66,217.33) circle (  1.96);

\path[draw=drawColor,line width= 0.4pt,line join=round,line cap=round,fill=fillColor] (130.31,196.13) circle (  1.96);

\path[draw=drawColor,line width= 0.4pt,line join=round,line cap=round,fill=fillColor] (165.95,189.54) circle (  1.96);

\path[draw=drawColor,line width= 0.4pt,line join=round,line cap=round,fill=fillColor] (201.60,172.70) circle (  1.96);

\path[draw=drawColor,line width= 0.4pt,line join=round,line cap=round,fill=fillColor] (237.24,165.55) circle (  1.96);

\path[draw=drawColor,line width= 0.4pt,line join=round,line cap=round,fill=fillColor] (272.89,155.55) circle (  1.96);

\path[draw=drawColor,line width= 0.4pt,line join=round,line cap=round,fill=fillColor] ( 59.02,241.25) circle (  1.96);

\path[draw=drawColor,line width= 0.4pt,line join=round,line cap=round,fill=fillColor] ( 94.66,215.29) circle (  1.96);

\path[draw=drawColor,line width= 0.4pt,line join=round,line cap=round,fill=fillColor] (130.31,191.33) circle (  1.96);

\path[draw=drawColor,line width= 0.4pt,line join=round,line cap=round,fill=fillColor] (165.95,180.84) circle (  1.96);

\path[draw=drawColor,line width= 0.4pt,line join=round,line cap=round,fill=fillColor] (201.60,164.00) circle (  1.96);

\path[draw=drawColor,line width= 0.4pt,line join=round,line cap=round,fill=fillColor] (237.24,155.13) circle (  1.96);

\path[draw=drawColor,line width= 0.4pt,line join=round,line cap=round,fill=fillColor] (272.89,142.74) circle (  1.96);

\path[draw=drawColor,line width= 0.4pt,line join=round,line cap=round,fill=fillColor] ( 59.02,240.44) circle (  1.96);

\path[draw=drawColor,line width= 0.4pt,line join=round,line cap=round,fill=fillColor] ( 94.66,211.76) circle (  1.96);

\path[draw=drawColor,line width= 0.4pt,line join=round,line cap=round,fill=fillColor] (130.31,182.63) circle (  1.96);

\path[draw=drawColor,line width= 0.4pt,line join=round,line cap=round,fill=fillColor] (165.95,171.39) circle (  1.96);

\path[draw=drawColor,line width= 0.4pt,line join=round,line cap=round,fill=fillColor] (201.60,155.14) circle (  1.96);

\path[draw=drawColor,line width= 0.4pt,line join=round,line cap=round,fill=fillColor] (237.24,145.89) circle (  1.96);

\path[draw=drawColor,line width= 0.4pt,line join=round,line cap=round,fill=fillColor] (272.89,133.42) circle (  1.96);
\end{scope}
\begin{scope}
\path[clip] (  0.00,  0.00) rectangle (289.08,289.08);
\definecolor{drawColor}{gray}{0.30}

\node[text=drawColor,anchor=base east,inner sep=0pt, outer sep=0pt, scale=  1.40] at ( 43.37, 45.07) {0};

\node[text=drawColor,anchor=base east,inner sep=0pt, outer sep=0pt, scale=  1.40] at ( 43.37, 96.03) {25};

\node[text=drawColor,anchor=base east,inner sep=0pt, outer sep=0pt, scale=  1.40] at ( 43.37,147.00) {50};

\node[text=drawColor,anchor=base east,inner sep=0pt, outer sep=0pt, scale=  1.40] at ( 43.37,197.97) {75};

\node[text=drawColor,anchor=base east,inner sep=0pt, outer sep=0pt, scale=  1.40] at ( 43.37,248.94) {100};
\end{scope}
\begin{scope}
\path[clip] (  0.00,  0.00) rectangle (289.08,289.08);
\definecolor{drawColor}{gray}{0.20}

\path[draw=drawColor,line width= 0.6pt,line join=round] ( 45.57, 49.89) --
	( 48.32, 49.89);

\path[draw=drawColor,line width= 0.6pt,line join=round] ( 45.57,100.85) --
	( 48.32,100.85);

\path[draw=drawColor,line width= 0.6pt,line join=round] ( 45.57,151.82) --
	( 48.32,151.82);

\path[draw=drawColor,line width= 0.6pt,line join=round] ( 45.57,202.79) --
	( 48.32,202.79);

\path[draw=drawColor,line width= 0.6pt,line join=round] ( 45.57,253.76) --
	( 48.32,253.76);
\end{scope}
\begin{scope}
\path[clip] (  0.00,  0.00) rectangle (289.08,289.08);
\definecolor{drawColor}{gray}{0.20}

\path[draw=drawColor,line width= 0.6pt,line join=round] ( 94.66, 36.94) --
	( 94.66, 39.69);

\path[draw=drawColor,line width= 0.6pt,line join=round] (165.95, 36.94) --
	(165.95, 39.69);

\path[draw=drawColor,line width= 0.6pt,line join=round] (237.24, 36.94) --
	(237.24, 39.69);
\end{scope}
\begin{scope}
\path[clip] (  0.00,  0.00) rectangle (289.08,289.08);
\definecolor{drawColor}{gray}{0.30}

\node[text=drawColor,anchor=base,inner sep=0pt, outer sep=0pt, scale=  1.40] at ( 94.66, 25.10) {100};

\node[text=drawColor,anchor=base,inner sep=0pt, outer sep=0pt, scale=  1.40] at (165.95, 25.10) {200};

\node[text=drawColor,anchor=base,inner sep=0pt, outer sep=0pt, scale=  1.40] at (237.24, 25.10) {300};
\end{scope}
\begin{scope}
\path[clip] (  0.00,  0.00) rectangle (289.08,289.08);
\definecolor{drawColor}{RGB}{0,0,0}

\node[text=drawColor,anchor=base,inner sep=0pt, outer sep=0pt, scale=  1.60] at (165.95,  8.61) {\itshape $t$};
\end{scope}
\begin{scope}
\path[clip] (  0.00,  0.00) rectangle (289.08,289.08);
\definecolor{drawColor}{RGB}{0,0,0}

\node[text=drawColor,rotate= 90.00,anchor=base,inner sep=0pt, outer sep=0pt, scale=  1.60] at ( 16.52,151.82) {remaining taxa $\left[\%\right]$};
\end{scope}
\begin{scope}
\path[clip] (  0.00,  0.00) rectangle (289.08,289.08);
\definecolor{drawColor}{RGB}{0,0,0}

\node[text=drawColor,anchor=base,inner sep=0pt, outer sep=0pt, scale=  1.60] at (165.95,272.56) {(g) $k=35$};
\end{scope}
\end{tikzpicture}}
\end{minipage}
\begin{minipage}[t]{0.66\textwidth}
\scalebox{0.48}{\input{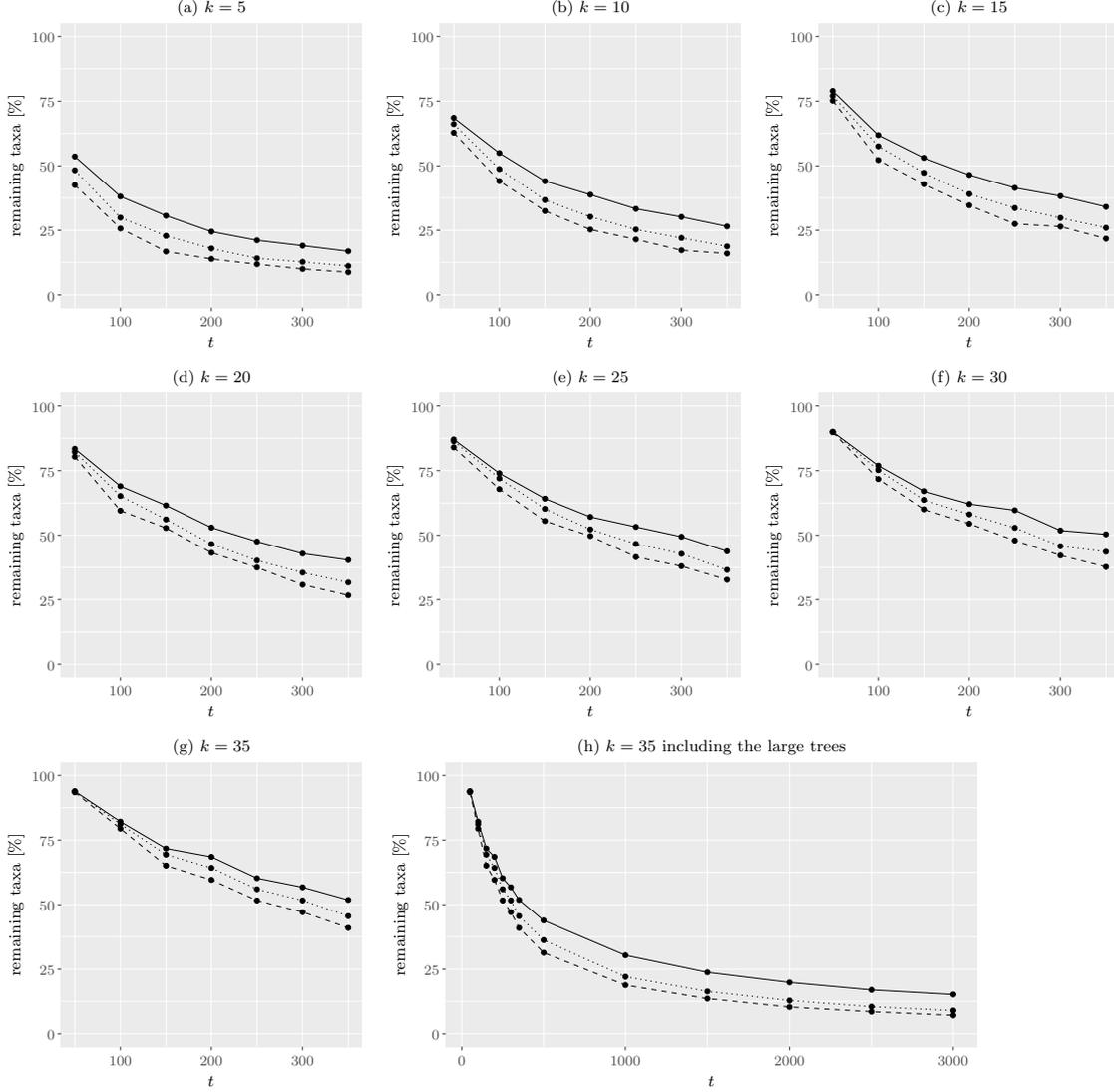}}
\end{minipage}
\hspace{3cm}
%\begin{minipage}[t]{0.09\textwidth}
%\vspace{-3cm}
%\includegraphics[width=\textwidth]{figs/legend1}
%\end{minipage}
\caption{The average percentage of taxa that remains after the subtree reduction (solid line), after the subtree and chain reduction (dotted line), and after all seven reductions (dashed line) \blue{depending on $t$.}}
%: (a) $k=5$, (b) $k=10$, (c) $k=15$, (d) $k=20$, (e) $k=25$, (f) $k=30$, (g) $k=35$, (h) $k=35$ including tree pairs of the large trees dataset.}
\label{fig:reductions}
\end{figure}

\begin{figure}[h!]
\noindent\begin{minipage}[t]{0.33\textwidth}
\scalebox{0.48}{% Created by tikzDevice version 0.12.3.1 on 2021-08-29 17:59:21
% !TEX encoding = UTF-8 Unicode
\begin{tikzpicture}[x=1pt,y=1pt]
\definecolor{fillColor}{RGB}{255,255,255}
\path[use as bounding box,fill=fillColor,fill opacity=0.00] (0,0) rectangle (289.08,289.08);
\begin{scope}
\path[clip] (  0.00,  0.00) rectangle (289.08,289.08);
\definecolor{drawColor}{RGB}{255,255,255}
\definecolor{fillColor}{RGB}{255,255,255}

\path[draw=drawColor,line width= 0.6pt,line join=round,line cap=round,fill=fillColor] (  0.00,  0.00) rectangle (289.08,289.08);
\end{scope}
\begin{scope}
\path[clip] ( 48.32, 39.69) rectangle (283.58,263.95);
\definecolor{fillColor}{gray}{0.92}

\path[fill=fillColor] ( 48.32, 39.69) rectangle (283.58,263.95);
\definecolor{drawColor}{RGB}{255,255,255}

\path[draw=drawColor,line width= 0.3pt,line join=round] ( 48.32, 75.37) --
	(283.58, 75.37);

\path[draw=drawColor,line width= 0.3pt,line join=round] ( 48.32,126.34) --
	(283.58,126.34);

\path[draw=drawColor,line width= 0.3pt,line join=round] ( 48.32,177.31) --
	(283.58,177.31);

\path[draw=drawColor,line width= 0.3pt,line join=round] ( 48.32,228.27) --
	(283.58,228.27);

\path[draw=drawColor,line width= 0.3pt,line join=round] ( 59.02, 39.69) --
	( 59.02,263.95);

\path[draw=drawColor,line width= 0.3pt,line join=round] (130.31, 39.69) --
	(130.31,263.95);

\path[draw=drawColor,line width= 0.3pt,line join=round] (201.60, 39.69) --
	(201.60,263.95);

\path[draw=drawColor,line width= 0.3pt,line join=round] (272.89, 39.69) --
	(272.89,263.95);

\path[draw=drawColor,line width= 0.6pt,line join=round] ( 48.32, 49.89) --
	(283.58, 49.89);

\path[draw=drawColor,line width= 0.6pt,line join=round] ( 48.32,100.85) --
	(283.58,100.85);

\path[draw=drawColor,line width= 0.6pt,line join=round] ( 48.32,151.82) --
	(283.58,151.82);

\path[draw=drawColor,line width= 0.6pt,line join=round] ( 48.32,202.79) --
	(283.58,202.79);

\path[draw=drawColor,line width= 0.6pt,line join=round] ( 48.32,253.76) --
	(283.58,253.76);

\path[draw=drawColor,line width= 0.6pt,line join=round] ( 94.66, 39.69) --
	( 94.66,263.95);

\path[draw=drawColor,line width= 0.6pt,line join=round] (165.95, 39.69) --
	(165.95,263.95);

\path[draw=drawColor,line width= 0.6pt,line join=round] (237.24, 39.69) --
	(237.24,263.95);
\definecolor{drawColor}{RGB}{0,0,0}

\path[draw=drawColor,line width= 0.6pt,line join=round] ( 59.02,208.91) --
	( 94.66,179.96) --
	(130.31,163.24) --
	(165.95,153.86) --
	(201.60,138.28) --
	(237.24,128.99) --
	(272.89,122.93);

\path[draw=drawColor,line width= 0.6pt,dash pattern=on 1pt off 3pt ,line join=round] ( 59.02,208.91) --
	( 94.66,176.29) --
	(130.31,156.99) --
	(165.95,144.69) --
	(201.60,129.97) --
	(237.24,119.47) --
	(272.89,110.23);

\path[draw=drawColor,line width= 0.6pt,dash pattern=on 4pt off 4pt ,line join=round] ( 59.02,204.01) --
	( 94.66,169.76) --
	(130.31,154.54) --
	(165.95,136.74) --
	(201.60,123.61) --
	(237.24,112.68) --
	(272.89,101.61);
\definecolor{fillColor}{RGB}{0,0,0}

\path[draw=drawColor,line width= 0.4pt,line join=round,line cap=round,fill=fillColor] ( 59.02,208.91) circle (  1.96);

\path[draw=drawColor,line width= 0.4pt,line join=round,line cap=round,fill=fillColor] ( 94.66,179.96) circle (  1.96);

\path[draw=drawColor,line width= 0.4pt,line join=round,line cap=round,fill=fillColor] (130.31,163.24) circle (  1.96);

\path[draw=drawColor,line width= 0.4pt,line join=round,line cap=round,fill=fillColor] (165.95,153.86) circle (  1.96);

\path[draw=drawColor,line width= 0.4pt,line join=round,line cap=round,fill=fillColor] (201.60,138.28) circle (  1.96);

\path[draw=drawColor,line width= 0.4pt,line join=round,line cap=round,fill=fillColor] (237.24,128.99) circle (  1.96);

\path[draw=drawColor,line width= 0.4pt,line join=round,line cap=round,fill=fillColor] (272.89,122.93) circle (  1.96);

\path[draw=drawColor,line width= 0.4pt,line join=round,line cap=round,fill=fillColor] ( 59.02,208.91) circle (  1.96);

\path[draw=drawColor,line width= 0.4pt,line join=round,line cap=round,fill=fillColor] ( 94.66,176.29) circle (  1.96);

\path[draw=drawColor,line width= 0.4pt,line join=round,line cap=round,fill=fillColor] (130.31,156.99) circle (  1.96);

\path[draw=drawColor,line width= 0.4pt,line join=round,line cap=round,fill=fillColor] (165.95,144.69) circle (  1.96);

\path[draw=drawColor,line width= 0.4pt,line join=round,line cap=round,fill=fillColor] (201.60,129.97) circle (  1.96);

\path[draw=drawColor,line width= 0.4pt,line join=round,line cap=round,fill=fillColor] (237.24,119.47) circle (  1.96);

\path[draw=drawColor,line width= 0.4pt,line join=round,line cap=round,fill=fillColor] (272.89,110.23) circle (  1.96);

\path[draw=drawColor,line width= 0.4pt,line join=round,line cap=round,fill=fillColor] ( 59.02,204.01) circle (  1.96);

\path[draw=drawColor,line width= 0.4pt,line join=round,line cap=round,fill=fillColor] ( 94.66,169.76) circle (  1.96);

\path[draw=drawColor,line width= 0.4pt,line join=round,line cap=round,fill=fillColor] (130.31,154.54) circle (  1.96);

\path[draw=drawColor,line width= 0.4pt,line join=round,line cap=round,fill=fillColor] (165.95,136.74) circle (  1.96);

\path[draw=drawColor,line width= 0.4pt,line join=round,line cap=round,fill=fillColor] (201.60,123.61) circle (  1.96);

\path[draw=drawColor,line width= 0.4pt,line join=round,line cap=round,fill=fillColor] (237.24,112.68) circle (  1.96);

\path[draw=drawColor,line width= 0.4pt,line join=round,line cap=round,fill=fillColor] (272.89,101.61) circle (  1.96);
\end{scope}
\begin{scope}
\path[clip] (  0.00,  0.00) rectangle (289.08,289.08);
\definecolor{drawColor}{gray}{0.30}

\node[text=drawColor,anchor=base east,inner sep=0pt, outer sep=0pt, scale=  1.40] at ( 43.37, 45.07) {0};

\node[text=drawColor,anchor=base east,inner sep=0pt, outer sep=0pt, scale=  1.40] at ( 43.37, 96.03) {25};

\node[text=drawColor,anchor=base east,inner sep=0pt, outer sep=0pt, scale=  1.40] at ( 43.37,147.00) {50};

\node[text=drawColor,anchor=base east,inner sep=0pt, outer sep=0pt, scale=  1.40] at ( 43.37,197.97) {75};

\node[text=drawColor,anchor=base east,inner sep=0pt, outer sep=0pt, scale=  1.40] at ( 43.37,248.94) {100};
\end{scope}
\begin{scope}
\path[clip] (  0.00,  0.00) rectangle (289.08,289.08);
\definecolor{drawColor}{gray}{0.20}

\path[draw=drawColor,line width= 0.6pt,line join=round] ( 45.57, 49.89) --
	( 48.32, 49.89);

\path[draw=drawColor,line width= 0.6pt,line join=round] ( 45.57,100.85) --
	( 48.32,100.85);

\path[draw=drawColor,line width= 0.6pt,line join=round] ( 45.57,151.82) --
	( 48.32,151.82);

\path[draw=drawColor,line width= 0.6pt,line join=round] ( 45.57,202.79) --
	( 48.32,202.79);

\path[draw=drawColor,line width= 0.6pt,line join=round] ( 45.57,253.76) --
	( 48.32,253.76);
\end{scope}
\begin{scope}
\path[clip] (  0.00,  0.00) rectangle (289.08,289.08);
\definecolor{drawColor}{gray}{0.20}

\path[draw=drawColor,line width= 0.6pt,line join=round] ( 94.66, 36.94) --
	( 94.66, 39.69);

\path[draw=drawColor,line width= 0.6pt,line join=round] (165.95, 36.94) --
	(165.95, 39.69);

\path[draw=drawColor,line width= 0.6pt,line join=round] (237.24, 36.94) --
	(237.24, 39.69);
\end{scope}
\begin{scope}
\path[clip] (  0.00,  0.00) rectangle (289.08,289.08);
\definecolor{drawColor}{gray}{0.30}

\node[text=drawColor,anchor=base,inner sep=0pt, outer sep=0pt, scale=  1.40] at ( 94.66, 25.10) {100};

\node[text=drawColor,anchor=base,inner sep=0pt, outer sep=0pt, scale=  1.40] at (165.95, 25.10) {200};

\node[text=drawColor,anchor=base,inner sep=0pt, outer sep=0pt, scale=  1.40] at (237.24, 25.10) {300};
\end{scope}
\begin{scope}
\path[clip] (  0.00,  0.00) rectangle (289.08,289.08);
\definecolor{drawColor}{RGB}{0,0,0}

\node[text=drawColor,anchor=base,inner sep=0pt, outer sep=0pt, scale=  1.60] at (165.95,  8.61) {\itshape $t$};
\end{scope}
\begin{scope}
\path[clip] (  0.00,  0.00) rectangle (289.08,289.08);
\definecolor{drawColor}{RGB}{0,0,0}

\node[text=drawColor,rotate= 90.00,anchor=base,inner sep=0pt, outer sep=0pt, scale=  1.60] at ( 16.52,151.82) {remaining taxa $\left[\%\right]$};
\end{scope}
\begin{scope}
\path[clip] (  0.00,  0.00) rectangle (289.08,289.08);
\definecolor{drawColor}{RGB}{0,0,0}

\node[text=drawColor,anchor=base,inner sep=0pt, outer sep=0pt, scale=  1.60] at (165.95,272.56) {(a) $s=50$};
\end{scope}
\end{tikzpicture}}
\end{minipage}
\begin{minipage}[t]{0.33\textwidth}
\scalebox{0.48}{% Created by tikzDevice version 0.12.3.1 on 2021-08-29 17:59:21
% !TEX encoding = UTF-8 Unicode
\begin{tikzpicture}[x=1pt,y=1pt]
\definecolor{fillColor}{RGB}{255,255,255}
\path[use as bounding box,fill=fillColor,fill opacity=0.00] (0,0) rectangle (289.08,289.08);
\begin{scope}
\path[clip] (  0.00,  0.00) rectangle (289.08,289.08);
\definecolor{drawColor}{RGB}{255,255,255}
\definecolor{fillColor}{RGB}{255,255,255}

\path[draw=drawColor,line width= 0.6pt,line join=round,line cap=round,fill=fillColor] (  0.00,  0.00) rectangle (289.08,289.08);
\end{scope}
\begin{scope}
\path[clip] ( 48.32, 39.69) rectangle (283.58,263.95);
\definecolor{fillColor}{gray}{0.92}

\path[fill=fillColor] ( 48.32, 39.69) rectangle (283.58,263.95);
\definecolor{drawColor}{RGB}{255,255,255}

\path[draw=drawColor,line width= 0.3pt,line join=round] ( 48.32, 75.37) --
	(283.58, 75.37);

\path[draw=drawColor,line width= 0.3pt,line join=round] ( 48.32,126.34) --
	(283.58,126.34);

\path[draw=drawColor,line width= 0.3pt,line join=round] ( 48.32,177.31) --
	(283.58,177.31);

\path[draw=drawColor,line width= 0.3pt,line join=round] ( 48.32,228.27) --
	(283.58,228.27);

\path[draw=drawColor,line width= 0.3pt,line join=round] ( 59.02, 39.69) --
	( 59.02,263.95);

\path[draw=drawColor,line width= 0.3pt,line join=round] (130.31, 39.69) --
	(130.31,263.95);

\path[draw=drawColor,line width= 0.3pt,line join=round] (201.60, 39.69) --
	(201.60,263.95);

\path[draw=drawColor,line width= 0.3pt,line join=round] (272.89, 39.69) --
	(272.89,263.95);

\path[draw=drawColor,line width= 0.6pt,line join=round] ( 48.32, 49.89) --
	(283.58, 49.89);

\path[draw=drawColor,line width= 0.6pt,line join=round] ( 48.32,100.85) --
	(283.58,100.85);

\path[draw=drawColor,line width= 0.6pt,line join=round] ( 48.32,151.82) --
	(283.58,151.82);

\path[draw=drawColor,line width= 0.6pt,line join=round] ( 48.32,202.79) --
	(283.58,202.79);

\path[draw=drawColor,line width= 0.6pt,line join=round] ( 48.32,253.76) --
	(283.58,253.76);

\path[draw=drawColor,line width= 0.6pt,line join=round] ( 94.66, 39.69) --
	( 94.66,263.95);

\path[draw=drawColor,line width= 0.6pt,line join=round] (165.95, 39.69) --
	(165.95,263.95);

\path[draw=drawColor,line width= 0.6pt,line join=round] (237.24, 39.69) --
	(237.24,263.95);
\definecolor{drawColor}{RGB}{0,0,0}

\path[draw=drawColor,line width= 0.6pt,line join=round] ( 59.02,224.40) --
	( 94.66,192.60) --
	(130.31,172.75) --
	(165.95,148.56) --
	(201.60,139.43) --
	(237.24,131.43) --
	(272.89,127.59);

\path[draw=drawColor,line width= 0.6pt,dash pattern=on 1pt off 3pt ,line join=round] ( 59.02,219.51) --
	( 94.66,186.89) --
	(130.31,164.33) --
	(165.95,139.39) --
	(201.60,128.50) --
	(237.24,120.29) --
	(272.89,114.43);

\path[draw=drawColor,line width= 0.6pt,dash pattern=on 4pt off 4pt ,line join=round] ( 59.02,218.69) --
	( 94.66,177.10) --
	(130.31,155.36) --
	(165.95,134.09) --
	(201.60,124.26) --
	(237.24,111.32) --
	(272.89,102.54);
\definecolor{fillColor}{RGB}{0,0,0}

\path[draw=drawColor,line width= 0.4pt,line join=round,line cap=round,fill=fillColor] ( 59.02,224.40) circle (  1.96);

\path[draw=drawColor,line width= 0.4pt,line join=round,line cap=round,fill=fillColor] ( 94.66,192.60) circle (  1.96);

\path[draw=drawColor,line width= 0.4pt,line join=round,line cap=round,fill=fillColor] (130.31,172.75) circle (  1.96);

\path[draw=drawColor,line width= 0.4pt,line join=round,line cap=round,fill=fillColor] (165.95,148.56) circle (  1.96);

\path[draw=drawColor,line width= 0.4pt,line join=round,line cap=round,fill=fillColor] (201.60,139.43) circle (  1.96);

\path[draw=drawColor,line width= 0.4pt,line join=round,line cap=round,fill=fillColor] (237.24,131.43) circle (  1.96);

\path[draw=drawColor,line width= 0.4pt,line join=round,line cap=round,fill=fillColor] (272.89,127.59) circle (  1.96);

\path[draw=drawColor,line width= 0.4pt,line join=round,line cap=round,fill=fillColor] ( 59.02,219.51) circle (  1.96);

\path[draw=drawColor,line width= 0.4pt,line join=round,line cap=round,fill=fillColor] ( 94.66,186.89) circle (  1.96);

\path[draw=drawColor,line width= 0.4pt,line join=round,line cap=round,fill=fillColor] (130.31,164.33) circle (  1.96);

\path[draw=drawColor,line width= 0.4pt,line join=round,line cap=round,fill=fillColor] (165.95,139.39) circle (  1.96);

\path[draw=drawColor,line width= 0.4pt,line join=round,line cap=round,fill=fillColor] (201.60,128.50) circle (  1.96);

\path[draw=drawColor,line width= 0.4pt,line join=round,line cap=round,fill=fillColor] (237.24,120.29) circle (  1.96);

\path[draw=drawColor,line width= 0.4pt,line join=round,line cap=round,fill=fillColor] (272.89,114.43) circle (  1.96);

\path[draw=drawColor,line width= 0.4pt,line join=round,line cap=round,fill=fillColor] ( 59.02,218.69) circle (  1.96);

\path[draw=drawColor,line width= 0.4pt,line join=round,line cap=round,fill=fillColor] ( 94.66,177.10) circle (  1.96);

\path[draw=drawColor,line width= 0.4pt,line join=round,line cap=round,fill=fillColor] (130.31,155.36) circle (  1.96);

\path[draw=drawColor,line width= 0.4pt,line join=round,line cap=round,fill=fillColor] (165.95,134.09) circle (  1.96);

\path[draw=drawColor,line width= 0.4pt,line join=round,line cap=round,fill=fillColor] (201.60,124.26) circle (  1.96);

\path[draw=drawColor,line width= 0.4pt,line join=round,line cap=round,fill=fillColor] (237.24,111.32) circle (  1.96);

\path[draw=drawColor,line width= 0.4pt,line join=round,line cap=round,fill=fillColor] (272.89,102.54) circle (  1.96);
\end{scope}
\begin{scope}
\path[clip] (  0.00,  0.00) rectangle (289.08,289.08);
\definecolor{drawColor}{gray}{0.30}

\node[text=drawColor,anchor=base east,inner sep=0pt, outer sep=0pt, scale=  1.40] at ( 43.37, 45.07) {0};

\node[text=drawColor,anchor=base east,inner sep=0pt, outer sep=0pt, scale=  1.40] at ( 43.37, 96.03) {25};

\node[text=drawColor,anchor=base east,inner sep=0pt, outer sep=0pt, scale=  1.40] at ( 43.37,147.00) {50};

\node[text=drawColor,anchor=base east,inner sep=0pt, outer sep=0pt, scale=  1.40] at ( 43.37,197.97) {75};

\node[text=drawColor,anchor=base east,inner sep=0pt, outer sep=0pt, scale=  1.40] at ( 43.37,248.94) {100};
\end{scope}
\begin{scope}
\path[clip] (  0.00,  0.00) rectangle (289.08,289.08);
\definecolor{drawColor}{gray}{0.20}

\path[draw=drawColor,line width= 0.6pt,line join=round] ( 45.57, 49.89) --
	( 48.32, 49.89);

\path[draw=drawColor,line width= 0.6pt,line join=round] ( 45.57,100.85) --
	( 48.32,100.85);

\path[draw=drawColor,line width= 0.6pt,line join=round] ( 45.57,151.82) --
	( 48.32,151.82);

\path[draw=drawColor,line width= 0.6pt,line join=round] ( 45.57,202.79) --
	( 48.32,202.79);

\path[draw=drawColor,line width= 0.6pt,line join=round] ( 45.57,253.76) --
	( 48.32,253.76);
\end{scope}
\begin{scope}
\path[clip] (  0.00,  0.00) rectangle (289.08,289.08);
\definecolor{drawColor}{gray}{0.20}

\path[draw=drawColor,line width= 0.6pt,line join=round] ( 94.66, 36.94) --
	( 94.66, 39.69);

\path[draw=drawColor,line width= 0.6pt,line join=round] (165.95, 36.94) --
	(165.95, 39.69);

\path[draw=drawColor,line width= 0.6pt,line join=round] (237.24, 36.94) --
	(237.24, 39.69);
\end{scope}
\begin{scope}
\path[clip] (  0.00,  0.00) rectangle (289.08,289.08);
\definecolor{drawColor}{gray}{0.30}

\node[text=drawColor,anchor=base,inner sep=0pt, outer sep=0pt, scale=  1.40] at ( 94.66, 25.10) {100};

\node[text=drawColor,anchor=base,inner sep=0pt, outer sep=0pt, scale=  1.40] at (165.95, 25.10) {200};

\node[text=drawColor,anchor=base,inner sep=0pt, outer sep=0pt, scale=  1.40] at (237.24, 25.10) {300};
\end{scope}
\begin{scope}
\path[clip] (  0.00,  0.00) rectangle (289.08,289.08);
\definecolor{drawColor}{RGB}{0,0,0}

\node[text=drawColor,anchor=base,inner sep=0pt, outer sep=0pt, scale=  1.60] at (165.95,  8.61) {\itshape $t$};
\end{scope}
\begin{scope}
\path[clip] (  0.00,  0.00) rectangle (289.08,289.08);
\definecolor{drawColor}{RGB}{0,0,0}

\node[text=drawColor,rotate= 90.00,anchor=base,inner sep=0pt, outer sep=0pt, scale=  1.60] at ( 16.52,151.82) {remaining taxa $\left[\%\right]$};
\end{scope}
\begin{scope}
\path[clip] (  0.00,  0.00) rectangle (289.08,289.08);
\definecolor{drawColor}{RGB}{0,0,0}

\node[text=drawColor,anchor=base,inner sep=0pt, outer sep=0pt, scale=  1.60] at (165.95,272.56) {(b) $s=70$};
\end{scope}
\end{tikzpicture}}
\end{minipage}
\begin{minipage}[t]{0.33\textwidth}
\scalebox{0.48}{% Created by tikzDevice version 0.12.3.1 on 2021-08-29 17:59:21
% !TEX encoding = UTF-8 Unicode
\begin{tikzpicture}[x=1pt,y=1pt]
\definecolor{fillColor}{RGB}{255,255,255}
\path[use as bounding box,fill=fillColor,fill opacity=0.00] (0,0) rectangle (289.08,289.08);
\begin{scope}
\path[clip] (  0.00,  0.00) rectangle (289.08,289.08);
\definecolor{drawColor}{RGB}{255,255,255}
\definecolor{fillColor}{RGB}{255,255,255}

\path[draw=drawColor,line width= 0.6pt,line join=round,line cap=round,fill=fillColor] (  0.00,  0.00) rectangle (289.08,289.08);
\end{scope}
\begin{scope}
\path[clip] ( 48.32, 39.69) rectangle (283.58,263.95);
\definecolor{fillColor}{gray}{0.92}

\path[fill=fillColor] ( 48.32, 39.69) rectangle (283.58,263.95);
\definecolor{drawColor}{RGB}{255,255,255}

\path[draw=drawColor,line width= 0.3pt,line join=round] ( 48.32, 75.37) --
	(283.58, 75.37);

\path[draw=drawColor,line width= 0.3pt,line join=round] ( 48.32,126.34) --
	(283.58,126.34);

\path[draw=drawColor,line width= 0.3pt,line join=round] ( 48.32,177.31) --
	(283.58,177.31);

\path[draw=drawColor,line width= 0.3pt,line join=round] ( 48.32,228.27) --
	(283.58,228.27);

\path[draw=drawColor,line width= 0.3pt,line join=round] ( 59.02, 39.69) --
	( 59.02,263.95);

\path[draw=drawColor,line width= 0.3pt,line join=round] (130.31, 39.69) --
	(130.31,263.95);

\path[draw=drawColor,line width= 0.3pt,line join=round] (201.60, 39.69) --
	(201.60,263.95);

\path[draw=drawColor,line width= 0.3pt,line join=round] (272.89, 39.69) --
	(272.89,263.95);

\path[draw=drawColor,line width= 0.6pt,line join=round] ( 48.32, 49.89) --
	(283.58, 49.89);

\path[draw=drawColor,line width= 0.6pt,line join=round] ( 48.32,100.85) --
	(283.58,100.85);

\path[draw=drawColor,line width= 0.6pt,line join=round] ( 48.32,151.82) --
	(283.58,151.82);

\path[draw=drawColor,line width= 0.6pt,line join=round] ( 48.32,202.79) --
	(283.58,202.79);

\path[draw=drawColor,line width= 0.6pt,line join=round] ( 48.32,253.76) --
	(283.58,253.76);

\path[draw=drawColor,line width= 0.6pt,line join=round] ( 94.66, 39.69) --
	( 94.66,263.95);

\path[draw=drawColor,line width= 0.6pt,line join=round] (165.95, 39.69) --
	(165.95,263.95);

\path[draw=drawColor,line width= 0.6pt,line join=round] (237.24, 39.69) --
	(237.24,263.95);
\definecolor{drawColor}{RGB}{0,0,0}

\path[draw=drawColor,line width= 0.6pt,line join=round] ( 59.02,226.85) --
	( 94.66,199.12) --
	(130.31,190.15) --
	(165.95,171.19) --
	(201.60,162.75) --
	(237.24,151.28) --
	(272.89,146.00);

\path[draw=drawColor,line width= 0.6pt,dash pattern=on 1pt off 3pt ,line join=round] ( 59.02,224.40) --
	( 94.66,185.26) --
	(130.31,171.67) --
	(165.95,150.39) --
	(201.60,136.82) --
	(237.24,126.95) --
	(272.89,118.74);

\path[draw=drawColor,line width= 0.6pt,dash pattern=on 4pt off 4pt ,line join=round] ( 59.02,218.69) --
	( 94.66,166.91) --
	(130.31,162.69) --
	(165.95,142.85) --
	(201.60,130.95) --
	(237.24,114.04) --
	(272.89,108.83);
\definecolor{fillColor}{RGB}{0,0,0}

\path[draw=drawColor,line width= 0.4pt,line join=round,line cap=round,fill=fillColor] ( 59.02,226.85) circle (  1.96);

\path[draw=drawColor,line width= 0.4pt,line join=round,line cap=round,fill=fillColor] ( 94.66,199.12) circle (  1.96);

\path[draw=drawColor,line width= 0.4pt,line join=round,line cap=round,fill=fillColor] (130.31,190.15) circle (  1.96);

\path[draw=drawColor,line width= 0.4pt,line join=round,line cap=round,fill=fillColor] (165.95,171.19) circle (  1.96);

\path[draw=drawColor,line width= 0.4pt,line join=round,line cap=round,fill=fillColor] (201.60,162.75) circle (  1.96);

\path[draw=drawColor,line width= 0.4pt,line join=round,line cap=round,fill=fillColor] (237.24,151.28) circle (  1.96);

\path[draw=drawColor,line width= 0.4pt,line join=round,line cap=round,fill=fillColor] (272.89,146.00) circle (  1.96);

\path[draw=drawColor,line width= 0.4pt,line join=round,line cap=round,fill=fillColor] ( 59.02,224.40) circle (  1.96);

\path[draw=drawColor,line width= 0.4pt,line join=round,line cap=round,fill=fillColor] ( 94.66,185.26) circle (  1.96);

\path[draw=drawColor,line width= 0.4pt,line join=round,line cap=round,fill=fillColor] (130.31,171.67) circle (  1.96);

\path[draw=drawColor,line width= 0.4pt,line join=round,line cap=round,fill=fillColor] (165.95,150.39) circle (  1.96);

\path[draw=drawColor,line width= 0.4pt,line join=round,line cap=round,fill=fillColor] (201.60,136.82) circle (  1.96);

\path[draw=drawColor,line width= 0.4pt,line join=round,line cap=round,fill=fillColor] (237.24,126.95) circle (  1.96);

\path[draw=drawColor,line width= 0.4pt,line join=round,line cap=round,fill=fillColor] (272.89,118.74) circle (  1.96);

\path[draw=drawColor,line width= 0.4pt,line join=round,line cap=round,fill=fillColor] ( 59.02,218.69) circle (  1.96);

\path[draw=drawColor,line width= 0.4pt,line join=round,line cap=round,fill=fillColor] ( 94.66,166.91) circle (  1.96);

\path[draw=drawColor,line width= 0.4pt,line join=round,line cap=round,fill=fillColor] (130.31,162.69) circle (  1.96);

\path[draw=drawColor,line width= 0.4pt,line join=round,line cap=round,fill=fillColor] (165.95,142.85) circle (  1.96);

\path[draw=drawColor,line width= 0.4pt,line join=round,line cap=round,fill=fillColor] (201.60,130.95) circle (  1.96);

\path[draw=drawColor,line width= 0.4pt,line join=round,line cap=round,fill=fillColor] (237.24,114.04) circle (  1.96);

\path[draw=drawColor,line width= 0.4pt,line join=round,line cap=round,fill=fillColor] (272.89,108.83) circle (  1.96);
\end{scope}
\begin{scope}
\path[clip] (  0.00,  0.00) rectangle (289.08,289.08);
\definecolor{drawColor}{gray}{0.30}

\node[text=drawColor,anchor=base east,inner sep=0pt, outer sep=0pt, scale=  1.40] at ( 43.37, 45.07) {0};

\node[text=drawColor,anchor=base east,inner sep=0pt, outer sep=0pt, scale=  1.40] at ( 43.37, 96.03) {25};

\node[text=drawColor,anchor=base east,inner sep=0pt, outer sep=0pt, scale=  1.40] at ( 43.37,147.00) {50};

\node[text=drawColor,anchor=base east,inner sep=0pt, outer sep=0pt, scale=  1.40] at ( 43.37,197.97) {75};

\node[text=drawColor,anchor=base east,inner sep=0pt, outer sep=0pt, scale=  1.40] at ( 43.37,248.94) {100};
\end{scope}
\begin{scope}
\path[clip] (  0.00,  0.00) rectangle (289.08,289.08);
\definecolor{drawColor}{gray}{0.20}

\path[draw=drawColor,line width= 0.6pt,line join=round] ( 45.57, 49.89) --
	( 48.32, 49.89);

\path[draw=drawColor,line width= 0.6pt,line join=round] ( 45.57,100.85) --
	( 48.32,100.85);

\path[draw=drawColor,line width= 0.6pt,line join=round] ( 45.57,151.82) --
	( 48.32,151.82);

\path[draw=drawColor,line width= 0.6pt,line join=round] ( 45.57,202.79) --
	( 48.32,202.79);

\path[draw=drawColor,line width= 0.6pt,line join=round] ( 45.57,253.76) --
	( 48.32,253.76);
\end{scope}
\begin{scope}
\path[clip] (  0.00,  0.00) rectangle (289.08,289.08);
\definecolor{drawColor}{gray}{0.20}

\path[draw=drawColor,line width= 0.6pt,line join=round] ( 94.66, 36.94) --
	( 94.66, 39.69);

\path[draw=drawColor,line width= 0.6pt,line join=round] (165.95, 36.94) --
	(165.95, 39.69);

\path[draw=drawColor,line width= 0.6pt,line join=round] (237.24, 36.94) --
	(237.24, 39.69);
\end{scope}
\begin{scope}
\path[clip] (  0.00,  0.00) rectangle (289.08,289.08);
\definecolor{drawColor}{gray}{0.30}

\node[text=drawColor,anchor=base,inner sep=0pt, outer sep=0pt, scale=  1.40] at ( 94.66, 25.10) {100};

\node[text=drawColor,anchor=base,inner sep=0pt, outer sep=0pt, scale=  1.40] at (165.95, 25.10) {200};

\node[text=drawColor,anchor=base,inner sep=0pt, outer sep=0pt, scale=  1.40] at (237.24, 25.10) {300};
\end{scope}
\begin{scope}
\path[clip] (  0.00,  0.00) rectangle (289.08,289.08);
\definecolor{drawColor}{RGB}{0,0,0}

\node[text=drawColor,anchor=base,inner sep=0pt, outer sep=0pt, scale=  1.60] at (165.95,  8.61) {\itshape $t$};
\end{scope}
\begin{scope}
\path[clip] (  0.00,  0.00) rectangle (289.08,289.08);
\definecolor{drawColor}{RGB}{0,0,0}

\node[text=drawColor,rotate= 90.00,anchor=base,inner sep=0pt, outer sep=0pt, scale=  1.60] at ( 16.52,151.82) {remaining taxa $\left[\%\right]$};
\end{scope}
\begin{scope}
\path[clip] (  0.00,  0.00) rectangle (289.08,289.08);
\definecolor{drawColor}{RGB}{0,0,0}

\node[text=drawColor,anchor=base,inner sep=0pt, outer sep=0pt, scale=  1.60] at (165.95,272.56) {(c) $s=90$};
\end{scope}
\end{tikzpicture}}
\end{minipage}
\caption{\blue{The average percentage of taxa that remains after the subtree reduction (solid line), after the subtree and chain reduction (dotted line), and after all seven reductions (dashed line) for all tree pairs with $k=20$,} \revisionSK{for different skew values.}}
\label{fig:reductions-zoom-in}
\end{figure}
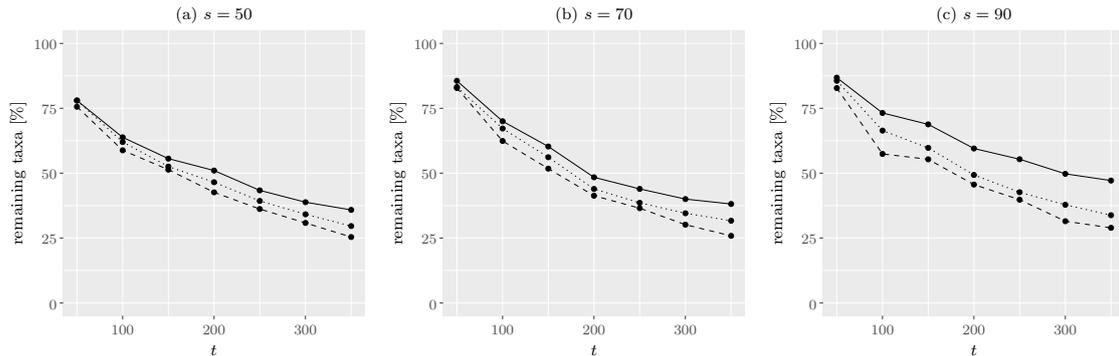

To evaluate whether or not the five new reductions further reduce phylogenetic trees that have already been reduced by the subtree and chain reduction, we have analyzed the percentage of remaining taxa of all 735 tree pairs  after (i) the subtree reduction,  (ii) the subtree and chain reductions, and (iii)  all seven reductions. More specifically, for each pair $t$ and $k$ with $t\in\{50,100,150,\ldots,350\}$ and $k\in\{5,10,15,\ldots,35\}$, we have calculated the three average percentages $$\frac{100}{t}\cdot s(T,T')\hspace{1cm} \frac{100}{t}\cdot sc(T,T') \hspace{1cm}\frac{100}{t}\cdot scn(T,T')$$ over all 15 tree pairs $(T,T')$ with parameter combination $(t,s,k)$, where $s\in\{50,70,90\}$. The results are summarized in Figure~\ref{fig:reductions}. For example, the seven reductions reduce tree pairs with $t=350$ and $k=5$ by about $90\%$ and tree pairs with $t=350$ and $k=35$ by about $60\%$. Except for the two parameter combinations $(50,s,30)$ and $(50,s,35)$, Figure~\ref{fig:reductions} shows that applying all seven reductions always results (on average) in smaller tree pairs than applying the subtree and chain reductions only.
%For a fixed $k$, the power of the five new reductions becomes first more pronounced for increasing $t$ before the difference $$\frac{100}{t}\cdot sc(T,T')-\frac{100}{t}\cdot scn(T,T')$$ becomes and remains (roughly) constant. This can be observed in Figure~\ref{fig:reductions} as the vertical gap between the dotted and dashed line in each subfigure first increases before both lines run (roughly) in parallel to each other.
% Furthermore
Observe that, regardless of $k$, the percentage of the number of remaining taxa decreases as $t$ increases. For example, for $k=20$ and $t=50$, we have $\frac{100}{50}\cdot scn(T,T')\approx 80\%$, and for $k=20$ and $t=350$, we have $\frac{100}{350}\cdot scn(T,T')\approx 27\%$.  An analogous trend applies to all other values of $k$. This can be partially explained by recalling that, when applied to exhaustion, the seven reduction rules guarantee that the resulting reduced instance
has no more than $11d_{\TBR}-9$ taxa \cite{kelk2020new}. So, taking $k$ as a proxy for $d_{\TBR}$, the ratio $(11k-9)/t$ decreases as $t$ increases, if $k$ is fixed. \revisionSK{Hence, the curve behaves like $1/t$ (a hyperbola) for fixed $k$; decreasing, but at an ever slower rate.} A similar observation holds if we only apply the subtree and chain reduction, since the reduced instances are then guaranteed to have at most $15d_{\TBR}-9$ taxa. Interestingly, when only the subtree reduction is applied, the \blue{solid} curves in  \blue{Figure~\ref{fig:reductions}} have \blue{a very similar} downward-sloping tendency as the $sc$ and $scn$ curves, despite the fact that, when applied in isolation there is no $g(k)$
theoretical upper bound on the size of the reduced instance\footnote{\steven{Such a bound cannot exist. Consider, for example, the situation when $T$ and $T'$ are caterpillar trees - every internal vertex is incident to at least one taxon - on $n+2$ taxa, $n\geq 3$, with taxa ordered $x, 1, 2, \ldots, n, y$ and $y, 1, 2, \ldots, n, x$ in $T$ and $T'$, respectively. \purple{Irrespective of $n$, the TBR distance between $T$ and $T'$ is two, the} subtree reduction does nothing here, and the chain reduction collapses both trees to 5 taxa.}}. 
 \blue{To gain more insight into this phenomena, we have redone the analysis underlying Figure~\ref{fig:reductions}(d), but separately for each skew value $s\in\{50,70,90\}$.} 
 %Tree pairs that are expected to have more common chains and less common subtrees as $s$ increases.
\revisionSK{The results, presented in Figure~\ref{fig:reductions-zoom-in}, suggest that for $s \in \{50,70\}$ similar performance is observed, but that at $s=90$ the achieved reduction is decreasing when only the subtree reduction is applied. It seems that high skew militates against the creation of common subtrees. Intuitively, one would expect the chain reduction to work well at $s=90$, since high skew would seem to support the creation of multiple common long chains; perhaps the two effects (i.e. less effective subtree reduction, possibly more effective chain reduction) cancel each other out.}
%\st{Presumably the pairs of trees used in our experiments have the property that the number of taxa in common subtrees is somehow correlated with the number of taxa in common chains.}

\steven{\blue{Returning to Figure~\ref{fig:reductions}, a} further observation is as follows. \revisionSK{Recalling that $g(k)$ denotes the theoretical upper bound on the size of the kernel, }the fact that $g(k)/t$ decreases for fixed $k$ and increasing $t$ explains the downward-sloping shape of the $sc$ and $scn$ curves, but it only has limited explanatory power. If $t$ is smaller than
%, say,
$11d_{\TBR}-9$, then the reduction rules offer no guarantees at all. For example, for $k=35$, \blue{we have $11k-9=376$}, which is larger than any $t$ used in our main dataset. Nevertheless, as shown in \purple{Figure~\ref{fig:reductions}(g)}, significant reduction is still achieved. This suggests that the upper bound on kernel size may, in practice, be \purple{much} smaller than $11d_{\TBR}-9$, a point we elaborate on in \purple{Section~\ref{sec:f(k)}.}}

\steven{Finally, a brief reflection on the fact that the $(50,s,30)$ and $(50,s,35)$ instances exhibited little reduction and that the difference between $sc$ and $scn$ was negligible. Such instances have an extremely high TBR distance relative to the number of taxa. This seems to destroy any of the structures that are targeted by the reduction rules. Interestingly, such instances are not only difficult to reduce, they also \purple{appear} potentially difficult to solve, as we saw in the earlier section. In both cases this seems linked to the phenomenon that, for such instances,  a maximum agreement forest is likely to have many very small components. This is a point we will return to later.}

\subsection{Empirical kernel size} \label{sec:f(k)}

\begin{figure}[t]
\begin{minipage}[t]{0.45\textwidth}
\scalebox{0.55}{\input{figs/bounds-on-f_k_}}
\end{minipage}
\hspace{-1cm}
\begin{minipage}[t]{0.45\textwidth}
\scalebox{0.55}{\input{figs/bounds-on-f_k_-BIG-TREES}}
\end{minipage}
\hspace{-1cm}
\begin{minipage}[t]{0.06\textwidth}
\vspace{-4cm}
\includegraphics[width=\textwidth]{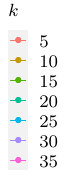}
\end{minipage}
\caption{\blue{Empirical estimates $f(k)$ of kernel size, depending on $t$, for the fully reduced tree pairs.}}
%: (a) averages for the main dataset and for all values of $k$; (b) averages for the main and large trees dataset for $k=35$.}
\label{fig:kernel-bound}
\end{figure}

\begin{figure}[t]
\begin{minipage}[t]{0.3\textwidth}
\scalebox{0.45}{\input{figs/bounds-on-f_k_-skew50}}
\end{minipage}
\begin{minipage}[t]{0.3\textwidth}
\scalebox{0.45}{\input{figs/bounds-on-f_k_-skew70}}
\end{minipage}
\begin{minipage}[t]{0.3\textwidth}
\scalebox{0.45}{\input{figs/bounds-on-f_k_-skew90}}
\end{minipage}
\hspace{0.2cm}
\begin{minipage}[t]{0.05\textwidth}
\vspace{-3.6cm}
\includegraphics[width=\textwidth]{figs/legend2}
\end{minipage}
\begin{minipage}[t]{0.3\textwidth}
\scalebox{0.45}{\input{figs/bounds-on-f_k_-BIG-TREES-skew50}}
\end{minipage}
\begin{minipage}[t]{0.3\textwidth}
\scalebox{0.45}{\input{figs/bounds-on-f_k_-BIG-TREES-skew70}}
\end{minipage}
\begin{minipage}[t]{0.3\textwidth}
\scalebox{0.45}{\input{figs/bounds-on-f_k_-BIG-TREES-skew90}}
\end{minipage}
\hspace{0.2cm}
\begin{minipage}[t]{0.05\textwidth}
\vspace{-3.6cm}
\includegraphics[width=\textwidth]{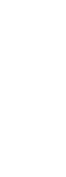}
\end{minipage}
\caption{\blue{Empirical estimates $f(k)$ of kernel size, depending on $t$, for the fully reduced tree pairs broken down by different skew values.}}
\label{fig:kernel-bound-skew}
\end{figure}

%\marginpar{SK: In some of the figures some words have been cut off (e.g. datase)}

By~\cite[Theorem 6]{kelk2020new}, we have a theoretical upper bound of $g(k) = 11d_\TBR(T,T')-9$ on the size of the TBR kernel if $T$ and $T'$ are fully reduced. Hence, $scn(T,T')\leq 11d_{\TBR}-9 \leq 11k-9$. To evaluate how well this theoretical kernel compares to the empirical kernel size  obtained from our experimental analysis, we have computed $$f(k)=\frac{scn(T,T')}{\theta(T,T')}$$ for all tree pairs of the main and larger trees dataset, where  $\theta=d_\TBR$ if $d_\TBR(T,T')$ is known and, otherwise, $\theta=d^\ell_\MP$. The results are summarized in \blue{Figures~\ref{fig:kernel-bound} and~\ref{fig:kernel-bound-skew}}. For each pair $t$ and $k$ with $t\in\{50,100,150,\ldots,350\}$ and $k\in\{5,10,15,\ldots,35\}$, we have plotted the average value of $f(k)$ over all 15 tree pairs with parameter combinations $(t,s,k)$, where $s\in\{50,70,90\}$ in Figure~\ref{fig:kernel-bound}(a). Across all 735 tree pairs of the main dataset, the maximum empirical kernel size is 9.8.
%Moreover, among all 38 tree pairs with $\theta=d^\ell_\MP$, the maximum empirical kernel size is 4.4.
Similarly for $k=35$ and each $$t\in\{50,100,150,\ldots,350,500,1000, 1500,\ldots,3000\},$$ we have plotted the average value of $f(k)$ over all 15 tree pairs with parameter combinations $(t,s,k)$ in Figure~\ref{fig:kernel-bound}(b). Across all 195 tree pairs (i.e. all 90 tree pairs of the large trees dataset and all 105 tree pairs of the main dataset with $k=35$), the maximum empirical kernel size is 7.37. In summary, our results show that the empirical kernel size $f(k)$ is, regardless of the parameter combination $(t,s,k)$, always significantly smaller than the theoretical upper bound. Note that the $g(k)$ bound in \cite{kelk2020new} is tight, in the sense that reduced instances can be constructed with \emph{exactly} $11d_{\TBR}-9$ taxa. This implies that the observed gap between the empirical and theoretical bound cannot be closed by strengthening the analysis in \cite{kelk2020new}. While $f(k)$ increases with $t$, the slope of the curve drops off for very large $t$ as shown in Figure~\ref{fig:kernel-bound}(b). However, even for $t=3000$, $f(k)=6.12$ is still much smaller than the theoretical upper bound of 11.  If we compare the empirical kernel sizes for a fixed $t$ and different values of $k$, we see in Figure~\ref{fig:kernel-bound}(a) that $f(k)$ decreases with increasing $k$. \steven{One possible explanation for this \blue{is} as follows. For a fixed $t$, the number of tree pairs in the dataset that have the property $t/k < c$ (for a given constant $c$) clearly increases as $k$ increases. For example, for $k=35$ all the trees with $t \in \{50,100\}$ satisfy the inequality if we take $c=3$; but at $k=10$ and $c=3$ none of the trees in the dataset do, because all the trees in the dataset have at least 50 taxa. Now, if we take $k$ as an estimator of $d_{\TBR}$, we see that $f(k)=\frac{scn(T,T')}{\theta(T,T')}$ will \emph{a priori} be at most $c$ for tree pairs satisfying the inequality, because $scn(T,T') \leq t$ and $\theta(T,T') \approx k$. Possibly this has the effect of pulling the $f(k)$ curves downwards for increasing $k$ i.e. because a higher proportion of the trees in the dataset have a small number of taxa relative to $d_{\TBR}$, and thus contribute very low $f(k)$ values.}

\blue{To evaluate the impact of different skew values on the $f(k)$ values, we have repeated the analysis that underlies Figure~\ref{fig:kernel-bound} for each of the three skew values $\{50,70,90\}$. The results are shown in Figure~\ref{fig:kernel-bound-skew} with curves that are overall very similar to those shown in Figure~\ref{fig:kernel-bound}. Moreover, for fixed values of $k$ and $t$, we see that $f(k)$ increases very slowly as $s$ increases which may again be the result of the subtree reduction being (on average) less effective when applied to trees pairs that are heavily skewed} \revisionSK{(possibly to the extent of cancelling out enhanced performance of the chain reduction - but whether this indeed occurs requires further investigation).}
%\marginpar{On the other hand, heavily skewed trees may have more chains .... but chains are reduced to 3 leaves and subtrees are reduced to 1 leaf. Could that make a difference?}

%So if
%For $c$ sufficiently small, say $c<3$, it is plausible based on the findings in the previous section that the reduction rules have much less effect on such pairs of trees. Hence, in these cases $f(k)=\frac{scn(T,T')}{\theta(T,T')}$ becomes roughly equal to $\frac{t}{k}$, which is at most $c$.

\begin{figure}[t]
\noindent\begin{minipage}[t]{0.33\textwidth}
\scalebox{0.48}{% Created by tikzDevice version 0.12.3.1 on 2021-08-05 16:12:33
% !TEX encoding = UTF-8 Unicode
\begin{tikzpicture}[x=1pt,y=1pt]
\definecolor{fillColor}{RGB}{255,255,255}
\path[use as bounding box,fill=fillColor,fill opacity=0.00] (0,0) rectangle (289.08,289.08);
\begin{scope}
\path[clip] (  0.00,  0.00) rectangle (289.08,289.08);
\definecolor{drawColor}{RGB}{255,255,255}
\definecolor{fillColor}{RGB}{255,255,255}

\path[draw=drawColor,line width= 0.6pt,line join=round,line cap=round,fill=fillColor] (  0.00,  0.00) rectangle (289.08,289.08);
\end{scope}
\begin{scope}
\path[clip] ( 41.33, 39.69) rectangle (283.58,263.95);
\definecolor{fillColor}{gray}{0.92}

\path[fill=fillColor] ( 41.33, 39.69) rectangle (283.58,263.95);
\definecolor{drawColor}{RGB}{255,255,255}

\path[draw=drawColor,line width= 0.3pt,line join=round] ( 41.33, 59.60) --
	(283.58, 59.60);

\path[draw=drawColor,line width= 0.3pt,line join=round] ( 41.33,147.85) --
	(283.58,147.85);

\path[draw=drawColor,line width= 0.3pt,line join=round] ( 41.33,236.11) --
	(283.58,236.11);

\path[draw=drawColor,line width= 0.3pt,line join=round] ( 52.34, 39.69) --
	( 52.34,263.95);

\path[draw=drawColor,line width= 0.3pt,line join=round] (125.75, 39.69) --
	(125.75,263.95);

\path[draw=drawColor,line width= 0.3pt,line join=round] (199.16, 39.69) --
	(199.16,263.95);

\path[draw=drawColor,line width= 0.3pt,line join=round] (272.57, 39.69) --
	(272.57,263.95);

\path[draw=drawColor,line width= 0.6pt,line join=round] ( 41.33,103.72) --
	(283.58,103.72);

\path[draw=drawColor,line width= 0.6pt,line join=round] ( 41.33,191.98) --
	(283.58,191.98);

\path[draw=drawColor,line width= 0.6pt,line join=round] ( 89.04, 39.69) --
	( 89.04,263.95);

\path[draw=drawColor,line width= 0.6pt,line join=round] (162.45, 39.69) --
	(162.45,263.95);

\path[draw=drawColor,line width= 0.6pt,line join=round] (235.86, 39.69) --
	(235.86,263.95);
\definecolor{drawColor}{RGB}{0,0,0}

\path[draw=drawColor,line width= 0.6pt,line join=round] ( 52.34,113.67) --
	( 89.04,155.56) --
	(125.75,177.62) --
	(162.45,190.98) --
	(199.16,201.63) --
	(235.86,216.92) --
	(272.57,223.75);

\path[draw=drawColor,line width= 0.6pt,dash pattern=on 1pt off 3pt ,line join=round] ( 52.34,104.02) --
	( 89.04,125.55) --
	(125.75,136.20) --
	(162.45,143.91) --
	(199.16,140.20) --
	(235.86,150.09) --
	(272.57,153.15);

\path[draw=drawColor,line width= 0.6pt,dash pattern=on 4pt off 4pt ,line join=round] ( 52.34, 93.84) --
	( 89.04,110.19) --
	(125.75,104.19) --
	(162.45,114.20) --
	(199.16,119.96) --
	(235.86,121.14) --
	(272.57,123.49);
\definecolor{fillColor}{RGB}{0,0,0}

\path[draw=drawColor,line width= 0.4pt,line join=round,line cap=round,fill=fillColor] ( 52.34,113.67) circle (  1.96);

\path[draw=drawColor,line width= 0.4pt,line join=round,line cap=round,fill=fillColor] ( 89.04,155.56) circle (  1.96);

\path[draw=drawColor,line width= 0.4pt,line join=round,line cap=round,fill=fillColor] (125.75,177.62) circle (  1.96);

\path[draw=drawColor,line width= 0.4pt,line join=round,line cap=round,fill=fillColor] (162.45,190.98) circle (  1.96);

\path[draw=drawColor,line width= 0.4pt,line join=round,line cap=round,fill=fillColor] (199.16,201.63) circle (  1.96);

\path[draw=drawColor,line width= 0.4pt,line join=round,line cap=round,fill=fillColor] (235.86,216.92) circle (  1.96);

\path[draw=drawColor,line width= 0.4pt,line join=round,line cap=round,fill=fillColor] (272.57,223.75) circle (  1.96);

\path[draw=drawColor,line width= 0.4pt,line join=round,line cap=round,fill=fillColor] ( 52.34,104.02) circle (  1.96);

\path[draw=drawColor,line width= 0.4pt,line join=round,line cap=round,fill=fillColor] ( 89.04,125.55) circle (  1.96);

\path[draw=drawColor,line width= 0.4pt,line join=round,line cap=round,fill=fillColor] (125.75,136.20) circle (  1.96);

\path[draw=drawColor,line width= 0.4pt,line join=round,line cap=round,fill=fillColor] (162.45,143.91) circle (  1.96);

\path[draw=drawColor,line width= 0.4pt,line join=round,line cap=round,fill=fillColor] (199.16,140.20) circle (  1.96);

\path[draw=drawColor,line width= 0.4pt,line join=round,line cap=round,fill=fillColor] (235.86,150.09) circle (  1.96);

\path[draw=drawColor,line width= 0.4pt,line join=round,line cap=round,fill=fillColor] (272.57,153.15) circle (  1.96);

\path[draw=drawColor,line width= 0.4pt,line join=round,line cap=round,fill=fillColor] ( 52.34, 93.84) circle (  1.96);

\path[draw=drawColor,line width= 0.4pt,line join=round,line cap=round,fill=fillColor] ( 89.04,110.19) circle (  1.96);

\path[draw=drawColor,line width= 0.4pt,line join=round,line cap=round,fill=fillColor] (125.75,104.19) circle (  1.96);

\path[draw=drawColor,line width= 0.4pt,line join=round,line cap=round,fill=fillColor] (162.45,114.20) circle (  1.96);

\path[draw=drawColor,line width= 0.4pt,line join=round,line cap=round,fill=fillColor] (199.16,119.96) circle (  1.96);

\path[draw=drawColor,line width= 0.4pt,line join=round,line cap=round,fill=fillColor] (235.86,121.14) circle (  1.96);

\path[draw=drawColor,line width= 0.4pt,line join=round,line cap=round,fill=fillColor] (272.57,123.49) circle (  1.96);
\end{scope}
\begin{scope}
\path[clip] (  0.00,  0.00) rectangle (289.08,289.08);
\definecolor{drawColor}{gray}{0.30}

\node[text=drawColor,anchor=base east,inner sep=0pt, outer sep=0pt, scale=  1.40] at ( 36.38, 98.90) {5};

\node[text=drawColor,anchor=base east,inner sep=0pt, outer sep=0pt, scale=  1.40] at ( 36.38,187.16) {10};
\end{scope}
\begin{scope}
\path[clip] (  0.00,  0.00) rectangle (289.08,289.08);
\definecolor{drawColor}{gray}{0.20}

\path[draw=drawColor,line width= 0.6pt,line join=round] ( 38.58,103.72) --
	( 41.33,103.72);

\path[draw=drawColor,line width= 0.6pt,line join=round] ( 38.58,191.98) --
	( 41.33,191.98);
\end{scope}
\begin{scope}
\path[clip] (  0.00,  0.00) rectangle (289.08,289.08);
\definecolor{drawColor}{gray}{0.20}

\path[draw=drawColor,line width= 0.6pt,line join=round] ( 89.04, 36.94) --
	( 89.04, 39.69);

\path[draw=drawColor,line width= 0.6pt,line join=round] (162.45, 36.94) --
	(162.45, 39.69);

\path[draw=drawColor,line width= 0.6pt,line join=round] (235.86, 36.94) --
	(235.86, 39.69);
\end{scope}
\begin{scope}
\path[clip] (  0.00,  0.00) rectangle (289.08,289.08);
\definecolor{drawColor}{gray}{0.30}

\node[text=drawColor,anchor=base,inner sep=0pt, outer sep=0pt, scale=  1.40] at ( 89.04, 25.10) {100};

\node[text=drawColor,anchor=base,inner sep=0pt, outer sep=0pt, scale=  1.40] at (162.45, 25.10) {200};

\node[text=drawColor,anchor=base,inner sep=0pt, outer sep=0pt, scale=  1.40] at (235.86, 25.10) {300};
\end{scope}
\begin{scope}
\path[clip] (  0.00,  0.00) rectangle (289.08,289.08);
\definecolor{drawColor}{RGB}{0,0,0}

\node[text=drawColor,anchor=base,inner sep=0pt, outer sep=0pt, scale=  1.60] at (162.45,  8.61) {\itshape $t$};
\end{scope}
\begin{scope}
\path[clip] (  0.00,  0.00) rectangle (289.08,289.08);
\definecolor{drawColor}{RGB}{0,0,0}

\node[text=drawColor,rotate= 90.00,anchor=base,inner sep=0pt, outer sep=0pt, scale=  1.60] at ( 16.52,151.82) {$f(k)$};
\end{scope}
\begin{scope}
\path[clip] (  0.00,  0.00) rectangle (289.08,289.08);
\definecolor{drawColor}{RGB}{0,0,0}

\node[text=drawColor,anchor=base,inner sep=0pt, outer sep=0pt, scale=  1.60] at (162.45,272.56) {(a) $k=5$ main dataset};
\end{scope}
\end{tikzpicture}}
\end{minipage}
\begin{minipage}[t]{0.33\textwidth}
\scalebox{0.48}{% Created by tikzDevice version 0.12.3.1 on 2021-08-05 16:12:33
% !TEX encoding = UTF-8 Unicode
\begin{tikzpicture}[x=1pt,y=1pt]
\definecolor{fillColor}{RGB}{255,255,255}
\path[use as bounding box,fill=fillColor,fill opacity=0.00] (0,0) rectangle (289.08,289.08);
\begin{scope}
\path[clip] (  0.00,  0.00) rectangle (289.08,289.08);
\definecolor{drawColor}{RGB}{255,255,255}
\definecolor{fillColor}{RGB}{255,255,255}

\path[draw=drawColor,line width= 0.6pt,line join=round,line cap=round,fill=fillColor] (  0.00,  0.00) rectangle (289.08,289.08);
\end{scope}
\begin{scope}
\path[clip] ( 41.33, 39.69) rectangle (283.58,263.95);
\definecolor{fillColor}{gray}{0.92}

\path[fill=fillColor] ( 41.33, 39.69) rectangle (283.58,263.95);
\definecolor{drawColor}{RGB}{255,255,255}

\path[draw=drawColor,line width= 0.3pt,line join=round] ( 41.33, 59.60) --
	(283.58, 59.60);

\path[draw=drawColor,line width= 0.3pt,line join=round] ( 41.33,147.85) --
	(283.58,147.85);

\path[draw=drawColor,line width= 0.3pt,line join=round] ( 41.33,236.11) --
	(283.58,236.11);

\path[draw=drawColor,line width= 0.3pt,line join=round] ( 52.34, 39.69) --
	( 52.34,263.95);

\path[draw=drawColor,line width= 0.3pt,line join=round] (125.75, 39.69) --
	(125.75,263.95);

\path[draw=drawColor,line width= 0.3pt,line join=round] (199.16, 39.69) --
	(199.16,263.95);

\path[draw=drawColor,line width= 0.3pt,line join=round] (272.57, 39.69) --
	(272.57,263.95);

\path[draw=drawColor,line width= 0.6pt,line join=round] ( 41.33,103.72) --
	(283.58,103.72);

\path[draw=drawColor,line width= 0.6pt,line join=round] ( 41.33,191.98) --
	(283.58,191.98);

\path[draw=drawColor,line width= 0.6pt,line join=round] ( 89.04, 39.69) --
	( 89.04,263.95);

\path[draw=drawColor,line width= 0.6pt,line join=round] (162.45, 39.69) --
	(162.45,263.95);

\path[draw=drawColor,line width= 0.6pt,line join=round] (235.86, 39.69) --
	(235.86,263.95);
\definecolor{drawColor}{RGB}{0,0,0}

\path[draw=drawColor,line width= 0.6pt,line join=round] ( 52.34, 50.47) --
	( 89.04, 61.88) --
	(125.75, 75.03) --
	(162.45, 87.01) --
	(199.16, 95.89) --
	(235.86,103.81) --
	(272.57,109.25);

\path[draw=drawColor,line width= 0.6pt,dash pattern=on 1pt off 3pt ,line join=round] ( 52.34, 50.47) --
	( 89.04, 61.32) --
	(125.75, 73.05) --
	(162.45, 82.54) --
	(199.16, 90.20) --
	(235.86, 95.85) --
	(272.57, 97.89);

\path[draw=drawColor,line width= 0.6pt,dash pattern=on 4pt off 4pt ,line join=round] ( 52.34, 50.31) --
	( 89.04, 60.34) --
	(125.75, 69.53) --
	(162.45, 77.72) --
	(199.16, 84.44) --
	(235.86, 88.78) --
	(272.57, 89.61);
\definecolor{fillColor}{RGB}{0,0,0}

\path[draw=drawColor,line width= 0.4pt,line join=round,line cap=round,fill=fillColor] ( 52.34, 50.47) circle (  1.96);

\path[draw=drawColor,line width= 0.4pt,line join=round,line cap=round,fill=fillColor] ( 89.04, 61.88) circle (  1.96);

\path[draw=drawColor,line width= 0.4pt,line join=round,line cap=round,fill=fillColor] (125.75, 75.03) circle (  1.96);

\path[draw=drawColor,line width= 0.4pt,line join=round,line cap=round,fill=fillColor] (162.45, 87.01) circle (  1.96);

\path[draw=drawColor,line width= 0.4pt,line join=round,line cap=round,fill=fillColor] (199.16, 95.89) circle (  1.96);

\path[draw=drawColor,line width= 0.4pt,line join=round,line cap=round,fill=fillColor] (235.86,103.81) circle (  1.96);

\path[draw=drawColor,line width= 0.4pt,line join=round,line cap=round,fill=fillColor] (272.57,109.25) circle (  1.96);

\path[draw=drawColor,line width= 0.4pt,line join=round,line cap=round,fill=fillColor] ( 52.34, 50.47) circle (  1.96);

\path[draw=drawColor,line width= 0.4pt,line join=round,line cap=round,fill=fillColor] ( 89.04, 61.32) circle (  1.96);

\path[draw=drawColor,line width= 0.4pt,line join=round,line cap=round,fill=fillColor] (125.75, 73.05) circle (  1.96);

\path[draw=drawColor,line width= 0.4pt,line join=round,line cap=round,fill=fillColor] (162.45, 82.54) circle (  1.96);

\path[draw=drawColor,line width= 0.4pt,line join=round,line cap=round,fill=fillColor] (199.16, 90.20) circle (  1.96);

\path[draw=drawColor,line width= 0.4pt,line join=round,line cap=round,fill=fillColor] (235.86, 95.85) circle (  1.96);

\path[draw=drawColor,line width= 0.4pt,line join=round,line cap=round,fill=fillColor] (272.57, 97.89) circle (  1.96);

\path[draw=drawColor,line width= 0.4pt,line join=round,line cap=round,fill=fillColor] ( 52.34, 50.31) circle (  1.96);

\path[draw=drawColor,line width= 0.4pt,line join=round,line cap=round,fill=fillColor] ( 89.04, 60.34) circle (  1.96);

\path[draw=drawColor,line width= 0.4pt,line join=round,line cap=round,fill=fillColor] (125.75, 69.53) circle (  1.96);

\path[draw=drawColor,line width= 0.4pt,line join=round,line cap=round,fill=fillColor] (162.45, 77.72) circle (  1.96);

\path[draw=drawColor,line width= 0.4pt,line join=round,line cap=round,fill=fillColor] (199.16, 84.44) circle (  1.96);

\path[draw=drawColor,line width= 0.4pt,line join=round,line cap=round,fill=fillColor] (235.86, 88.78) circle (  1.96);

\path[draw=drawColor,line width= 0.4pt,line join=round,line cap=round,fill=fillColor] (272.57, 89.61) circle (  1.96);
\end{scope}
\begin{scope}
\path[clip] (  0.00,  0.00) rectangle (289.08,289.08);
\definecolor{drawColor}{gray}{0.30}

\node[text=drawColor,anchor=base east,inner sep=0pt, outer sep=0pt, scale=  1.40] at ( 36.38, 98.90) {5};

\node[text=drawColor,anchor=base east,inner sep=0pt, outer sep=0pt, scale=  1.40] at ( 36.38,187.16) {10};
\end{scope}
\begin{scope}
\path[clip] (  0.00,  0.00) rectangle (289.08,289.08);
\definecolor{drawColor}{gray}{0.20}

\path[draw=drawColor,line width= 0.6pt,line join=round] ( 38.58,103.72) --
	( 41.33,103.72);

\path[draw=drawColor,line width= 0.6pt,line join=round] ( 38.58,191.98) --
	( 41.33,191.98);
\end{scope}
\begin{scope}
\path[clip] (  0.00,  0.00) rectangle (289.08,289.08);
\definecolor{drawColor}{gray}{0.20}

\path[draw=drawColor,line width= 0.6pt,line join=round] ( 89.04, 36.94) --
	( 89.04, 39.69);

\path[draw=drawColor,line width= 0.6pt,line join=round] (162.45, 36.94) --
	(162.45, 39.69);

\path[draw=drawColor,line width= 0.6pt,line join=round] (235.86, 36.94) --
	(235.86, 39.69);
\end{scope}
\begin{scope}
\path[clip] (  0.00,  0.00) rectangle (289.08,289.08);
\definecolor{drawColor}{gray}{0.30}

\node[text=drawColor,anchor=base,inner sep=0pt, outer sep=0pt, scale=  1.40] at ( 89.04, 25.10) {100};

\node[text=drawColor,anchor=base,inner sep=0pt, outer sep=0pt, scale=  1.40] at (162.45, 25.10) {200};

\node[text=drawColor,anchor=base,inner sep=0pt, outer sep=0pt, scale=  1.40] at (235.86, 25.10) {300};
\end{scope}
\begin{scope}
\path[clip] (  0.00,  0.00) rectangle (289.08,289.08);
\definecolor{drawColor}{RGB}{0,0,0}

\node[text=drawColor,anchor=base,inner sep=0pt, outer sep=0pt, scale=  1.60] at (162.45,  8.61) {\itshape $t$};
\end{scope}
\begin{scope}
\path[clip] (  0.00,  0.00) rectangle (289.08,289.08);
\definecolor{drawColor}{RGB}{0,0,0}

\node[text=drawColor,rotate= 90.00,anchor=base,inner sep=0pt, outer sep=0pt, scale=  1.60] at ( 16.52,151.82) {$f(k)$};
\end{scope}
\begin{scope}
\path[clip] (  0.00,  0.00) rectangle (289.08,289.08);
\definecolor{drawColor}{RGB}{0,0,0}

\node[text=drawColor,anchor=base,inner sep=0pt, outer sep=0pt, scale=  1.60] at (162.45,272.56) {(b) $k=35$ main dataset};
\end{scope}
\end{tikzpicture}}
\end{minipage}
\begin{minipage}[t]{0.33\textwidth}
\scalebox{0.48}{\input{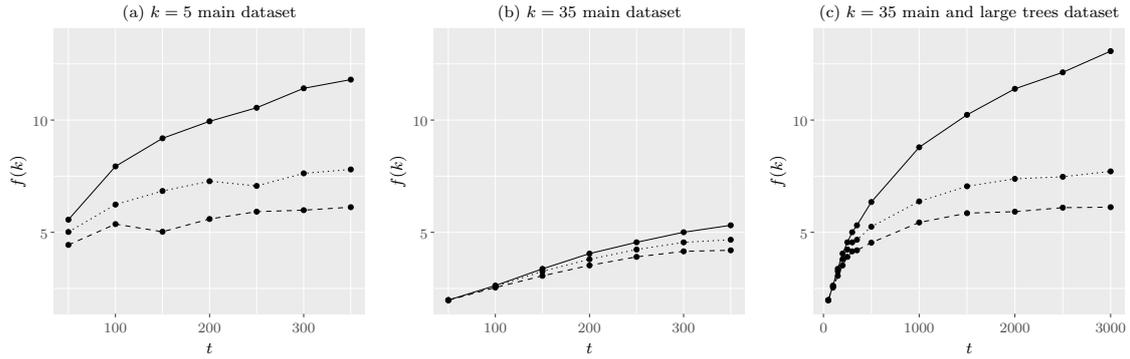}}
\end{minipage}
%\begin{minipage}[t]{0.09\textwidth}
%\vspace{-2.8cm}
%\includegraphics[width=\textwidth]{figs/legend1}
%\end{minipage}
\caption{\blue{Empirical estimates $f(k)$ of kernel size, depending on $t$, for the subtree reduced (solid), the subtree and chain reduced (dotted), and the fully reduced (dashed) tree pairs.}}
%: (a) averages for the main dataset with $k=5$; (b) averages for the main dataset with $k=35$; (c) averages for the main and large trees dataset for $k=35$.}
\label{fig:more-kernel-bounds}
\end{figure}

\blue{Finally, we note that similar to the analysis that underlies Figures~\ref{fig:kernel-bound},} empirical kernel sizes can also be computed for the tree pairs that have been reduced under the subtree reduction only, or under the subtree and chain reduction only. \blue{For a selection of tree pairs of both data sets,} these are shown in Figure~\ref{fig:more-kernel-bounds}. In particular, we note that the $sc$ curve remains well below the worst-case bound of $15$; in Figure~\ref{fig:more-kernel-bounds}(c) the curve seems to be flattening at around \blue{$f(k)=7$}.

\revisionSK{These results hint at the possibility that the empirical kernel size is determined by more parameters than just $d_{\TBR}$. It is difficult to say what these hidden parameters might be. Popular `local' phylogenetic parameters such as \emph{level} \cite{bordewich2017fixed} do not explain the phenomenon, since the seven reduction rules
are largely uninfluenced by low level.
%\marginpar{Does this need more explanation or a pointer to the last paragraph of the appendix?}
The phenomenon is probably linked to the fact that the tight instances described in \cite{tightkernel,kelk2020new} were extremely
carefully engineered. It is very unlikely that such instances will occur in our experimental setting, so the worst-case kernel size is unlikely to occur. To understand this phenomenon further it will be necessary to carefully analyze such tight instances and formalize how (and, quantitatively, how far) they differ from trees generated in our experimental setting; we defer this to future work.}

\subsection{Parameter reductions}\label{sec:param-red}
%Three of the five new reductions described in~\cite{kelk2020new} reduce two unrooted binary phylogenetic trees $T$ and $T'$ on the same leaf set to a pair of smaller trees $S$ and $S'$, respectively, such that $d_{\TBR}(T,T')=d_{\TBR}(S,S')+1$. We refer to these reductions as {\it parameter reductions}.
We have analyzed the frequency with which the five new reductions have triggered parameter reductions among all 735 tree pairs and observed that 499 tree pairs were reduced by at least one application of Reduction 3, 4, or 5. The maximum number of parameter reductions for a single tree pair is seven. A histogram that presents the distribution of the 499 tree pairs over all  combinations of $k\in\{5,10, 15,\ldots,35\}$ and $t=\{50,100,150,\ldots,350\}$ as well as a table that summarizes the frequency of parameter reductions over all 735 tree pairs is shown in Figure~\ref{fig:histo_para_red}.  Each of the 3  tree pairs that have undergone seven parameter reductions have a parameter combination with $t\geq 250$ and $k\geq 30$. Figure~\ref{fig:histo_para_red} indicates that the number of tree pairs that have triggered at least one parameter reduction increases as both $t$ and $k$ grow.  To confirm this trend, we have also analyzed the amount of parameter reductions for all tree pairs of the larger trees dataset. For this dataset, all 90 tree pairs have been reduced by at least one application of Reduction 3, 4, or 5 and the maximum number of such reductions applied to a single tree pair is ten. The associated tree pair has the parameter combination $t=3000$ and $k=35$.  These results show that not only the ordinary new reductions that preserve the TBR distance (these are Reductions 6 and 7) enhance the power of the subtree and chain reduction but that the same holds for the other three reductions.

 \begin{figure}[h]
 \begin{minipage}{0.5\linewidth}
		\centering
		\includegraphics{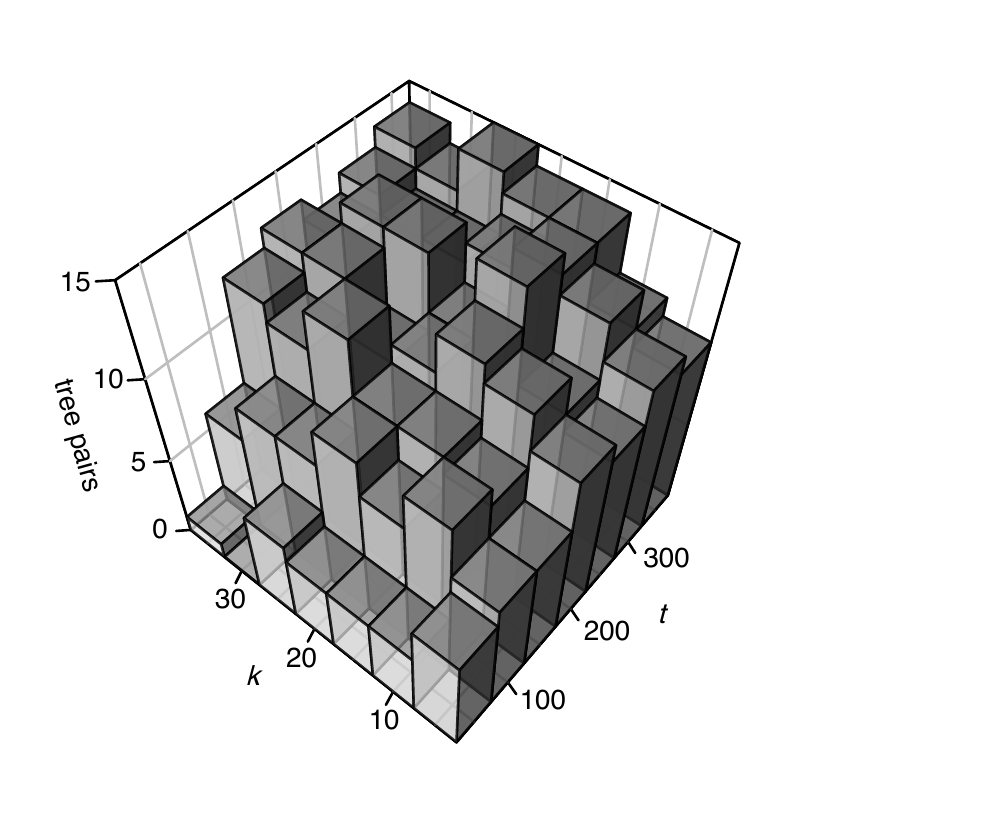}
	\end{minipage}
	\begin{minipage}{0.65\linewidth}
		\centering
		\begin{tabular}{|c|r|}
			\hline
			 parameter red.         & tree pairs \\
			\hline
			0&	236\\
			1&	217\\
			2&	133\\
			3&	89\\
			4&	34\\
			5&	17\\
			6&	6\\
			7&	3\\
			\hline
			skew         & tree pairs \\
			\hline
			50&	169\\
			70&	163\\
			90&	167\\
			\hline
		\end{tabular}
	\end{minipage}\hfill
%\captionof{figure}{Distribution of all 499 tree pairs that have undergone at least one parameter reduction.}
\caption{Left: Distribution of all 499 tree pairs of the main dataset that were reduced by at least one application of Reduction 3, 4, or 5. Right: Summary of all tree pairs that were reduced by at least one application of Reduction 3, 4, or 5 depending on the total number of such reductions and the skew of the tree pairs over all 735 tree pairs.}
\label{fig:histo_para_red}	
\end{figure}

 \revisionSK{It is known that two of the parameter-reducing reductions, Reductions 4 and 5, trigger an immediate subsequent application
of the subtree rule, but beyond this it is too difficult to argue analytically why any of the reduction rules should trigger each other. This is because the seven reduction rules target rather different local structures, and after applying a reduction rule the local structure is either destroyed or remains incompatible with other reduction rules. The phenomenon of parameter reduction occurring repeatedly for a single pair of trees is therefore quite possibly caused by multiple parts of the trees triggering reduction rules independently. Understanding whether this is indeed what happens, or whether prioritizing certain reduction rules over others works well in practice, requires further empirical research.}

\subsection{The \steven{quality} of the $d_{\MP}$ lower bound}
To evaluate the quality of $d_{\MP}^\ell$, which we have used as a lower bound on the TBR distance throughout our experimental study, we have compared  $d^\ell_{\MP}(T,T')$ and $d_{\TBR}(T,T')$ for all $697$ tree pairs $(T,T')$ for which we could calculate the  TBR distance exactly (see Section~\ref{sec:exactTBR}). Only~62 (that is $8.9\%$)  out of 697 tree pairs of the main dataset have $d^\ell_{\MP}(T,T')<d_{\TBR}(T,T')$. Among these 62 pairs, a majority of 92\% have $d_{\TBR}(T,T')-d^\ell_{\MP}(T,T')=1$. The maximum difference between $d_{\TBR}(T,T')$ and $d^\ell_{\MP}(T,T')$ is 3 over all 62 tree pairs. A histogram that presents the distribution of these 62 tree pairs over all parameter combinations of $k\in\{5,10, 15,\ldots,35\}$ and $t=\{50,100,150,\ldots,350\}$ is shown in Figure~\ref{fig:histo_dMP}.  The tree pair with $d_{\TBR}(T,T')-d^\ell_{\MP}(T,T')=3$ has the parameter combinations $t=350$ and $k=35$. In comparison, 48 out of all 90 tree pairs of the large trees dataset have  $d_{\MP}^\ell(T,T')<d_\TBR(T,T')$. However, the maximum difference between $d_{\TBR}(T,T')$ and $d^\ell_{\MP}(T,T')$ remains 3. These results verify that $d_{\MP}^\ell$ is an effective lower for computing $d_\TBR$ and that the difference between $d_{\TBR}(T,T')$ and $d^\ell_{\MP}(T,T')$  grows very slowly, even for large trees with up to 3000 leaves.

Another point worth drawing attention to, is the \emph{efficiency} of computing $d_{\MP}^\ell$. In our main dataset $d_{\MP}^\ell$ was, for 679 tree pairs, after 10 seconds greater than or equal to the best lower bound that uSPR had computed after 5 minutes of iterative deepening.

 \begin{figure}[h]
 \begin{minipage}{0.5\linewidth}
		\centering
		\includegraphics{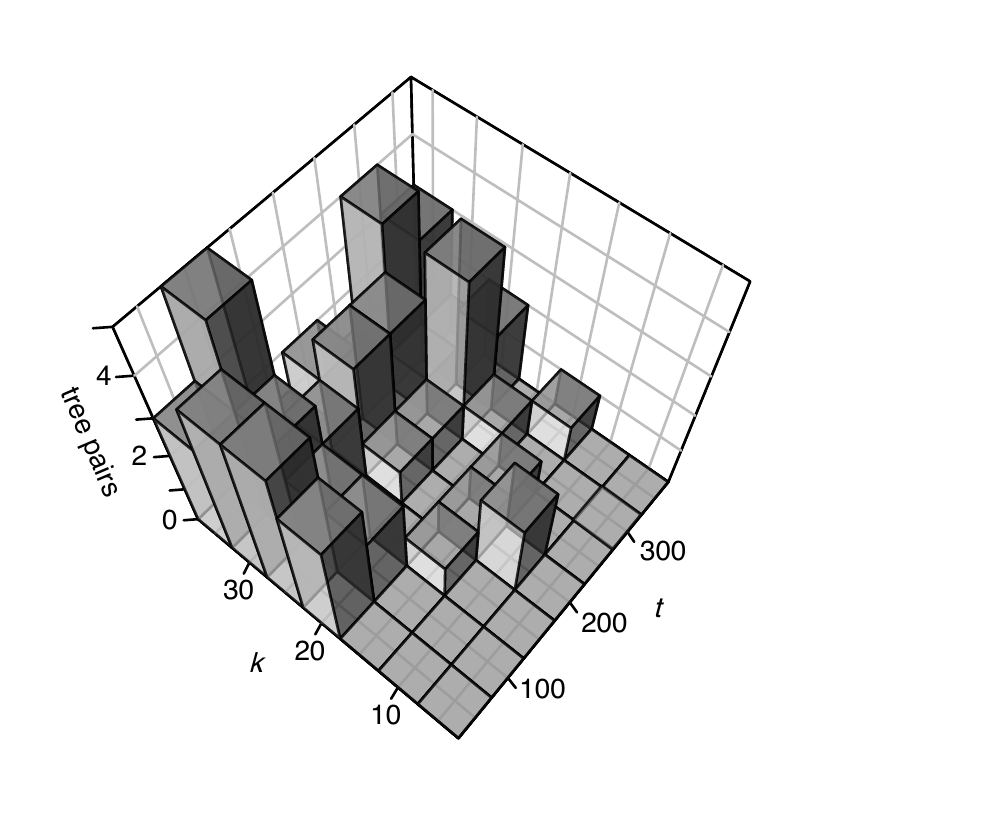}
	\end{minipage}
	\begin{minipage}{0.65\linewidth}
		\centering
		\begin{tabular}{|c|r|}
			\hline
			 $d_{\TBR}(T,T')-d^\ell_{\MP}(T,T')$        & tree pairs \\
			\hline
			0&	635\\
			1&	57\\
			2&	4\\
			3&	1\\	
			\hline
			skew         & tree pairs \\
			\hline
			50&	23\\
			70&	19\\
			90&	20\\
			\hline
		\end{tabular}
	\end{minipage}\hfill
%\captionof{figure}{Distribution of all 499 tree pairs that have undergone at least one parameter reduction.}
\caption{Left: Distribution of all 62 tree pairs $(T,T')$ with $d^\ell_{\MP}(T,T')<d_{\TBR}(T,T')$. Right: Summary of the difference between  $d_{\TBR}$ and $d^\ell_{\MP}$ over all 697 tree pairs with known TBR distance, and summary of all 62 tree pairs whose $d_{\MP}$ lower bound is less than their TBR distance depending on the skew of these tree pairs.}
\label{fig:histo_dMP}\end{figure}

\section{Discussion: shattered forests?}
\label{sec:toughcookies}

Recall that there were 38 pairs of trees for which neither uSPR nor Tubro could compute $d_{\TBR}$ in 5 minutes. These all had high $d_{\TBR}$ ($k \geq 25$). One of the most striking tractability insights from the experiments is that, amongst these instances, 24 had only 50 taxa. Indeed, Figure~\ref{fig:histo-unsolved-tree-pairs} suggests that, within the group of 38 instances, and for fixed $k$, trees with fewer taxa are \emph{more} likely to confound the solvers than trees with more taxa. This is \revisionSK{noteworthy} and, as explained in Section \ref{sec:exactTBR}, it is not obvious why uSPR or Tubro would exhibit this behavior. Less surprisingly, instances with $t=50$ and high TBR distance $(k \in \{30,35\})$
are largely unaffected by the seven reduction rules (see Section \ref{subsec:avgreduc}), while all other parameter combinations exhibit higher levels of reduction. The emerging picture is that tree pairs with a small number of taxa $t$ but (very) high TBR distance relative to $t$, are both hard to solve, and hard to reduce. The common factor here, we suspect, is that for such pairs of trees a maximum agreement forest will necessarily have many components, of which many will be small. Manual inspection of such instances suggests that many components of the forest indeed contain only a single, or perhaps two, taxa. Such ``shattered'' forests will only very occasionally trigger the seven reduction rules. Apparently, such forests and their lack of topological structure can also cause both the combinatorial branching strategy of uSPR and the ILP-based branching of Tubro severe problems. %This phenomenon requires further research.
%Computing $d_{\TBR}$ in such cases probably requires different algorithmic techniques to be developed.
 FPT algorithms such as uSPR are not designed to run quickly when the parameter in question (here, $d_{\TBR}$) is high, and the focus when developing such algorithms is usually to ensure fast running time when the size of the instance is \emph{large} relative to the parameter. Here we have a high parameter combined with \emph{small} instances. This requires further research and different algorithmic techniques.
 % This warrants further research into exponential-time, as oppozed to parameterized, algorithms for this problem.
 It is probably relevant that such instances increasingly start to resemble \emph{random} pairs of trees; the literature on the expected number of components in maximum agreement forests of random pairs of trees is therefore worth exploring~\cite{extremal2019}.
%\steven{This is the main section left to do now, yay!}

\section{Conclusions and future work}\label{sec:discussion}

In this article we have demonstrated that reduction rules for TBR distance have the potential to significantly reduce the size of instances, and that \steven{theoretically} stronger reduction rules do have clear added value \steven{in practice}. \purple{Our experimental results also highlight that}
%, in practice,
the empirical bound on the size of the \purple{TBR} kernel is significantly lower than that predicted in the worst-case ($15k-9$ and $11k-9$ respectively, for the two different sets of reduction rules). \purple{Lastly, we} have shown that parameter reduction occurs quite often, and that a sampling strategy based on sampling convex characters quickly yields strong lower bounds on $d_{\TBR}$.

In addition to the phenomenon of ``shattered forests'' raised in the previous section, a number of interesting questions have emerged from our work
 concerning (i) kernelization and (ii) the actual computation of TBR distance. Regarding kernelization, it is natural to ask why the empirical kernel bound is
%so
much better than the worst-case bound. Clearly, the carefully-constructed tight instances described in \cite{tightkernel,kelk2020new} are unlikely to occur in an experimental setting
such as ours. Such phenomena are pervasive in theoretical computer science. Nevertheless, can we rigorously explain \emph{why} the experimentally-generated instances can be more successfully reduced? One way to tackle this would be to search for \revisionSK{additional parameters} which, when combined with $d_{\TBR}$, produce a more accurate prediction of the obtained kernel size. \revisionSK{We currently do not have any concrete candidates for what these parameters could be, but understanding how the tight instances from~\cite{tightkernel,kelk2020new} differ from those in our experimental setting, and quantifying this dissimilarity, is likely to yield insights in this direction.} Next, it is natural to ask: is it possible, by developing new rules, to obtain a kernel smaller than $11k-9$, and does it make sense to implement these rules in practice? \revisionSK{Breaking the $11k-9$ barrier is likely to be primarily a theoretical exercise, guided by the limits of the counting arguments in \cite{kelk2020new}. However, perhaps an empirical examination of how solvers such as uSPR tackle fully reduced instances could offer some extra clues, at least on the question of whether any future reduction rules will be effective in practice.}

On the \steven{computational} side, there is \steven{still room} for improvement. \steven{While the combination of uSPR and Tubro is capable of solving most of the trees in our experimental dataset, there still exist} pairs of
%strikingly
sometimes small trees where neither uSPR or Tubro can compute $d_{\TBR}$ in 5 minutes. Perhaps the branching factor in the
FPT branching algorithm underpinning uSPR \cite{chen2015parameterized,whidden2018calculating} can be lowered by a deeper study of the combinatorics of agreement forests.
\revisionSK{It also seems important to understand why certain pairs of trees trigger multiple parameter reductions.} On the engineering side we are optimistic that more advanced ILP engineering techniques
%, from fields such as operations research,
can be used to speed up Tubro. \purple{Future research should focus on continuing the strategy of blending and merging existing techniques into an ensemble that we adopted in this article. These techniques include}  kernelization to reduce the size of instances (and to generate useful combinatorial insights, such as chain \purple{preservation}), fast generation of lower and upper bounds, FPT branching algorithms, ILP,  \purple{and} other exponential-time algorithms. In this article the techniques were coupled fairly loosely. We anticipate that an ensemble in which the techniques are more deeply integrated, \revisionSK{i.e. interleaved,} will yield further speedups. This is a challenging engineering task, but also a challenging mathematical one, since it is not always easy to translate intermediate solutions or bounds produced by one of the techniques to
another. \revisionSK{For example, as soon as a branching algorithm starts cutting edges in the trees, the instance becomes a more general type of instance: a pair of forests, rather than a pair of trees. To adapt the reduction rules from \cite{kelk2020new} to forests will require a significant theoretical effort.} 
%Nevertheless we are optimistic that with enough effort significant progress can be made.

%\section{Availabilty of code and data - put this information somewhere in the main text (probably the experimental section)}

\section*{Acknowledgements}

We would like to thank the participants of the \emph{Algorithms and Complexity in Phylogenetics} seminar in 2019 (Dagstuhl Seminar 19443) \purple{and Schloss Dagstuhl for hosting the seminar. We also thank A. Alhazmi for useful conversations on the topic of this paper and an anonymous referee for their helpful comments.}

\clearpage
\pagebreak
\appendix
\section{Tubro: Using kernelization and Integer Linear Programming to compute $d_{\TBR}$}

As is well known, kernelization does not in itself solve a problem, it simply reduces it in size. To solve the kernel an (efficient) exponential-time algorithm is required. Throughout the experimental section we mainly used uSPR to compute TBR distances, and this was primarily to be able to compute $f(k)$ values \revisionSK{(i.e. estimates of the empirical kernel size)}. As part of our work, we experimented not just with kernelization but also with an alternative exact algorithm for computing the TBR distance exactly. The result is our package Tubro, which incorporates all seven reductions and can, if desired, output the reduced trees to be solved by another package. However, it also incorporates its own exact algorithm, based on augmenting a core Integer Linear Programming (ILP) formulation with certain additional features. Before we present the ILP formulation, and describe the various augmentations as well as Tubro's strengths and weaknesses, we need a couple new definitions.

Let $T$ be a phylogenetic tree on $X$. A {\it quartet} is a phylogenetic tree with exactly four leaves. For example, if $\{a,b,c,d\}\subseteq X$ , we say that $ab|cd$ is a quartet of $T$ if  the path from $a$ to $b$ does not intersect the path from $c$ to $d$ in $T$. Note that, if $ab|cd$ is not a quartet of $T$, then either $ac|bd$ or $ad|bc$ is a quartet of $T$.  If $ab|cd$ is a quartet of $T$, we say that $T$ {\it displays} $ab|cd$. As a consequence, $T$ displays exactly $\binom{n}{4}$ quartets.
\\

\subsection{A polynomial-size Integer Linear Programming formulation for {\MAF}}

The ILP formulation we use is adapted from \cite{MILPsprWu}. Both our formulation and the formulation in~\cite{MILPsprWu} are based on the idea of cutting a minimum number of edges in one of the trees, such that an agreement forest is obtained. \purple{To this end,} decision variables represent which edges are cut. The formulation in \cite{MILPsprWu} works on \emph{rooted} binary phylogenetic trees, leading to two different types of constraints. Here we are working with \emph{unrooted} binary phylogenetic trees. As we shall see this allows us to obtain a simplified ILP in which there is only one type of constraint. We start by describing the formulation and proving its correctness.

Let $T$ and $T'$ be the two phylogenetic trees on $X$ with $n=|X|$. We select arbitrarily one of the trees; here we take $T$. Let $E(T)$ be the set of edges in $T$. For two taxa $x,y \in X$, let $P_T(xy)$ be the edges on the unique path from $x$ to $y$ in $T$.
%Consider a quartet $ab|cd$ that is displayed by $T$ but not by $T'$; we call such a quartet an \emph{incompatible quartet}.
Furthermore, let $Q$ be the set consisting of those quartets which are displayed by $T$ but not by $T'$.  For each $e \in E(T)$ the ILP contains a binary decision variable $x_e$, so $2n-3$ such variables in total. There is one constraint for each quartet in $Q$. The number of constraints thus depends on how different $T$ and $T'$ are, but is in any case $O(n^4)$. The ILP formulation has a Hitting Set flavor:

\begin{equation*}
\begin{array}{l@{}ll}
\text{minimize}  & \hspace{0.2cm} \displaystyle\sum\limits_{e \in E(T)}x_{e} &\\
%\text{subject to}& \displaystyle\sum\limits_{ e \in \substack{P_T(ab) \cup \\ P_T(cd)} } \hspace{-0.5cm}  x_{e} \geq 1  & \hspace{1cm}\text{for each } ab|cd \in Q\\
\text{subject to}& \displaystyle\sum\limits_{ e \in P_T(ab) \cup P_T(cd) } \hspace{-0.5cm}  x_{e} \geq 1  & \hspace{1cm}\text{for each } ab|cd \in Q\\
\text{and}                 &                                                x_{e} \in \{0,1\} &\hspace{1cm}\text{for each }e \in E(T)
\end{array}
\end{equation*}
\noindent
Rather than proving (only) that the optima of the ILP coincides with $d_\MAF(T,T')$, we prove the following slightly more general statement, which allows also non-optimal solutions to the ILP to be mapped to agreement forests. This \purple{is crucial} to extract valid agreement forests even if it takes too long for the ILP solver to reach optimality.

\begin{theorem}
\label{thm:ILPcorrect}
Let $T$ and $T'$ be two phylogenetic trees on $X$. An agreement forest $F$ for $T$ and $T'$ with $|F|=t$ induces a feasible solution to the ILP with objective function value $t-1$. Furthermore, a feasible solution to the ILP with objective function value $t'$ induces an agreement forest $F$ of $T$ and $T'$ with $|F|\leq t'+1$ components.
\end{theorem}
\begin{proof}
Let $F$ be an agreement forest of $T$ and $T'$ containing $t$ components. By the definition of an agreement forest, there exists a subset $E_F$ of $E(T)$ with  $|E_F|=t-1$ such that deleting exactly the edges of $E_F$  in $T$ results in a graph that contains $t$ connected components and the partition of $X$ induced by the taxa in these components is exactly equal to $F$.
%This fact follows from the definition of an agreement forest. Let $E_F$ be this (not necessarily unique) set of $t-1$ edges.
Note that $E_F$ is not necessarily unique. We argue that setting the decision variables corresponding to $E_F$ to 1, and all other decision variables to 0, yields a feasible solution to the ILP. Towards a contradiction, assume that this is not the case. Then there is a quartet $ab|cd$ such that $T$ displays $ab|cd$, $T'$ does not display  $ab|cd$, and none of the decision variables corresponding to edges in $P(ab) \cup P(cd)$ are set to 1. Observe that, for each component $B \in F$, $\{a,b,c,d\} \not \subseteq B$, because otherwise $T|B\ne T'|B$ contradicting that $F$ is an agreement forest for $T$ and $T'$.
%both $T$ and $T'$ would display $ab|cd$.
So $\{a,b,c,d\}$ intersects at least two components of $F$. If $\{a,b\} \subseteq B$ and $\{c,d\} \subseteq B'$, where $B \neq B'$, then because $T[B]$ and $T[B']$ (resp. $T'[B]$ and $T'[B']$) are vertex-disjoint in $T$ (resp. $T'$) both $T$ and $T'$ display $ab|cd$; again a contradiction. In fact, the only two possibilities are (1) that each taxon in $\{a,b,c,d\}$ intersects a different component of $F$, and (2) three of $\{a,b,c,d\}$ occur in a component $B$ and the remaining taxon in $B'$, where $B' \neq B$. Hence, regardless which of (1) and (2) applies, $a$ and $b$ occur in different components of $F$ and/or $c$ and $d$ occur in different components of $F$. Suppose without loss of generality that $a$ and $b$ occur in different components of $F$. Then $P(ab) \cap E_F \neq \emptyset$ which implies that there exists an edge in $P(ab) $ whose corresponding decision variable is set to 1;
%,
yielding a final contradiction.

For the second statement, assume that we have a feasible solution to the ILP with objective function value $t'$. Let $E_I$ be the edges of $T$ corresponding to decision variables that have been set to 1 in this solution. Let $P$ be the partition of $X$ induced by the connected components of $E(T) \setminus E_I$ after deleting any connected components that do not contain any taxa, if they exist\footnote{An optimal solution to the \purple{ILP  never induces} such taxa-free connected components, but a sub-optimal solution might.}. Clearly, $|P| \leq t'+1$, since the deletion of an edge increases the number of connected components by at most one. We claim that $P$ is in fact an agreement forest of $T$ and $T'$. Once again towards a contradiction, assume this is not the case. Suppose that condition (1) in the definition of an agreement forest is violated. Then, there exists a block $B \in P$ such that $|B| \geq 4$ and $T|B \neq T'|B$. It follows that there exist 4 taxa $\{a,b,c,d\} \subseteq B$ such that, without loss of generality, $ab|cd$ is displayed by $T|B$ (and hence by $T$) but not by $T'|B$ (and hence not by $T'$). But none of the edges on $T[B]$ are in $E_I$. In particular, \purple{none of the edges on the path from $a$ to $b$ and  none of the edges on the path from $c$ to $d$ in $T$ have been cut}, contradicting the feasibility of the ILP solution. Now, suppose that condition (2) is violated. Recall that, in $T$, the images of $B$ and $B'$ do not intersect, by construction. Thus there exist two distinct blocks $B, B' \in P$ such that the two embeddings $T'[B]$ and $T'[B']$ are not vertex-disjoint in $T'$.  In fact, they are not edge-disjoint, due to the fact that $T'$ is binary.
%trees have maximum degree 3.
In turn, this implies that $|B|, |B'| \geq 2$.  Let $e = \{u,v\}$ be any edge in $T'$ that is shared by $T'[B]$ and $T'[B']$. Deleting $e$ naturally induces a partition of $B$ into $B_u$ and $B_v$, where $B_u$ (resp. $B_v$) are those taxa in $B$ that are closer in $T'$ to $u$ than $v$ (resp. closer in $T'$ to $v$ than $u$). \purple{Observe that each of $B_u$ and $B_v$ contains at least one taxon.}
%Due to the minimality of $T'[B]$, note that $B_u$ and $B_v$ are both non empty.
Let $B'_u$ and $B'_v$ be defined analogously \purple{with respect to $B'$}. Furthermore, let $a \in B_u$, $b \in B_v$, $c \in B'_u$ and $d \in B'_v$. Since $T[B]$ and  $T[B']$ do not intersect in $T$, and $\{a,b\} \subseteq B$, and $\{c,d\} \subseteq B'$, it follows that $T$ displays $ab|cd$.  However, $T'$ displays $ac|bd$, because there is a path from $a$ to $c$ passing through $u$ and a path from $b$ to $d$ passing through $v$ and these two paths are disjoint. Hence, the ILP included a constraint to cut at least one edge on the paths in $T$ from $a$ to $b$ and from $c$ to $d$.  But no such edge was cut, because $\{a,b\} \subseteq B$ and $\{c,d\} \subseteq B'$; a contradiction. $\qed$
\end{proof}

\purple{The next corollary is an immediate consequence of Theorem~\ref{thm:ILPcorrect}.}

\begin{corollary}
The optimum value computed by the ILP equals $d_\TBR(T,T') = d_\MAF(T,T') - 1$.
\end{corollary}

Note that, if it is known a priori that there exists a maximum agreement forest in which a certain edge $e$ of $T$ is \emph{not}
cut, this can easily be enforced by adding the constraint $x_e = 0$ (or better by removing $x_e$ from the set of decision variables). As we shall see in due course, Tubro can make good use of this, for example to stipulate that it is unnecessary to cut edges in preserved common chains even if the chains themselves cannot be further reduced. More complex restrictions on feasible solutions, such as ``at least one of edge $e$ and $e'$ must be cut'', can easily be added using standard ILP modeling techniques.

\subsection{High-level description of Tubro}
\label{subsec:def_tubro2}

Tubro is highly configurable, with many options that can be switched on or off. \purple{The} core functionality can be summarized as follows.

\begin{enumerate}
\item First, the instance is kernelized by applying all seven reductions.
\item Next, the ILP formulation described in the previous section is generated.
\item \label{stepstep}The ILP formulation is simplified using \emph{chain preservation} (see below).
\item \label{step} A simple greedy algorithm is applied to produce an agreement forest that is not-necessarily maximum, which gives an upper bound on $d_\TBR$. This is the classical greedy algorithm for Hitting Set instances~\cite{chvatal1979greedy}: select a decision variable not yet set to 1 which intersects with a maximum number of unsatisfied constraints (cq. resolves a maximum number of as yet unresolved quartet conflicts which are quartets displayed by only one of $T$ and $T'$).
\item The ILP is fed to the ILP solver Gurobi~\cite{gurobi}.  In the experiments we used Gurobi version 8.1.1. and warm-started it with the upper bound from the greedy algorithm described in Step~\ref{step}.
\item While Gurobi is solving, the $d_\MP$ lower bound sampling strategy algorithm is run in parallel and every time an improved lower bound is found, this is communicated to Gurobi.
\end{enumerate}
The chain preservation  mentioned above in Step~\ref{stepstep} leverages \cite[Theorem 5]{kelk2020new}, which is as follows:

\begin{theorem}
\label{thm:allchainsintact}
Let $T$ and $T'$ be two phylogenetic trees on $X$.
Let $K$ be an (arbitrary) set of mutually taxa-disjoint chains that are common to $T$ and $T'$.
Then there exists a maximum agreement forest $F$ of $T$ and $T'$ such that
\begin{enumerate}
\item every $n$-chain in $K$ with $n\geq 3$
is preserved in $F$, and
\item every 2-chain in $K$
that is pendant in at least one of $T$ and $T'$
is preserved in $F$.
\end{enumerate}
\end{theorem}

In \cite{kelk2020new} the theorem was used as part of a more comprehensive argument to prove the correctness of several of the new reduction rules introduced there. Interestingly, even after all reduction rules have been applied to exhaustion, there can exist common chains that, although they cannot be reduced further in length, still obey
the above theorem. As a result, it is safe to assume the existence of a maximum agreement forest in which no such chain is split across two or more components of this forest. Translated into the language of ILP, this means that for each edge $e$ contained in such a chain, we can \emph{a priori} set $x_e = 0$ in the ILP.  These variables can thus be removed from the ILP, simplifying the system. Tubro uses a simple greedy strategy to select a set of mutually taxa-disjoint chains with maximum total length, thus optimizing the number of variables that can be removed. In our analysis of the main dataset, chain preservation was switched \emph{on}.
%\\
%In our main experiment, chain protection was swithced \emph{on}.

Although a full explanation is beyond the scope of this paper, we observed that chain preservation does remove quite a lot of variables, \emph{even in fully reduced instances}.  To illustrate this, we performed a secondary experiment on all fully-reduced instances from the main dataset where the original trees had 50, 100, or 150 taxa. This was in total 315 tree pairs. We observed that on average the number of variables in the ILP was reduced by 26.29\% (standard deviation of 8.77) when chain preservation was switched on, compared to when it was switched off. This suggests that a significant number of common chains survive, that cannot be further reduced by the new reduction rules, but which \emph{do} fall under the chain preservation theorem. Indeed, the counting argument used in \cite{kelk2020new} to obtain the $11k-9$ upper bound on kernel size is based on the possibility of $O(k)$ irreducible 3-chains or 2-chains existing. These chains cannot be reduced by the existing reduction rules, but (for those that fall under the chain preservation theorem for a given set of mutually taxa-disjoint chains) we can still  algorithmically exploit the fact that they are preserved in some maximum agreement forest, by stipulating that the edges in the chains should not be cut.

% use ILP to solve the reduced instance. Another TBR solver already exists as a subfunction of the \textsc{uSPR} package \cite{whidden2018calculating}, which we will return to later, but here we chose for ILP to better understand the impact of kernelization in its own right: the TBR solver inside \textsc{uSPR} is an FPT branching algorithm, so is only moderately influenced by changes in the size of the input. Plus, this reflects the situation, often encountered in practice, when no sophisticated exact algorithm is known to solve the kernel.

\subsection{Optional extra: cluster reduction}
Let $T$ and $T'$ be two phylogenetic trees on $X$. For $Y\subset X$ with $|Y|\geq 2$ and $|X-Y|\geq 2$, we say that $Y$ is a {\it cluster} of $T$ if there exists a single edge in $T$ whose deletion disconnects $T$ into two parts such that the leaves of one part are bijectively labeled by elements in $Y$ while the leaves of the other part are bijectively labeled by elements in $X-Y$.
%If $|Y|$=1, then the cluster is called {\it trivial} and, otherwise, it is called {\it non-trivial}. Note that, if $Y$ is a cluster of $T$, then $X-Y$ is also a cluster of $N$.
%We say that $Y, X-Y$ is a {\it bipartition} of $T$.
Now, let $Y$ be a cluster of both $T$ and $T'$; this is often referred to as a \emph{common cluster}. In \cite{bordewich2017fixed} a divide-and-conquer reduction rule is presented which,  at a high level, computes $d_{\MAF}(T, T')$ by computing $d_{\MAF}(T|Y, T'|Y)$ and $d_{\MAF}(T|X-Y, T'|X-Y)$ separately. It can be shown that $d_{\MAF}(T,T')$ is either equal to,
\begin{equation}
d_{\MAF}(T|Y, T'|Y) + d_{\MAF}(T|X-Y, T'|X-Y)
\end{equation}
or
\begin{equation}
d_{\MAF}(T|Y, T'|Y) + d_{\MAF}(T|X-Y, T'|X-Y) - 1.
\end{equation}
% and adding them together\footnote{The original exposition by \cite{bordewich2017fixed} is in terms of $d_{\TBR}$ rather than $d_{\MAF}$.}.
Deciding  whether (1) or (2) holds, first requires the addition of a ``placeholder'' taxon $\rho_1$ into \purple{$T|Y$ and $T'|Y$} which represents the location of the $X-Y$ side of the trees. Similarly, a placeholder taxon $\rho_2$ is added to
\purple{$T|X-Y$ and $T'|X-Y$} which represents the location of the $Y$ side of the trees. It is then necessary to query whether there exists a maximum agreement forest for the $\rho_1$-augmented instance (respectively, the $\rho_2$-augmented
instance) that isolates $\rho_1$ (respectively, $\rho_2$) in a singleton component. Depending on the outcome, a distinction can be made between (1) and~(2); for full details see \cite{bordewich2017fixed}. Tubro answers these queries by first solving
two separate ILPs for the $\rho_1$-augmented instance; one in which the edge entering $\rho_1$ is definitely cut, and one where the decision whether to cut this edge is left to the ILP. Symmetrically, it then solves two ILPs for the $\rho_2$-augmented
instance.
Hence, for a \purple{cluster $Y$ that is common to $T$ and $T'$}, four ILPs need to be solved in total. There are a number of additional subtle technicalities concerning the interaction of the placeholder taxon and the reduction rules from \cite{kelk2020new} but these are out of scope of the current article.

The decomposition of one ILP into four ILPs sounds rather cumbersome, but given that the bottleneck for \purple{Tubro's} performance is the number of taxa in an instance - recall that the ILP has size $O(|X|^4)$ - the cluster reduction can potentially cause a significant speedup
%if $|Y|$ and $|X-Y|$ are both smaller than $|X|$, particularly
%%SL. commented out because it's always the case that $|Y|$ and $|X-Y|$ are both smaller than $|X|$
\purple{if $|Y|$ and $|X-Y|$ are roughly the same, and/or the reduction activates multiple times.}
%\marginpar{Commented something out here.}

We did not switch \purple{on} the cluster reduction when applying Tubro to the 89 instances that uSPR could not solve (post-kernelization) after 5 minutes. However, we did briefly check whether the cluster reduction might have helped
with these 89 instances. We observed that 86 of the instances had no \purple{common cluster}.
%SL. By definition, a cluster is not trivial.
\purple{The} 3 instances that did have \purple{a common cluster}  had \purple{parameter combinations} $(t,s,k)$  of $(250,50,35)$, $(300,50,35)$ and $(200,70,35)$. \purple{For these tree pairs, we observed} that the clusters were heavily imbalanced, meaning that $|Y|$ was small and $|X-Y|$ was large.  This limited the practical impact of the cluster reduction since, for the larger \purple{tree pair on $X-Y$}, it was still necessary to construct a prohibitively large ILP.\\

%\subsection{The unreasonable effectiveness of the dMP sampling lower bounds}

%Possibly include ``big tree'' data here too.\\

\bibliography{article}{}

\begin{thebibliography}{10}

\bibitem{alber2006experiments}
J.~Alber, N.~Betzler, and R.~Niedermeier.
\newblock Experiments on data reduction for optimal domination in networks.
\newblock {\em Annals of Operations Research}, 146(1):105--117, 2006.

\bibitem{AllenSteel2001}
B.~Allen and M.~Steel.
\newblock Subtree transfer operations and their induced metrics on evolutionary
  trees.
\newblock {\em Annals of Combinatorics}, 5:1--15, 2001.

\bibitem{extremal2019}
R.~Atkins and C.~McDiarmid.
\newblock Extremal distances for subtree transfer operations in binary trees.
\newblock {\em Annals of Combinatorics}, 23(1):1--26, 2019.

\bibitem{bordewich2017fixed}
M.~Bordewich, C.~Scornavacca, N.~Tokac, and M.~Weller.
\newblock On the fixed parameter tractability of agreement-based phylogenetic
  distances.
\newblock {\em Journal of Mathematical Biology}, 74(1-2):239--257, 2017.

\bibitem{chen2015parameterized}
J.~Chen, J.-H. Fan, and S.-H. Sze.
\newblock Parameterized and approximation algorithms for maximum agreement
  forest in multifurcating trees.
\newblock {\em Theoretical Computer Science}, 562:496--512, 2015.

\bibitem{chvatal1979greedy}
V.~Chvatal.
\newblock A greedy heuristic for the set-covering problem.
\newblock {\em Mathematics of Operations Research}, 4(3):233--235, 1979.

\bibitem{Cygan:2015:PA:2815661}
M.~Cygan, F.~Fomin, L.~Kowalik, D.~Lokshtanov, D.~Marx, M.~Pilipczuk,
  M.~Pilipczuk, and S.~Saurabh.
\newblock {\em Parameterized Algorithms}.
\newblock Springer Publishing Company, Incorporated, 1st edition, 2015.

\bibitem{downey2013fundamentals}
R.~Downey and M.~Fellows.
\newblock {\em Fundamentals of parameterized complexity}, volume~4.
\newblock Springer, 2013.

\bibitem{fellows2018known}
M.~Fellows, L.~Jaffke, Aliz~Izabella Kir{\'a}ly, F.~Rosamond, and M.~Weller.
\newblock What is known about vertex cover kernelization?
\newblock In {\em Adventures Between Lower Bounds and Higher Altitudes}, pages
  330--356. Springer, 2018.

\bibitem{felsenstein2004inferring}
J.~Felsenstein.
\newblock {\em Inferring Phylogenies}.
\newblock Sinauer Associates, Incorporated, 2004.

\bibitem{ferizovic2020engineering}
D.~Ferizovic, D.~Hespe, S.~Lamm, M.~Mnich, C.~Schulz, and D.~Strash.
\newblock Engineering kernelization for maximum cut.
\newblock In {\em 2020 Proceedings of the Twenty-Second Workshop on Algorithm
  Engineering and Experiments (ALENEX)}, pages 27--41. SIAM, 2020.

\bibitem{fischer2014}
M.~Fischer and S.~Kelk.
\newblock On the {M}aximum {P}arsimony distance between phylogenetic trees.
\newblock {\em Annals of Combinatorics}, 20(1):87--113, 2016.

\bibitem{fitch1971}
W.~M. Fitch.
\newblock Toward defining the course of evolution: minimum change for a
  specific tree topology.
\newblock {\em Systematic Biology}, 20(4):406--416, 1971.

\bibitem{kernelization2019}
F.~Fomin, D.~Lokshtanov, S.~Saurabh, and M.~Zehavi.
\newblock {\em Kernelization: Theory of Parameterized Preprocessing}.
\newblock Cambridge University Press, 2019.

\bibitem{gurobi}
LLC Gurobi~Optimization.
\newblock Gurobi optimizer reference manual, 2020.

\bibitem{harding1971probabilities}
E.F. Harding.
\newblock The probabilities of rooted tree-shapes generated by random
  bifurcation.
\newblock {\em Advances in Applied Probability}, pages 44--77, 1971.

\bibitem{hein1996complexity}
J.~Hein, T.~Jiang, L.~Wang, and K.~Zhang.
\newblock On the complexity of comparing evolutionary trees.
\newblock {\em Discrete Applied Mathematics}, 71(1-3):153--169, 1996.

\bibitem{henzinger2020shared}
M.~Henzinger, A.~Noe, and C.~Schulz.
\newblock Shared-memory branch-and-reduce for multiterminal cuts.
\newblock In {\em 2020 Proceedings of the Twenty-Second Workshop on Algorithm
  Engineering and Experiments (ALENEX)}, pages 42--55. SIAM, 2020.

\bibitem{hickey2008spr}
Glenn Hickey, Frank Dehne, Andrew Rau-Chaplin, and Christian Blouin.
\newblock Spr distance computation for unrooted trees.
\newblock {\em Evolutionary Bioinformatics}, 4:EBO--S419, 2008.

\bibitem{HusonRuppScornavacca10}
D.~Huson, R.~Rupp, and C.~Scornavacca.
\newblock {\em Phylogenetic Networks: Concepts, Algorithms and Applications}.
\newblock Cambridge University Press, 2011.

\bibitem{john2017shape}
K.~St John.
\newblock The shape of phylogenetic treespace.
\newblock {\em Systematic Biology}, 66(1):e83, 2017.

\bibitem{kelk2017complexity}
S.~Kelk and M.~Fischer.
\newblock On the complexity of computing mp distance between binary
  phylogenetic trees.
\newblock {\em Annals of Combinatorics}, 21(4):573--604, 2017.

\bibitem{kelk2016reduction}
S.~Kelk, M.~Fischer, V.~Moulton, and T.~Wu.
\newblock Reduction rules for the maximum parsimony distance on phylogenetic
  trees.
\newblock {\em Theoretical Computer Science}, 646:1--15, 2016.

\bibitem{tightkernel}
S.~Kelk and S.~Linz.
\newblock A tight kernel for computing the tree bisection and reconnection
  distance between two phylogenetic trees.
\newblock {\em SIAM Journal on Discrete Mathematics}, 33(3):1556--1574, 2019.

\bibitem{kelk2020new}
S.~Kelk and S.~Linz.
\newblock New reduction rules for the tree bisection and reconnection distance.
\newblock {\em Annals of Combinatorics}, pages 1--28, 2020.

\bibitem{kelk2017note}
S.~Kelk and G.~Stamoulis.
\newblock A note on convex characters, fibonacci numbers and exponential-time
  algorithms.
\newblock {\em Advances in Applied Mathematics}, 84:34--46, 2017.

\bibitem{Kuhner2015}
M.~Kuhner and J.~Yamato.
\newblock Practical performance of tree comparison metrics.
\newblock {\em Systematic Biology}, 64(2):205--214, 2015.

\bibitem{li2017computing}
Z.~Li and N.~Zeh.
\newblock Computing maximum agreement forests without cluster partitioning is
  folly.
\newblock In {\em 25th Annual European Symposium on Algorithms (ESA 2017)}.
  Schloss Dagstuhl-Leibniz-Zentrum fuer Informatik, 2017.

\bibitem{mertzios2020power}
G.~Mertzios, A.~Nichterlein, and R.~Niedermeier.
\newblock The power of linear-time data reduction for maximum matching.
\newblock {\em Algorithmica}, 82:3521--3565, 2020.

\bibitem{moulton2015}
V.~Moulton and T.~Wu.
\newblock A parsimony-based metric for phylogenetic trees.
\newblock {\em Advances in Applied Mathematics}, 66:22--45, 2015.

\bibitem{richards2018variation}
E.~Richards, J.~Brown, A.~Barley, R.~Chong, and R.~Thomson.
\newblock Variation across mitochondrial gene trees provides evidence for
  systematic error: How much gene tree variation is biological?
\newblock {\em Systematic Biology}, 67(5):847--860, 2018.

\bibitem{SempleSteel2003}
C.~Semple and M.~Steel.
\newblock {\em Phylogenetics}.
\newblock Oxford University Press, 2003.

\bibitem{vanIersel20161075}
L.~van Iersel, S.~Kelk, and C.~Scornavacca.
\newblock Kernelizations for the hybridization number problem on multiple
  nonbinary trees.
\newblock {\em Journal of Computer and System Sciences}, 82(6):1075--1089,
  2016.

\bibitem{whidden2013fixed}
C.~Whidden, R.~G. Beiko, and N.~Zeh.
\newblock Fixed-parameter algorithms for maximum agreement forests.
\newblock {\em SIAM Journal on Computing}, 42(4):1431--1466, 2013.

\bibitem{whidden2018calculating}
C.~Whidden and F.~Matsen.
\newblock Calculating the unrooted subtree prune-and-regraft distance.
\newblock {\em IEEE/ACM Transactions on Computational Biology and
  Bioinformatics}, 2018.

\bibitem{Whidden2014}
C.~Whidden, N.~Zeh, and R.~G. Beiko.
\newblock Supertrees based on the subtree prune-and-regraft distance.
\newblock {\em Systematic Biology}, 63(4):566--581, 2014.

\bibitem{MILPsprWu}
Y.~Wu.
\newblock A practical method for exact computation of subtree prune and regraft
  distance.
\newblock {\em Bioinformatics}, 25(2):190--196, 2009.

\bibitem{yoshida2019multilocus}
R.~Yoshida, K.~Fukumizu, and C.~Vogiatzis.
\newblock Multilocus phylogenetic analysis with gene tree clustering.
\newblock {\em Annals of Operations Research}, 276(1-2):293--313, 2019.

\end{thebibliography}
\bibliographystyle{plain}

\end{document}